\newcommand{\shortversion}[1]{}
\newcommand{\longversion}[1]{#1}

\newcommand{\papertitle}{Vertex Cover Gets Faster and Harder on Low Degree Graphs}

\documentclass[svgnames]{llncs}

\usepackage{etex}
\usepackage[usenames,svgnames,dvipsnames]{xcolor}
\usepackage{graphicx}
\usepackage{pstricks}
\usepackage{tikz}
\usetikzlibrary{decorations.shapes,arrows}

\usepackage{wrapfig}
\usepackage{xspace}
\usepackage[ruled,vlined,linesnumbered]{algorithm2e}
\usepackage{multicol}
\usepackage{multirow}
\longversion{\usepackage{fullpage}}

\usepackage{framed}
\usepackage{boxedminipage}

\usepackage{amsmath,amssymb}
\usepackage[retainorgcmds]{IEEEtrantools}
\usepackage{pdfpages}
\usepackage{lastpage}

\usepackage{authblk}

\usepackage{cleveref}

\usepackage[backend=bibtex]{biblatex}
\bibliography{lowdegreevc}

\AtEveryBibitem{
 \clearlist{address}
 \clearfield{date}
 \clearfield{eprint}
 \clearfield{isbn}
 \clearfield{issn}
 \clearlist{location}
 \clearfield{month}
 \clearfield{series}
 \clearfield{url} 
 \ifentrytype{book}{}{
  \clearlist{publisher}
  \clearname{editor}
  \clearfield{doi}
 }
}

\usepackage{hyperref}

\hypersetup{
    pdfauthor={Neeldhara Misra},     
    pdfnewwindow=true,      
    colorlinks=true,       
    linkcolor=SlateBlue,          
    citecolor=FireBrick,        
    filecolor=magenta,      
    urlcolor=Emerald           
}

\usepackage{url}

\usepackage[backgroundcolor=Moccasin,colorinlistoftodos]{todonotes}

\makeatletter
 \providecommand\@dotsep{5}
 \def\listtodoname{}
 \def\listoftodos{\@starttoc{tdo}\listtodoname}
 \makeatother

\newcounter{nmcomment}

\newcounter{vkcomment}

\newcommand{\NPC}{\textrm{\textup{NP-complete}}\xspace}

\newcommand{\name}[1]{\textsc{#1}}

\usepackage{savesym}
\savesymbol{labelindent}
\usepackage{enumitem}
\usepackage{mdwlist}

\newcounter{descriptcount}
\newlist{enumdescript}{description}{2}

\setlist[enumdescript,1]{%
  before={\setcounter{descriptcount}{0}%
          \renewcommand*\thedescriptcount{\Roman{descriptcount}}}
  ,font=\bfseries\stepcounter{descriptcount}\thedescriptcount~
}
\setlist[enumdescript,2]{%
  before={\setcounter{descriptcount}{0}%
          \renewcommand*\thedescriptcount{\alph{descriptcount}}}
  ,font=\bfseries\stepcounter{descriptcount}\thedescriptcount~
}

\newcommand{\OO}{{\mathcal O}}

\newcommand{\branchvector}[1]{{\color{IndianRed}{$(#1)$}}}
\newcommand{\braid}{braid}
\newcommand{\Braid}{Braid}
\newcommand{\runtime}{\ensuremath{1.2637}}

\title{\papertitle}
\author{Akanksha Agrawal\inst{1}, Sathish Govindarajan\inst{1}, Neeldhara Misra\inst{1}}
\institute{Indian Institute of Science, Bangalore\\
\textit {\{akanksha.agrawal$\lvert$gsat$\rvert$neeldhara\}@csa.iisc.ernet.in}}

\begin{document}

\maketitle

\begin{abstract}

The problem of finding an optimal vertex cover in a graph is a classic NP-complete problem, and is a special case of the hitting set question. On the other hand, 
the hitting set problem, when asked in the context of induced geometric objects, often turns out to be exactly the vertex cover problem on restricted classes of graphs. In this work we explore a particular instance of 
such a phenomenon. We consider the problem of hitting all axis-parallel slabs induced by a point set $P$, and show that it is equivalent to the problem of finding a vertex cover on a graph whose edge 
set is the union of two Hamiltonian Paths. We show the latter problem to be \NPC{}, and we also give an algorithm to find a vertex cover of size at most $k$, on graphs of maximum degree four, whose running time is $\runtime{}^k n^{O(1)}$.

\end{abstract}

\section{Introduction}

Let $P$ be a set of $n$ points in $\mathbb{R}^2$ and let $\mathcal{R}$ be the family of all distinct objects of a particular kind
(disks, rectangles, triangles, \dots), such that each object in $\mathcal{R}$ has a distinct tuple of points from $P$ on
its boundary. For example, $\mathcal{R}$ could be the family of $n \choose 2$ axis parallel rectangles such that each rectangle has a
distinct pair of points of $P$ as its diagonal corners. $\mathcal{R}$ is called the set of all objects induced (spanned) by $P$.

Various questions related to geometric objects induced by a point set have been studied in the last few decades.
A classical result in discrete geometry is the {\em First Selection Lemma}~\cite{BF84} which shows that there exists a point that is present in a constant fraction of triangles induced by $P$. Another interesting question is to compute the minimum set of points in $P$ that ``hits'' all
the induced objects in $\mathcal{R}$. This is a special case of the classical Hitting Set problem, which we will refer to as {\em Hitting Set for Induced Objects}.

For most geometric objects, it is not known if the {\em Hitting Set for induced objects} problem
is polynomially solvable. It is known to be polynomial solvable for skyline rectangles and halfspaces.
Recently, Rajgopal et al~\cite{RAGKM13} showed that this problem is NP-complete for lines.

The problem of finding an optimal vertex cover in a graph is a classic NP-complete problem, and is a special case of the Hitting Set problem. 
On the other hand, the hitting set for induced objects problem often turns out to be exactly the vertex cover problem, even on
restricted classes of graphs. For example, the problem of hitting set for induced axis-parallel rectangles is equivalent to the vertex cover on the Delaunay graph of the point set with respect to axis-parallel rectangles.

We study a particular phenomenon of this type, where the hitting set question in the geometric setting boils down to a vertex cover problem on a structured graph
class. We consider the problem of hitting set for induced axis-parallel slabs (rectangles whose horizontal or vertical sides are unbounded). Note that
this is even more structured than general axis-parallel rectangles, and indeed, it turns out that the corresponding Delaunay graph has a very special property --- its edge set
is the union of two Hamiltonian paths. Since any hitting set for the class of axis-parallel slabs induced by a point set $P$ is exactly the vertex cover of the Delaunay graph
with respect to axis-parallel slabs for $P$, our problem reduces to solving vertex cover on the class of graphs whose edge set is simply the union of two Hamiltonian Paths.

Despite the appealing structure, we show that -- surprisingly -- deciding  $k$-vertex cover on this class of graphs is \NPC{}. This involves a rather intricate reduction from the problem of
finding a vertex cover on cubic graphs. We also appeal to the fact that the edge set of four-regular graphs can be partitioned into two two-factors, and the main challenge in the
reduction involves stitching the components of the two-factors into two Hamiltonian paths while preserving the size of the vertex cover in an appropriate manner.

Having established the NP-hardness of the problem, we pursue the question of improved fixed-parameter algorithms on this special case. Vertex Cover is one of the most well-studied
problems in the context of fixed-parameter algorithm design, it enjoys a long list of improvements even on special graph classes. We note that for {\sc Vertex Cover}, the
goal is to find a vertex cover of size at most $k$ in time $\OO(c^k)$, and the ``race'' involves exploring algorithms that reduce the value of the best known constant $c$.

In particular, even for sub-cubic graphs (where the maximum degree is at most three, and the problem remains \NPC{}), Xiao~\cite{X10} proposed an
algorithm with running time $\OO^\star(1.1616^k)$, improving on the previous best record 
~\cite{Chen:2005:LST:1108358.1108360} of $\OO^\star(1.1940^k)$ by Chen, Kanj and Xia, and prior
to this, Razgon~\cite{R09} had a $\OO^\star(1.1864^k)$. The best-known algorithm for Vertex Cover~\cite{CKX10} on general graphs has a running time of $\OO(1.2738^k + kn)$ and uses polynomial-space.

Typically, these algorithms involve extensive case analysis on a cleverly designed search tree. In the second part of this work, we propose a branching algorithm with running time $O^\star(\runtime{}^k)$ for graphs with maximum degree bounded by at most four. This improves the best known algorithm for this class, which surprisingly has been no better than the algorithm for general graphs. We note that this implies faster algorithms for the case of graphs that can be decomposed into the union of two Hamiltonian Paths (since they have maximum degree at most four), however, whether they admit additional structure that can be exploited for even better algorithms remains an open direction.


\section{Preliminaries}

In this section, we state some basic definitions and introduce terminology from graph theory and algorithms. We also establish some of the notation that will be used throughout. 

We denote the set of natural numbers by $\mathbb{N}$ and set of real numbers by $\mathbb{R}$ . 
For a natural number $n$, we use $[n]$ to denote the set $\{1,2,\ldots,n\}$. For a finite set $A$ we denote by \(\mathfrak{S}_A\) the set of all permutations of the elements of set \(A\). To describe the running times of our algorithms, we will use the $O^*$ notation. Given $f: \mathbb{N} \rightarrow \mathbb{N}$, we define $O^{*}(f(n))$ to be $O(f(n) \cdot p(n))$, where $p(\cdot)$ is some polynomial
function. That is, the $O^{*}$ notation suppresses polynomial
factors in the running-time expression. 


\paragraph{Graphs.} In the following, let $G=\left(V,E\right)$ be a graph. For any non-empty subset $W\subseteq V$, the subgraph of $G$ induced
by $W$ is denoted by $G[W]$; its vertex set is $W$ and its edge set consists of all those edges of $E$ with both endpoints in $W$.
For $W\subseteq V$, by $G\setminus W$ we denote the graph obtained by deleting the vertices in $W$ and all edges which are incident to at least one vertex in $W$.

A \emph{vertex cover} is a subset of vertices $S$ such that $G \setminus S$ has no edges. We denote a vertex cover of size at most $k$ of a graph 
$G$ by $k-VC(G)$. For $v \in V$ we denote the open-neighborhood of $v$ by $N(v)=\{u \in V \lvert (u,v)\in E\}$,
closed-neighborhood of $v$ by $N[v]=N(v)\cup \{v\}$, second-open neighborhood by $N_2(v)=\{u \in V \lvert \exists u' \in N(v)$ s.t. $(u,u')\in E\}$
second-closed neighborhood by $N_2[v]=N_2(v)\cup N[v]$.

When we are discussing a pair of vertices $u,v$, then the common neighborhood of $u$ and $v$ is the set of vertices that are adjacent to both $u$ and $v$. 
In this context, a vertex $w$ is called a \emph{private neighbor} of $u$ if $(w,u)$ is an edge and $(w,v)$ is not an edge. We denote the degree of a vertex $v \in V$ by $d(v)$.

A \emph{path} in a graph is a sequence of distinct vertices $v_0, v_1, \ldots, v_k$ such that $(v_i,v_{i+1})$ is an edge for all $0 \leq i \leq (k-1)$.
A \emph{Hamiltonian path} of a graph $G$ is a path featuring every vertex of $G$. The following class of graphs will be of special interest to us.

\begin{definition}[\Braid{} graphs]
A graph $G$ on the vertex set $[n]$ is a \braid{} graph if the edges of the graph can be covered by two Hamiltonian paths. In other words, there exist permutations $\sigma$, $\tau$ of the vertex set for 
which $E(G) = \{(\sigma(i),\sigma(i+1)) ~|~  1 \leq i \leq n-1\} \cup \{(\tau(i),\tau(i+1)) ~|~  1 \leq i \leq n-1\}.$ 
\end{definition}


{\em Induced axis-parallel slabs}: Axis-parallel slabs are a special class
of axis-parallel rectangles where two horizontal or two vertical sides are
unbounded. Each pair of points $p(x_1,y_1)$ and $q(x_2,y_2)$ induces two
axis-parallel slabs of the form $[x_1,x_2]\times (-\infty,+\infty)$ and
$(-\infty,+\infty)\times [y_1,y_2]$. Let $\mathcal{R}$ represent the
family of $2 {n\choose 2}$ axis-parallel slabs induced by $P$.


We refer the reader to \cite{graphbook} for details on standard graph theoretic notation and terminology we use in the paper. 

\paragraph{Parameterized Complexity.} A parameterized problem $\Pi$ is a subset of $\Gamma^{*}\times
\mathbb{N}$, where $\Gamma$ is a finite alphabet. An instance of a
parameterized problem is a tuple $(x,k)$, where $x$ is a classical problem instance, 
and $k$ is called the parameter. A central notion in parameterized complexity is {\em
  fixed-parameter tractability (FPT)} which means, for a given
instance $(x,k)$, decidability in time $f(k)\cdot p(|x|)$, where
$f$ is an arbitrary function of $k$ and $p$ is a polynomial in the
input size. 

%


\section{Hitting Set for Induced Axis-Parallel Slabs}

We show here that the problem of finding a hitting set of size at most $k$ for the family of all axis-parallel slabs induced by a point set is equivalent to the problem of 
finding a vertex cover of a graph whose edges can be partitioned into two Hamiltonian Paths. In subsequent sections, we establish the NP-hardness of the latter problem, and also provide better FPT algorithms. Due the equivalence of these problems, we note that both the hardness and the algorithmic results apply to the problem of finding a hitting set for induced axis parallel slabs. 

\begin{lemma}
An instance of $k$-vertex cover in a \braid{} graph $G=(V,E)$ with permutations $\sigma,\tau \in \mathfrak{S}_V$ can be reduced to the problem of finding a hitting set for the collection of all axis-parallel slabs induced 
by a point set. 
\label{dgapstohp}
d\end{lemma}

\begin{proof}[Sketch] Given an instance of Vertex Cover on a braid graph $G$ with permutations $\sigma$ and $\tau$, we create $n$ points in $\mathbb{R}^2$ in an $(n\times n)$-grid as follows. We assume, by renaming if necessary, that $\sigma$ is the identity permutation. For every $1 \leq i \leq n$, we let $p_i = (i,\tau^{-1}(i))$. Since we only need to hit empty vertical and horizontal slabs, in the induced setting this amounts to hitting all consecutive slabs in the horizontal and vertical directions. It is easy to check that a hitting set for such slabs would exactly correspond to a vertex cover of $G$. \qed
\end{proof}

\begin{lemma}
The problem of finding a hitting set for all induced axis-parallel slabs by a point set $P$ can be reduced to the problem of finding a Vertex Cover in a \braid{} graph. 
\end{lemma}
\begin{proof}From the given point set $P$, we sort the points in $P$ according to their $x$-coordinates to obtain a permutation of the point set $\sigma$. Similarly, we sort with respect to $y$-coordinate to get a permutation 
$\tau$. Note that there exists a empty axis-parallel slab between two points if and only if they are adjacent with respect to at least one of the $x$- or $y$-coordinates, These are, on the other
hand, precisely the edges in the braid graph with $\sigma$ and $\tau$ as the permutations, which shows the equivalence. 
\qed
\end{proof}

\section{NP-completenes of Vertex Cover on \Braid{}s}

In this section, we show that the problem of determining a vertex cover on the class of \braid{}s is hard even when the permutations of the \braid{} are given as input. 

The intuition for the hardness is the following. Consider a four-regular graph. By a theorem of Peterson, we know that the edges of such a graph can be partitioned into two sets, each of which would
be a two-factor in the graph $G$. In other words, every four-regular graph can be thought of as a union of two collections of disjoint cycles, defined on same vertex set. It is conceivable
that these cycles can be patched together into paths, leading us to a braid graph. As it turns out, for such a patching, we need to have some control over the cycles in the decomposition to begin with. 
So we start with an instance of Vertex Cover on a cubic 2-connected planar graphs, morph such an instance to a four-regular graph while keeping track of a special cycle decomposition, which 
we later exploit for the ``stitching'' of cycles into Hamiltonian paths. 

Formally, therefore, the proof is by a reduction from Vertex Cover on a cubic 2-connected planar graph to an instance of $k$-vertex cover on a \braid{} graph, noting that~\cite{Mohar2001102} shows the NP-hardness of 
Vertex Cover for cubic planar 2-connected graphs. We describe the construction in two stages, first showing the transformation to a four-regular graph and then proceeding to illustrate the transformation to a braid graph. 

\shortversion{

Due to space constraints, we only provide the highlights of the reduction. One of the main tasks is to merge the cycles in each decomposition. Let us first illustrate a gadget that combines two cycles into a 
longer one.\footnote{At this point, we are not concerned that this is leading us to, eventually, a Hamiltonian cycle rather than a path, because it is quite easy to convert the former to the latter.} Note that the gadget
itself must be a braid, and of course, we need to ensure equivalence. 

For the purpose of this brief discussion, our starting point is a four-regular graph $G$. Recall that the edge set of $G$ can be decomposed into two collections of cycles. Note that every 
vertex $v$ participates in two cycles, say $C_v$ and $C_v^\prime$ --- these would be cycles from different collections. Now let the neighbors of $v$ in $C_v$ be $v_1,v_2$, and let the neighbors in $C_v^\prime$ be $v_3$ and $v_4$.

\begin{figure}[ht]
\centering
\scalebox{0.45} 
{
\begin{pspicture}(0,-4.6840625)(7.205625,4.6840625)
\psarc[linewidth=0.04](2.35,2.389375){1.55}{137.24574}{58.760784}
\psdots[dotsize=0.3](1.18,3.459375)
\psdots[dotsize=0.3](3.16,3.719375)
\psdots[dotsize=0.3](2.02,4.079375)
\psline[linewidth=0.04cm](1.18,3.439375)(1.98,4.079375)
\psline[linewidth=0.04cm](2.02,4.119375)(3.14,3.759375)
\usefont{T1}{ppl}{m}{n}
\rput(1.9528126,4.424375){\Large $v$}
\usefont{T1}{ppl}{m}{n}
\rput(0.8428125,3.664375){\Large $v_1$}
\usefont{T1}{ppl}{m}{n}
\rput(3.6228125,3.824375){\Large $v_2$}
\usefont{T1}{ppl}{m}{n}
\rput(3.3728125,0.484375){\Large $C_v  \in  \mathfrak{C}_H$}
\psarc[linewidth=0.04](2.25,-2.030625){1.55}{137.24574}{58.760784}
\psdots[dotsize=0.3](1.08,-0.960625)
\psdots[dotsize=0.3](1.92,-0.340625)
\psline[linewidth=0.04cm](1.08,-0.980625)(1.88,-0.340625)
\psline[linewidth=0.04cm](1.92,-0.300625)(3.04,-0.660625)
\usefont{T1}{ppl}{m}{n}
\rput(1.8528125,0.004375){\Large $v$}
\usefont{T1}{ppl}{m}{n}
\rput(0.6828125,-0.775625){\Large $v_3$}
\usefont{T1}{ppl}{m}{n}
\rput(3.5228126,-0.595625){\Large $v_4$}
\usefont{T1}{ppl}{m}{n}
\rput(3.5928125,-3.895625){\Large $C'_v  \in \mathfrak{C}_{H'}$}
\psdots[dotsize=0.3](3.06,-0.700625)
\usefont{T1}{ppl}{m}{n}
\rput(2.2623436,-4.435625){\Large (a)}
\end{pspicture} 
}
\scalebox{0.32} 
{
\begin{pspicture}(0,-6.5620313)(16.489063,6.5620313)
\definecolor{color1801}{rgb}{0.12549019607843137,0.13725490196078433,0.18823529411764706}
\definecolor{color1801d}{rgb}{0.9098039215686274,0.8901960784313725,0.8901960784313725}
\definecolor{color1804}{rgb}{0.054901960784313725,0.07450980392156863,0.12941176470588237}
\definecolor{color1817b}{rgb}{0.996078431372549,0.996078431372549,0.996078431372549}
\definecolor{color1817}{rgb}{0.07450980392156863,0.07450980392156863,0.9176470588235294}
\psline[linewidth=0.04cm](8.301875,5.7501564)(9.321875,4.8101563)
\psline[linewidth=0.04cm](1.8645116,3.716866)(4.101875,2.7501562)
\psline[linewidth=0.04cm](2.641875,5.3901563)(8.281875,5.7701564)
\psbezier[linewidth=0.06,linecolor=color1801,doubleline=true,doublesep=0.08,doublecolor=color1801d](1.581875,-1.8140152)(1.8243687,-2.870711)(2.3975358,-3.7308128)(3.2793314,-4.418894)(4.161127,-5.106975)(4.921875,-5.409844)(5.981875,-5.409844)(7.041875,-5.409844)(7.527849,-5.2612906)(8.481875,-4.6698437)(9.435901,-4.078397)(10.267561,-2.6638515)(10.421875,-1.7698437)
\psbezier[linewidth=0.06,linecolor=color1804,doubleline=true,doublesep=0.08,doublecolor=color1801d](3.501875,-1.875014)(3.637935,-2.5478036)(3.9595308,-3.095423)(4.4542937,-3.5335186)(4.949057,-3.9716144)(5.2460384,-4.0185547)(5.864369,-4.034199)(6.482699,-4.049844)(7.0176253,-3.9552612)(7.546563,-3.5961037)(8.075501,-3.236946)(8.375292,-2.3287559)(8.461875,-1.7498437)
\psline[linewidth=0.04cm,linestyle=dashed,dash=0.16cm 0.16cm](2.501875,-0.6298438)(1.541875,-1.6898438)
\psline[linewidth=0.04cm,linestyle=dashed,dash=0.16cm 0.16cm](2.621875,-0.72984374)(3.501875,-1.6898438)
\usefont{T1}{ppl}{m}{n}
\rput(2.5364063,-1.1398437){\huge $a$}
\usefont{T1}{ppl}{m}{n}
\rput(0.93640625,-1.5598438){\huge $v_1$}
\usefont{T1}{ppl}{m}{n}
\rput(4.1764064,-1.5398438){\huge $v_3$}
\psdots[dotsize=0.4](2.501875,-0.60984373)
\pscircle[linewidth=0.06,linecolor=color1817,dimen=outer,fillstyle=solid,fillcolor=color1817b](1.571875,-1.6998438){0.21}
\pscircle[linewidth=0.06,linecolor=color1817,dimen=outer,fillstyle=solid,fillcolor=color1817b](3.451875,-1.6998438){0.21}
\psline[linewidth=0.04cm,linestyle=dashed,dash=0.16cm 0.16cm](9.341875,-0.7498438)(8.281875,-1.8298438)
\psline[linewidth=0.04cm,linestyle=dashed,dash=0.16cm 0.16cm](9.481875,-0.6898438)(10.461875,-1.5898438)
\usefont{T1}{ppl}{m}{n}
\rput(9.546406,-1.2198437){\huge $b$}
\usefont{T1}{ppl}{m}{n}
\rput(11.176406,-1.6198437){\huge $v_2$}
\usefont{T1}{ppl}{m}{n}
\rput(7.856406,-1.6198437){\huge $v_4$}
\psdots[dotsize=0.4](9.481875,-0.60984373)
\pscircle[linewidth=0.06,linecolor=color1817,dimen=outer,fillstyle=solid,fillcolor=color1817b](8.491875,-1.6998438){0.21}
\pscircle[linewidth=0.06,linecolor=color1817,dimen=outer,fillstyle=solid,fillcolor=color1817b](10.471875,-1.6998438){0.21}
\psarc[linewidth=0.04](7.771875,2.4001563){3.41}{-60.05122}{80.26301}
\psdots[dotsize=0.4](13.521875,0.73015624)
\usefont{T1}{ppl}{m}{n}
\rput(15.596406,1.9801563){\huge $v''$}
\psline[linewidth=0.04cm,linestyle=dashed,dash=0.16cm 0.16cm](13.481875,2.7301562)(14.881875,1.8301562)
\psline[linewidth=0.04cm,linestyle=dashed,dash=0.16cm 0.16cm](14.881875,1.8301562)(13.541875,0.73015624)
\psline[linewidth=0.04cm](13.521875,0.71015626)(13.521875,2.7701561)
\usefont{T1}{ppl}{m}{n}
\rput(13.696406,3.2001562){\huge $y$}
\psdots[dotsize=0.4](13.541875,2.6901562)
\usefont{T1}{ppl}{m}{n}
\rput(13.956407,0.40015626){\huge $y'$}
\pscircle[linewidth=0.06,linecolor=color1817,dimen=outer,fillstyle=solid,fillcolor=color1817b](14.811875,1.8001562){0.21}
\usefont{T1}{ppl}{m}{n}
\rput(0.9164063,5.0201564){\huge $v'$}
\psline[linewidth=0.04cm,linestyle=dashed,dash=0.16cm 0.16cm](1.3927741,4.96797)(1.7927623,3.773144)
\psline[linewidth=0.04cm,linestyle=dashed,dash=0.16cm 0.16cm](1.4244376,4.936388)(2.7899532,5.3935175)
\psline[linewidth=0.04cm](1.8054603,3.735213)(2.7457547,5.3365393)
\psdots[dotsize=0.4,dotangle=18.50884](2.7647202,5.3428884)
\psdots[dotsize=0.4,dotangle=18.50884](1.8117278,3.779493)
\rput{18.50884}(1.6423894,-0.20688945){\pscircle[linewidth=0.06,linecolor=color1817,dimen=outer,fillstyle=solid,fillcolor=color1817b](1.4560604,4.936429){0.21}}
\usefont{T1}{ppl}{m}{n}
\rput(2.8764062,5.800156){\huge $x$}
\usefont{T1}{ppl}{m}{n}
\rput(1.7364062,3.3201563){\huge $x'$}
\psdots[dotsize=0.4](9.281875,4.7901564)
\psdots[dotsize=0.4](7.301875,4.7901564)
\psdots[dotsize=0.4](8.321875,5.7501564)
\psdots[dotsize=0.4](4.101875,2.7101562)
\usefont{T1}{ppl}{m}{n}
\rput(8.596406,6.200156){\huge $w_{\mbox{\small $1$}}$}
\psline[linewidth=0.04cm](8.341875,5.7501564)(7.341875,4.8501563)
\psline[linewidth=0.04cm](7.301875,4.7901564)(9.321875,4.7901564)
\psline[linewidth=0.04cm](8.281366,2.770167)(7.2623634,3.711248)
\psdots[dotsize=0.4,dotangle=-180.06073](7.3023844,3.7312055)
\psdots[dotsize=0.4,dotangle=-180.06073](9.282383,3.7291064)
\psdots[dotsize=0.4,dotangle=-180.06073](8.261366,2.7701883)
\psline[linewidth=0.04cm](8.241366,2.7702096)(9.24232,3.669149)
\psline[linewidth=0.04cm](9.282383,3.7291064)(7.2623844,3.731248)
\psline[linewidth=0.04cm](7.301875,4.7901564)(9.301875,3.8501563)
\psline[linewidth=0.04cm](9.281875,4.8101563)(7.321875,3.7901564)
\psline[linewidth=0.04cm](4.041875,2.6501563)(2.521875,-0.52984375)
\psline[linewidth=0.04cm](8.241875,2.7501562)(2.521875,-0.56984377)
\psline[linewidth=0.04cm](4.141875,2.6101563)(9.481875,-0.6298438)
\psline[linewidth=0.04cm](4.241875,2.6901562)(13.401875,0.7501562)
\psline[linewidth=0.04cm](8.201875,2.8101563)(13.501875,2.7101562)
\usefont{T1}{ppl}{m}{n}
\rput(4.036406,3.1401563){\huge $w$}
\usefont{T1}{ppl}{m}{n}
\rput(6.916406,4.5001564){\huge $w_{\mbox{\small $2$}}$}
\usefont{T1}{ppl}{m}{n}
\rput(9.776406,4.4601564){\huge $w_{\mbox{\small $3$}}$}
\usefont{T1}{ppl}{m}{n}
\rput(6.956406,3.3001564){\huge $w_{\mbox{\small $4$}}$}
\usefont{T1}{ppl}{m}{n}
\rput(9.716406,3.3001564){\huge $w_{\mbox{\small $5$}}$}
\usefont{T1}{ppl}{m}{n}
\rput(8.5364065,2.3601563){\huge $w_{\mbox{\small $6$}}$}
\usefont{T1}{ppl}{m}{n}
\rput(6.238125,-6.219844){\huge (b)}
\end{pspicture} 
}
\caption{Stitching together one pair of cycles with a common vertex $v$}.
\label{fig:W}
\end{figure}

We are now ready to describe the gadget $W_v$. This gadget has four entry points, namely $v^\prime, v^{\prime\prime}$, and $a,b$. The gadget is shown in Figure~\ref{fig:W}. It is easy to check that the gadget induces
a braid. Now, in $G$, to insert this gadget, we remove $v$ from $G$, and make $v_1,v_2$ adjacent to $a$ and $v_3,v_4$ adjacent to $b$. Let us denote this graph by $G^\prime$. Note that there is a path 
from $v_1$ to $v_2$ along the cycle $C_v$ and there is a path from $v_3$ to $v_4$ along the cycle $C_v^\prime$. 

For equivalence, we need to be sure that if even one of $v_1, v_2, v_3, v_4$ is not picked in a vertex cover of $G^\prime$, then we have enough room for the vertex $v$ 
in the reverse direction. To this end, we show the following crucial property of the gadget $W_v$.

\begin{lemma}
\label{lem:basic_correctness}
Let $S^\prime$ be any vertex cover of $G$. If one of $a$ or $b$ belongs to $S^\prime$, then $|S^\prime \cap V(W)| = 10$. On the other hand, there exists a vertex cover  $S^\prime$ of $G^\prime$, that contains neither $a$ nor $b$, for which $|S^\prime \cap V(W)| = 9$. 
\label{W1vc}
\end{lemma}
\begin{proof}
The vertices $\{v^\prime,x,x^\prime\}$, $\{v^{\prime\prime},y,y^\prime\}$, $\{w_1,w_2,w_3\}$, $\{w_4,w_5,w_6\}$ form triangles, and $(w,a)$ is an edge disjoint from these triangles. 
Therefore, we clearly require minimum of 9 vertices to cover edges of $W$ alone. If we have $S^\prime=\{x,x^\prime,y,y^\prime,w_1,w_2,w_4,w_6,w\}$ then we can cover all edges of $W$ with 9 vertices. 
This proves the second part of the claim. 

Now, let $S^\prime$ be a vertex cover such that $a \in S^\prime$. Then, apart from the $K_3$'s present we have an edge $(w,b)$, so including $a$ we need at least 10 vertices. 
An analogous argument holds when $b$ is present (edge $(w,a)$ left).
If we have $V'=\{x,x^\prime,y,y^\prime,w_1,w_2,w_4,w_6,a,b\}$ then we can cover all edges of $W$ with 10 vertices.
\qed
\end{proof}


\begin{corollary}We have a vertex cover of size $k$ in $G$ if and only if $G^\prime$ admits a vertex cover of size $k+9$. 
\end{corollary}

\begin{proof}
Let $S$ be a vertex cover of $G$. If $v \notin S$, then $\{v_1,v_2,v_3,v_4\} \subseteq S$. Therefore, we may cover the edges of $W$ using the vertices $\{x,x^\prime,y,y^\prime,w_1,w_2,w_4,w_6,w\}$ since 
there is no external obligation to pick either $a$ or $b$, and this would be an extension of $S$ with at most nine additional vertices. If $v \in S$, then let $S^\star := S \setminus \{v\}$. We extend $S^\star$ by 
the set $\{x,x^\prime,y,y^\prime,w_1,w_2,w_4,w_6,a,b\}$, which also adds up to $k+9$. In the reverse direction, given a vertex cover of size $k+9$, we know that at least nine vertices of $S^\prime$ are from $W$.  
Let $S^\dagger$ denote the remaining vertices of $S^\prime$. If all of $\{v_1,v_2,v_3,v_4\} \in S^\prime$, then note that $S^\star$ is a vertex cover of size $k$ for $G$. If one of $\{v_1,v_2,v_3,v_4\} \notin S^\prime$, 
then $S^\dagger \cup \{v\}$ is a vertex cover of size at most $k$ in $G$. The size bound comes from Lemma~\ref{lem:basic_correctness} and the fact that either $a$ or $b$ 
belongs to $S^\prime$ due to the case we are considering. 
\end{proof}

We should be able to use this gadget repeatedly to stitch all cycles into two single long cycles. However, the iterative process involves several challenges. For instance, if some gadgets are already inserted, 
the paths along the cycles that we had before may not be so readily available. Also, while the process of breaking the cycle at a vertex $v$ is clear, it is not obvious as to how one would mimic this 
construction for a neighbor of $v$ after $v$ has been suitably replaced. The concern here is that a straightforward application of the gadget will cause vertices from two different gadgets
to become adjacent, which we would like to avoid if we are to maintain the braid structure of the gadget itself. 

To address the former problem we create a slightly different gadget that creates artificial paths that can be used if the original cycle is broken by some previously inserted gadget. 
For the second problem, we start our reduction from cubic graphs and proceed in a manner so as to ensure that the cycle decompositions of the reduced graph are somewhat special. This allows us to choose mutually non-adjacent breakpoints $v$.

Finally, these gadgets must be tied together into a path, which we organize with the help of connection gadgets. The reader is referred to the full version of this work for the complete details. 

\begin{theorem}
The problem of finding a vertex cover of size at most $k$ in a braid graph is \NPC{}.
\end{theorem}

}

\longversion{

\paragraph{Transforming $2$-connected cubic planar graphs to $4$-regular graphs.}  Consider an instance of vertex cover on a cubic $2$-connected planar graph $G=(V,E)$, where $\lvert V \rvert=n$. We transform this graph to a four-regular graph in two steps. This transformation is important because for turning a four-regular graph into a union of two Hamiltonian paths, we need the underlying decomposition into two-factors to have certain properties, which we will ensure in the first half of the reduction. 

The transformation from cubic graphs to four-regular graphs happens in two stages. First, we replace certain edges with gadgets that only serve to elongate certain paths in the graph. This is a technical artifact that will be useful later. Next, we add gadgets that increase the degree of every vertex so as to obtain a four-regular graph. 

\sloppypar
We begin by making two copies of $G$, say, $G_u$ and $G_v$. Here, we let $G_u=(V_u,E_u),G_v=(V_v,E_v)$, $V_u=\{u_0,u_1,...,u_{n-1}\}, V_v=\{v_0,v_1,...,v_{n-1}\}$
and let $f \colon V_u \rightarrow V_v$ be a function with $u_i \mapsto v_i, 0 \leq i < n$,
determining the graph isomorphism from $G_u$ to $G_v$. 

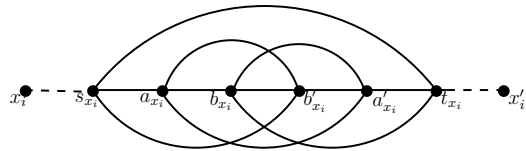
\begin{figure}[!ht]
\centering
\scalebox{0.65} 
{
\begin{pspicture}(0,0.47)(11.406875,3.41)
\psdots[dotsize=0.24](3.4634376,1.65)
\psdots[dotsize=0.24](4.8634377,1.65)
\psdots[dotsize=0.24](6.2634373,1.65)
\psdots[dotsize=0.24](7.6434374,1.65)
\psdots[dotsize=0.24](9.063437,1.65)
\psdots[dotsize=0.24](10.443438,1.65)
\psdots[dotsize=0.24](0.6634375,1.65)
\psdots[dotsize=0.24](2.0434375,1.65)
\psline[linewidth=0.04cm,linestyle=dashed,dash=0.16cm 0.16cm](0.6434375,1.67)(2.1463437,1.63)
\psline[linewidth=0.04cm](2.1463437,1.67)(9.276875,1.67)
\psline[linewidth=0.04cm,linestyle=dashed,dash=0.16cm 0.16cm](9.276875,1.67)(10.463437,1.67)
\psarc[linewidth=0.04](5.5534377,-1.04){4.43}{38.38654}{141.53214}
\psarc[linewidth=0.04](4.8734374,1.26){1.43}{19.612093}{160.9249}
\psarc[linewidth=0.04](6.2534375,1.18){1.43}{19.612093}{160.9249}
\psarc[linewidth=0.04](4.1534376,2.98){2.49}{-147.05078}{-33.178513}
\psarc[linewidth=0.04](5.5334377,2.98){2.49}{-147.05078}{-33.178513}
\psarc[linewidth=0.04](6.9334373,2.98){2.49}{-147.05078}{-33.178513}
\usefont{T1}{ppl}{m}{n}
\rput(0.516875,1.465){\large $x_i$}
\usefont{T1}{ppl}{m}{n}
\rput(10.716875,1.485){\large $x'_i$}
\usefont{T1}{ppl}{m}{n}
\rput(1.916875,1.465){\large $s_{x_i}$}
\usefont{T1}{ppl}{m}{n}
\rput(3.256875,1.465){\large $a_{x_i}$}
\usefont{T1}{ppl}{m}{n}
\rput(4.666875,1.465){\large $b_{x_i}$}
\usefont{T1}{ppl}{m}{n}
\rput(6.586875,1.445){\large $b'_{x_i}$}
\usefont{T1}{ppl}{m}{n}
\rput(8.016875,1.425){\large $a'_{x_i}$}
\usefont{T1}{ppl}{m}{n}
\rput(9.376875,1.465){\large $t_{x_i}$}
\end{pspicture} 
}

\caption{Figure: Gadget $J_{x_i}$ with solid lines showing gadget edges and dotted lines its connection to outside vertices.}
 \label{J}
\end{figure}

It is shown in ~\cite{Chudnovsky08perfectmatchings} that a planar cubic graph with no cut edge has exponentially many perfect matchings. Since we have a 2-connected cubic planar graph, there are evidently no cut edges. Therefore, we compute a perfect matching $M_u=\{(u_0,u'_0),(u_1,u'_1 ),...,(u_{\frac{n}{2}-1},u'_{\frac{n}{2}-1})\}$
of $G_u$ by the algorithm given in ~\cite{4567800}. Let the corresponding matching of $G_v$ be $M_v=\{(v_0,v'_0),(v_1,v'_1),...,(v_{\frac{n}{2}-1},v'_{\frac{n}{2}-1})\}$.


For $x \in \{u,v\}$ we now describe how to construct $\hat{G_x}$ from $G_x$. We use gadget $J_{x_i}$ shown in Figure \ref{J}, where subscript $x_i$ corresponds to graph $G_x$ and edge $(x_i,x'_i)$. In Lemma ~\ref{VCJ}
we show the minimum number of vertices required to cover edges in $E(J_{x_i}) \cup \{(x_i,s_{x_i}),(x'_i,t_{x_i})\}$ based on whether $x_i,x'_i$ are picked in the vertex cover or not. We define the sequence
$\rho(J_{x_i}):=(s_{x_i},a_{x_i},b_{x_i},b'_{x_i},a'_{x_i},t_{x_i})$, notice that this is a path in $J_{x_i}$ and we also have the sequence $\rho'(J_{x_i}):=(s_{x_i},t_{x_i},b_{x_i},a'_{x_i},a_{x_i},b'_{x_i},s_{x_i})$, which is a cycle in $J_{x_i}$. It is easy to see that all edges in $J_{x_i}$ is either in the path $\rho(J_{x_i})$ or on the cycle $\rho'(J_{x_i})$. We will refer to these sequences later, when we are specifying how the reduced graph is an union of two Hamiltonian paths.

To construct $\hat{G_x}$ we initially set $\hat{G_x}=G_x$.
Fix the matching $M_x=\{(x_0,x'_0),(x_1,x'_1 ),...,(x_{\frac{n}{2}-1},x'_{\frac{n}{2}-1})\}$. We are going to replace every edge in this matching with the gadget $J_{x_i}$. More formally, we have:
$$V(\hat{G_x})=V(\hat{G_x}) \biguplus \left( \bigcup_{0 \leq i < n/2} V(J_{x_i}) \right)$$ 

and  

$$E(\hat{G_x})=E(\hat{G_x}) \smallsetminus \left( \bigcup_{0 \leq i < n/2} \{(x_i,x'_i)\} \right) \biguplus \left( \bigcup_{0 \leq i < n/2} E(J_{x_i}) \uplus \{(x_i,s_{x_i}),(t_{x_i},x'_i)\} \right).$$

We first have the following observation.


\begin{lemma}
Consider the graph $J_i=(V(J_{x_i})\uplus \{x_i,x'_i\},E(J_{x_i})\uplus \{(x_i,s_{x_i}),(t_{x_i},x'_i)\})$ described above. Let $S$ be an optimal vertex cover, and let $S_i$ denote $S \cap V(J_{x_i})$. If $S$ includes at least one of $x_i,x'_i$, then $|S_i| = 4$, else $|S_i| = 5$.
\label{VCJ}
\end{lemma}

\begin{proof}
The vertices $a_{x_i},b_{x_i},b'_{x_i},a'_{x_i}$ induced a $K_4$, therefore any vertex cover includes at least 3 vertices among them. The following cases arise depending on whether $x_i,x'_{i}$ is included in $S$ or not.
\begin{itemize}
 \item [case 1] If $x_i,x'_i \notin S$ then we are forced to include $s_{x_i},t_{x_i}$ in $S_i$. So we need at least give vertices in $S_i$ and note that $X := \{s_{x_i},t_{x_i},a_{x_i},a'_{x_i},b_{x_i}\}$ is a set of five vertices that covers all edges in $E(J_i)$.
 \item [case 2] Suppose $S$ picks exactly one of $x_i$ and $x'_i$. In particular, let's say that $x_i \notin S$, then we have $s_{x_i}$. Therefore, we need at least four vertices in $S_i$. Note also that $X := \{s_{x_i},a_{x_i},b_{x_i},a'_{x_i}\}$ is a set of four vertices that covers all edges in $E(J_i)$.
 \item [case 3]  Finally, let $x_i,x'_i \in S$. Since $(s_{x_i},t_{x'_i}) \in E(J_i)$, at least one of $s,t$ should be in $S$. The rest of the argument is analogous to the previous case.
\end{itemize}
\qed
\end{proof}

In the next Lemma, we establish the relation between the vertex cover of $G_x$ and vertex cover of $\hat{G_{x}}$.

\begin{lemma}
For $x \in \{u,v\}$, the graph $G_x$ admits a vertex cover of size $p$ if, and only if, the graph $\hat{G_x}$ has a vertex cover of size $(p+2n)$.
 \label{pvcGin4}
\end{lemma}
\begin{proof}
In the forward direction, consider a vertex cover $V_c$ of $G_x$. Note that for every matching edge $(x_i,x'_i) \in M_x$, $V_c \cap \{x_i,x'_i\} \neq \varnothing$. So for each $J_{x_i}$ corresponding to $(x_i,x'_i) \in M_x$ we need exactly four more vertices to cover $E(J_{x_i}) \cup \{(x_i,s_{x_i}),(t_{x_i},x'_i)\}$ by Lemma~\ref{VCJ}. But $ \lvert M_x \rvert =n/2$, so $p+4\frac{n}{2}=p+2n$ vertices are sufficient to cover $E(\hat{G_x})$.

In the reverse direction, let $V_c$ be a vertex cover of size $(p+2n)$ for $\hat{G_x}$, and let $V'=V_c \cap V(G_x)$. Then $V'$ covers all edges in $G_x$ except possibly edges $(x_i,x'_i) \in M_x$. For each gadget $J_{x_i}$ inserted between $(x_i,x'_i) \in M_x$ we know if both of $x_i,x'_i \notin V_c$ then from lemma ~\ref{VCJ} we require five additional vertices to cover $E(\hat{G_x}) \cup \{(x_i,s_{x_i}),(t_{x_i},x'_i)\}$.
For covering the edge $x_i,x'_i$ in $G_x$ we need at least one of $x_i,x'_i$, so we modify $V_c$ to include any one of $x_i,x'_i$ and four more vertices from $V(J_{x_i})$ (note that the size of the vertex cover remains unchanged after this operation). After repeating this operation for all matching edges where it is necessary, we have that $V'$ forms a vertex cover of size $p$ for $G_x$.
\qed

\end{proof}

We now turn to the gadgets that add to the degree of every vertex in the graph, turning it into a four-regular graph. For this we need a new gadget, which we refer to as $\tilde{J}_i$ (indexed by $i$), as shown in figure ~\ref{tildeJ}. Here irrespective of whether $u'_i,v'_i,u_{i+1},v_{i+1}$ is 
included or not we require $6$ vertices as $p_i,q_i,m_i,n_i$ and
$p'_i,q'_i,m'_i,n'_i$ induce complete graphs. Further, we can include $m_i,n_i,m'_i,n'_i,p_i,p'_i$ in vertex cover to cover edges incident to vertices in $V(J_i)$, and we will also cover any of the edges connecting these vertices to the rest of the graph. Now we are ready to describe the actual construction.  


We arbitrarily order the edges in matching, say as follows:
$$M_u=\{(u_0,u'_0),(u_1,u'_1),...,(u_{\frac{n}{2}-1},u'_{\frac{n}{2}-1})\}\mbox{ of }G_u,$$ and

$$M_v=\{(v_0,v'_0),(v_1,v'_1),...,(v_{\frac{n}{2}-1},v'_{\frac{n}{2}-1})\}\mbox{ of }G_v.$$ 

We will follow this ordering in every step of reduction wherever required.
Note that all
 vertices of $G_x$ in $\hat{G}_x$ are still degree 3 vertices, where $x \in \{u,v\}$. We insert a copy of gadget $\tilde{J}_i$ between edges $(u_i,u'_i),(u_{i+1},u'_{i+1}) \in M_u$ and $(v_i,v'_i),(v_{i+1},v'_{i+1}) \in M_v$ to
increase the degree of the vertices $u'_i,u_{i+1},v'_i$ and $v_{i+1}$. We do this
by adding edges $(u'_i,m_i),(u_{i+1},n'_i),(v'_i,n_i)$ and $(v_{i+1},m'_i)$.
for $0 \leq i \leq \frac{n}{2}-1$, where the index computed modulo $\frac{n}{2}$. 
We refer to the graph constructed above as $\tilde{G}$. It is easy to see that $\tilde{G}$ thus obtained is a $4$-regular graph.

\begin{figure}
\centering
\scalebox{0.5} 
{
\begin{pspicture}(0,-2.1340625)(18.445625,2.0940626)
\definecolor{color2118}{rgb}{0.09411764705882353,0.27450980392156865,0.8941176470588236}
\psdots[dotsize=0.5,linecolor=color2118,fillstyle=solid,dotstyle=o](11.0,0.769375)
\psdots[dotsize=0.5,linecolor=color2118,fillstyle=solid,dotstyle=o](11.0,-0.430625)
\psdots[dotsize=0.5,linecolor=color2118,fillstyle=solid,dotstyle=o](12.6,0.769375)
\psdots[dotsize=0.5,linecolor=color2118,fillstyle=solid,dotstyle=o](13.8,0.769375)
\psdots[dotsize=0.5,linecolor=color2118,fillstyle=solid,dotstyle=o](15.4,-0.410625)
\psdots[dotsize=0.5,linecolor=color2118,fillstyle=solid,dotstyle=o](15.38,0.749375)
\psdots[dotsize=0.24](2.12,0.769375)
\psdots[dotsize=0.24](3.72,0.769375)
\psdots[dotsize=0.24](2.12,-0.430625)
\psdots[dotsize=0.24](3.7,-0.410625)
\psdots[dotsize=0.24](4.92,0.769375)
\psdots[dotsize=0.24](4.92,-0.430625)
\psdots[dotsize=0.24](6.5,0.769375)
\psdots[dotsize=0.24](6.52,-0.430625)
\psdots[dotsize=0.24](6.54,-0.410625)
\psline[linewidth=0.04](2.1,0.789375)(2.1,-0.430625)(3.68,-0.430625)(3.68,-0.430625)
\psline[linewidth=0.04](3.68,-0.430625)(2.08,0.789375)(3.7,0.789375)(3.7,-0.430625)
\psline[linewidth=0.04cm](2.1,-0.410625)(3.7,0.749375)
\usefont{T1}{ppl}{m}{n}
\rput(0.7528125,1.794375){\Large $u'_i$}
\usefont{T1}{ppl}{m}{n}
\rput(0.6728125,-1.405625){\Large $v'_i$}
\usefont{T1}{ppl}{m}{n}
\rput(2.3928125,1.094375){\Large $m_i$}
\usefont{T1}{ppl}{m}{n}
\rput(3.6428125,1.034375){\Large $p_i$}
\usefont{T1}{ppl}{m}{n}
\rput(5.0728126,1.054375){\Large $p'_i$}
\usefont{T1}{ppl}{m}{n}
\rput(6.7828126,1.074375){\Large $m'_i$}
\usefont{T1}{ppl}{m}{n}
\rput(2.3028126,-0.725625){\Large $n_i$}
\usefont{T1}{ppl}{m}{n}
\rput(3.4628124,-0.785625){\Large $q_i$}
\usefont{T1}{ppl}{m}{n}
\rput(4.8928127,-0.765625){\Large $q'_i$}
\usefont{T1}{ppl}{m}{n}
\rput(6.7128124,-0.705625){\Large $n'_i$}
\psdots[dotsize=0.24](11.0,0.769375)
\psdots[dotsize=0.24](12.6,0.769375)
\psdots[dotsize=0.24](11.0,-0.430625)
\psdots[dotsize=0.24](12.58,-0.410625)
\psdots[dotsize=0.24](13.8,0.769375)
\psdots[dotsize=0.24](13.8,-0.430625)
\psdots[dotsize=0.24](15.38,0.769375)
\psdots[dotsize=0.24](15.4,-0.430625)
\psdots[dotsize=0.24](15.42,-0.410625)
\psline[linewidth=0.04](10.98,0.789375)(10.98,-0.430625)(12.56,-0.430625)(12.56,-0.430625)
\psline[linewidth=0.04](12.56,-0.430625)(10.96,0.789375)(12.58,0.789375)(12.58,-0.430625)
\psline[linewidth=0.04cm](10.98,-0.410625)(12.58,0.749375)
\usefont{T1}{ppl}{m}{n}
\rput(17.112812,-1.405625){\Large $u_{i+1}$}
\usefont{T1}{ppl}{m}{n}
\rput(17.052813,1.814375){\Large $v_{i+1}$}
\usefont{T1}{ppl}{m}{n}
\rput(10.712812,-0.725625){\Large $y$}
\usefont{T1}{ppl}{m}{n}
\rput(4.0223436,-1.885625){\Large (a)}
\usefont{T1}{ppl}{m}{n}
\rput(13.012343,-1.885625){\Large (b)}
\psline[linewidth=0.04cm](3.66,-0.430625)(4.92,0.749375)
\psline[linewidth=0.04cm](4.92,0.769375)(6.5,-0.350625)
\psline[linewidth=0.04](3.72,0.789375)(4.92,-0.450625)(4.92,0.749375)(6.52,0.769375)(6.52,-0.450625)(4.92,-0.450625)(6.48,0.749375)
\psline[linewidth=0.04](15.36,0.789375)(15.38,-0.450625)(13.78,-0.450625)(13.76,0.789375)(15.34,0.789375)(13.82,-0.390625)
\psline[linewidth=0.04cm](12.58,0.789375)(13.8,-0.450625)
\psline[linewidth=0.04](12.56,-0.430625)(13.76,0.809375)(15.34,-0.390625)
\psdots[dotsize=0.24](0.86,1.649375)
\psdots[dotsize=0.24](0.86,-1.330625)
\psdots[dotsize=0.24](7.86,1.669375)
\psdots[dotsize=0.24](7.84,-1.350625)
\psdots[dotsize=0.24](9.66,1.649375)
\psdots[dotsize=0.24](9.64,-1.330625)
\psdots[dotsize=0.24](16.66,1.649375)
\psdots[dotsize=0.24](16.66,-1.310625)
\psline[linewidth=0.04cm,linestyle=dashed,dash=0.16cm 0.16cm](6.5,0.749375)(7.88,1.649375)
\psline[linewidth=0.04cm,linestyle=dashed,dash=0.16cm 0.16cm](6.5,-0.430625)(7.8,-1.370625)
\psline[linewidth=0.04cm,linestyle=dashed,dash=0.16cm 0.16cm](9.64,1.669375)(11.0,0.709375)
\psline[linewidth=0.04cm,linestyle=dashed,dash=0.16cm 0.16cm](11.0,-0.430625)(9.64,-1.310625)
\psline[linewidth=0.04cm,linestyle=dashed,dash=0.16cm 0.16cm](15.38,0.789375)(16.68,1.629375)
\psline[linewidth=0.04cm,linestyle=dashed,dash=0.16cm 0.16cm](15.38,-0.390625)(16.62,-1.290625)
\psline[linewidth=0.04cm,linestyle=dashed,dash=0.16cm 0.16cm](0.8,1.689375)(2.12,0.749375)
\psline[linewidth=0.04cm,linestyle=dashed,dash=0.16cm 0.16cm](0.82,-1.350625)(2.1,-0.410625)
\usefont{T1}{ppl}{m}{n}
\rput(8.332812,-1.385625){\Large $u_{i+1}$}
\usefont{T1}{ppl}{m}{n}
\rput(8.272813,1.834375){\Large $v_{i+1}$}
\usefont{T1}{ppl}{m}{n}
\rput(9.392813,1.834375){\Large $u'_i$}
\usefont{T1}{ppl}{m}{n}
\rput(9.312813,-1.365625){\Large $v'_i$}
\usefont{T1}{ppl}{m}{n}
\rput(11.232813,1.094375){\Large $m_i$}
\usefont{T1}{ppl}{m}{n}
\rput(12.482813,1.034375){\Large $p_i$}
\usefont{T1}{ppl}{m}{n}
\rput(13.912812,1.054375){\Large $p'_i$}
\usefont{T1}{ppl}{m}{n}
\rput(15.622812,1.074375){\Large $m'_i$}
\usefont{T1}{ppl}{m}{n}
\rput(11.142813,-0.725625){\Large $n_i$}
\usefont{T1}{ppl}{m}{n}
\rput(12.302813,-0.785625){\Large $q_i$}
\usefont{T1}{ppl}{m}{n}
\rput(13.732813,-0.765625){\Large $q'_i$}
\usefont{T1}{ppl}{m}{n}
\rput(15.552813,-0.705625){\Large $n'_i$}
\end{pspicture} 
}

\caption{Figure: a) Gadget $\tilde{J}_i$. b)Minimum vertices required to cover all the edges of $\tilde{J}_i$}
\label{tildeJ}
\end{figure}
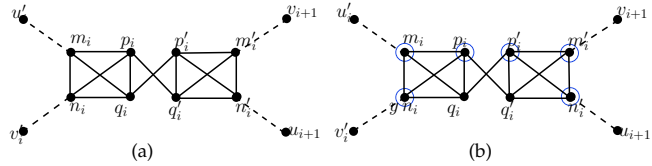

In lemma~\ref{3to4regular}, we establish the relation between vertex cover of $\tilde{G}$ and vertex cover of $\hat{G}_u \cup \hat{G}_v$.

\begin{lemma}The graph $\hat{G}$ has a vertex cover of size $p$ if, and only if,$\tilde{G}$ has a vertex cover of size $(p+3n)$, where $\hat{G}=\hat{G}_u \cup \hat{G}_v$.
\label{3to4regular} 
\end{lemma}
\begin{proof}

In the forward direction, suppose $\hat{G}$ has a vertex cover of size $p$.
Since $\hat{G} \subset \tilde{G}$, only extra edges in $\tilde{G}$ are those adjacent to vertices of gadgets $\tilde{J}_i$ for $0 \leq i < n/2$.
From $V(J_i)$ if we include vertices $\{m_i,m'_i,n_i,n'_ip_i,p'_i\}$, then they will cover all edges that are adjacent to at least one vertex in $V(J_i)$, for all $0 \leq i < n/2$. This implies that $p+3n$ vertices are sufficient to cover all edges of $\tilde{G}$.

In the reverse direction, let $S$ be a vertex cover of size $p+3n$ for $\tilde{G}$, and let $V' := V(\hat{G})\cap S$. For each $J_i, 0 \leq i <n/2$, we need 6 vertices to cover edges adjacent to $V(J_i)$.  Therefore $\lvert V'\rvert \leq p$. It is easy to see that $V'$ is a vertex cover for $\hat{G}$.
\qed
\end{proof}

Again, for ease of describing paths at the end of this discussion, we define $\tau(\tilde{J}_i)=(m_i,p_i,n_i,q_i,p'_i,m'_i,q'_i,n'_i)$
and $\tau(\tilde{J}_i)=(n_i,m_i,q_i,p_i,q'_i,p'_i,n'_i,m'_i)$ be the two paths that cover all vertices and edges of $\tilde{J}_i$. Combining the two steps of the reduction above, we have the following.

\begin{corollary}
Let $G_{uv}=G_u \cup G_v$. The graph $G_{uv}$ admits a vertex cover of size $p$ if, and only if, the graph $\tilde{G}$ has a vertex cover of size $(p+5n)$. 
\label{WSDNPH}
\end{corollary}

\paragraph{An useful 2-factor decomposition of $\tilde{G}$.}
\label{2-factor} From Corollary 2.1.5(Petersen 1891) in ~\cite{DBLP:books/daglib/0030488}, every $2k$-regular graph has a $k$-factor as a subgraph. Here, we give an explicit partition of $\tilde{G}$ into two $2$-factors, namely $H,H'$.  Initially, let $H,H'$ be empty (no vertices). We have fixed a matching
$M_u=\{(u_0,u'_0),(u_1,u'_1),...,(u_{\frac{n}{2}-1},u'_{\frac{n}{2}-1})\}$ of $G_u$ and $M_v=\{(v_0,v'_0),(v_1,v'_1),...,(v_{\frac{n}{2}-1},v'_{\frac{n}{2}-1})\}$ of $G_v$ corresponding to which $\tilde{G}$ was constructed.

Recall that for the construction of $\tilde{G}$, we have deleted the matching edge $(x_i,x'_i)$ and inserted gadget $J_{x_i}$, for $x \in \{u,v\}$ and $0 \leq i <\frac{n}{2}$ and increased the degree of each vertices
$u'_j,u_{j+1} \in V(G_u)$ corresponding to matching edges $(u_j,u'_j),(u_{j+1},u'_{j+1})$ and $v'_j,v_{j+1} \in V(G_v)$
corresponding to matching edge $(v_j,v'_j),(v_{j+1},v'_{j+1})$ by connecting it to gadget $\tilde{J_i}$ for $0 \leq j < \frac{n}{2}$ (index computed modulo $\frac{n}{2}$). 

We now construct two two-factors which will be convenient for us in the next phase of the reduction. These two-factors will be highly structured in the following sense. The first cycle in both two-factors will involve all the vertices of $G_u$ and $G_v$, respectively, and some vertices from the gadgets. The rest of the graph now decomposes into a collection of cycles which get distributed in a natural way. we now describe this formally. 

The first cycle $C_H$ that we include in $H$ contains all vertices of $\hat{G}_u$ and $\tilde{J}_i$, $0 \leq i < \frac{n}{2}$ that are in $\tilde{G}$,
similarly the first cycle $C_{H'}$ we include in $H'$ contains all vertices of $\hat{G}_v$ and $\tilde{J}_i$, $0 \leq i < \frac{n}{2}$.

Specifically, the first cycle included in $H$ is $u_0 \rightarrow \rho(J_{u_0}) \rightarrow u'_0 \rightarrow \tau(\tilde{J}_0) \rightarrow
u_1 \rightarrow \rho(J_{u_1}) \rightarrow u'_1 \rightarrow \tau(\tilde{J}_1) \rightarrow u_2 \rightarrow ... \rightarrow
u_{\frac{n}{2}-2} \rightarrow \rho(J_{u_{\frac{n}{2}-2}}) \rightarrow u'_{\frac{n}{2}-2} \rightarrow \tau(\tilde{J}_{\frac{n}{2}-2}) \rightarrow
u_{\frac{n}{2}-1} \rightarrow \rho(J_{u_{\frac{n}{2}-1}}) \rightarrow u'_{\frac{n}{2}-1} \rightarrow \tau(\tilde{J}_{\frac{n}{2}-1}) \rightarrow u_0$.

Similarly, the first cycle $C_{H'}$ included in $H'$ is $v_0 \rightarrow \rho(J_{v_0}) \rightarrow v'_0 \rightarrow \tau'(\tilde{J}_0) \rightarrow
v_1 \rightarrow \rho(J_{v_1}) \rightarrow v'_1 \rightarrow \tau'(\tilde{J}_1) \rightarrow v_2 \rightarrow ... \rightarrow
v_{\frac{n}{2}-2} \rightarrow \rho(J_{v_{\frac{n}{2}-2}}) \rightarrow v'_{\frac{n}{2}-2} \rightarrow \tau'(\tilde{J}_{\frac{n}{2}-2}) \rightarrow
v_{\frac{n}{2}-1} \rightarrow \rho(J_{v_{\frac{n}{2}-1}}) \rightarrow v'_{\frac{n}{2}-1} \rightarrow \tau'(\tilde{J}_{\frac{n}{2}-1}) \rightarrow v_0$.

For $H$ to be one of the factors we require cycles in $H$ containing vertices of $\hat{G}_v$ present in $\tilde{G}$, so we include the cycles $\rho'(J_{v_i})$, for $0 \leq i < \frac{n}{2}$ and
refer to these set of cycles as $\mathcal{C}_J$. Since we have already used two degrees of vertices of $G_v$ we are left with degree two from each vertex which forms disjoint union of cycles, so we let $\mathcal{C}_v$ denote the set of cycles that is left in the graph $\hat{G_v}$. Similarly we include in $H'$ the cycles $\rho'(J_{u_i})$, for $0 \leq i < \frac{n}{2}$ and refer to them by $\mathcal{C}_J'$
and include the graph $\hat{G_u}$ after deleting already included edge in corresponding cycle and refer to cycles obtained after deletion by $\mathcal{C}_u$.

Let $\mathfrak{C}_H,\mathfrak{C}_{H'}$ be set of cycles in $H$ and $H'$ respectively.
We will choose $V_H \subseteq V(\tilde{G})$ such that $\forall C \in \mathcal{C}_J \cup \mathcal{C}_v$, $\lvert V(C) \cap V_H \rvert = 1$ and $\forall C' \in \mathcal{C}_J' \cup \mathcal{C}_u$,
$\lvert V(C') \cap V_H \rvert = 1$.
We will delete vertices in $V_H$ and insert some other set of vertices, the aim is to break all cycles at some vertices 
and stitch together the broken cycles in $H$ to form one of the hamiltonin path and similarly form the other hamiltonian path using
the cycles in $H'$.

Initially let $V_H = \varnothing$. Let $V_1= \{b_{x_i} \in V(J_{x_i})$  $ \lvert x \in \{u,v\}$ and $0 \leq i < \frac{n}{2}\}$ and $V_2=\varnothing$. Observe that the cycles in $\mathcal{C}_u$ and $\mathcal{C}_v$ are disjoint among them (no common vertex), so we select one vertex from each cycle $C \in (\mathcal{C}_u \cup \mathcal{C}_v)$ and include it in $V_2$ and 
$V_H=V_1 \cup V_2$.
Note that the vertex selected from cycles in $\mathcal{C}_u$ and $\mathcal{C}_J$ are present in the cycle $C_{H'}$,
similarly vertex selected from cycles in $\mathcal{C}_v$ and $\mathcal{C}_J'$ are present in the cycle $C_{H}$.

\paragraph{From Constructed 4-regular graph with 2-factoring $H,H'$ to a \braid{} graph.} We order the cycles in $H$ and $H'$ and call the ordered sequence of cycles as $S_H,S_{H'}$. The first cycle in $S_H$ is $C_H$, 
next we arbitrary order the cycles in $\mathcal{C}_{J}$
from $C_2,C_3,...,C_{p'}$ and include in $S_H$,
and then we arbitrarily order the cycles in $\mathcal{C}_v$ from $C_{p'+1},C_{p'+2},...,C_k$ and include in $S_H$.

It is easy to see that once we have fixed the ordering of cycles in $H$ we already have the corresponding ordering of cycles in $H'$ which is given by the function $f$ indicating the graph isomorphism 
from $G_u$ to $G_v$. Also we order the vertices in $V_H$ corresponding to the ordering of cycles of $H$. First, we give the gadgets that we will use for breaking cycles that ensures 
the hamiltonian property and relates the vertex cover of $\tilde{G}$ to the vertex cover of new graph constructed.

\begin{figure}
\centering
\scalebox{0.55} 
{
\begin{pspicture}(0,-4.6840625)(7.205625,4.6840625)
\psarc[linewidth=0.04](2.35,2.389375){1.55}{137.24574}{58.760784}
\psdots[dotsize=0.3](1.18,3.459375)
\psdots[dotsize=0.3](3.16,3.719375)
\psdots[dotsize=0.3](2.02,4.079375)
\psline[linewidth=0.04cm](1.18,3.439375)(1.98,4.079375)
\psline[linewidth=0.04cm](2.02,4.119375)(3.14,3.759375)
\usefont{T1}{ppl}{m}{n}
\rput(1.9528126,4.424375){\Large $v_i$}
\usefont{T1}{ppl}{m}{n}
\rput(0.8128125,3.664375){\Large $v_{i_1}$}
\usefont{T1}{ppl}{m}{n}
\rput(3.6228125,3.824375){\Large $v_{i_2}$}
\usefont{T1}{ppl}{m}{n}
\rput(3.3728125,0.484375){\Large $C_i  \in  \mathfrak{C}_H$}
\psarc[linewidth=0.04](2.25,-2.030625){1.55}{137.24574}{58.760784}
\psdots[dotsize=0.3](1.08,-0.960625)
\psdots[dotsize=0.3](1.92,-0.340625)
\psline[linewidth=0.04cm](1.08,-0.980625)(1.88,-0.340625)
\psline[linewidth=0.04cm](1.92,-0.300625)(3.04,-0.660625)
\usefont{T1}{ppl}{m}{n}
\rput(1.8528125,0.004375){\Large $v_i$}
\usefont{T1}{ppl}{m}{n}
\rput(0.6828125,-0.775625){\Large $v_{i_3}$}
\usefont{T1}{ppl}{m}{n}
\rput(3.5228126,-0.595625){\Large $v_{i_4}$}
\usefont{T1}{ppl}{m}{n}
\rput(3.5928125,-3.895625){\Large $C'_j  \in \mathfrak{C}_{H'}$}
\psdots[dotsize=0.3](3.06,-0.700625)
\usefont{T1}{ppl}{m}{n}
\rput(2.2623436,-4.435625){\Large (a)}
\end{pspicture} 
}
\scalebox{0.4} 
{
\begin{pspicture}(0,-6.5620313)(16.489063,6.5620313)
\definecolor{color1687}{rgb}{0.12549019607843137,0.13725490196078433,0.18823529411764706}
\definecolor{color1687d}{rgb}{0.9098039215686274,0.8901960784313725,0.8901960784313725}
\definecolor{color1690}{rgb}{0.054901960784313725,0.07450980392156863,0.12941176470588237}
\definecolor{color1703b}{rgb}{0.996078431372549,0.996078431372549,0.996078431372549}
\definecolor{color1703}{rgb}{0.07450980392156863,0.07450980392156863,0.9176470588235294}
\psline[linewidth=0.04cm](8.301875,5.7501564)(9.321875,4.8101563)
\psline[linewidth=0.04cm](1.8645116,3.716866)(4.101875,2.7501562)
\psline[linewidth=0.04cm](2.641875,5.3901563)(8.281875,5.7701564)
\psbezier[linewidth=0.06,linecolor=color1687,doubleline=true,doublesep=0.08,doublecolor=color1687d](1.581875,-1.8140152)(1.8243687,-2.870711)(2.3975358,-3.7308128)(3.2793314,-4.418894)(4.161127,-5.106975)(4.921875,-5.409844)(5.981875,-5.409844)(7.041875,-5.409844)(7.527849,-5.2612906)(8.481875,-4.6698437)(9.435901,-4.078397)(10.267561,-2.6638515)(10.421875,-1.7698437)
\psbezier[linewidth=0.06,linecolor=color1690,doubleline=true,doublesep=0.08,doublecolor=color1687d](3.501875,-1.875014)(3.637935,-2.5478036)(3.9595308,-3.095423)(4.4542937,-3.5335186)(4.949057,-3.9716144)(5.2460384,-4.0185547)(5.864369,-4.034199)(6.482699,-4.049844)(7.0176253,-3.9552612)(7.546563,-3.5961037)(8.075501,-3.236946)(8.375292,-2.3287559)(8.461875,-1.7498437)
\psline[linewidth=0.04cm,linestyle=dashed,dash=0.16cm 0.16cm](2.501875,-0.6298438)(1.541875,-1.6898438)
\psline[linewidth=0.04cm,linestyle=dashed,dash=0.16cm 0.16cm](2.621875,-0.72984374)(3.501875,-1.6898438)
\usefont{T1}{ppl}{m}{n}
\rput(2.5364063,-1.1398437){\huge $a_i$}
\usefont{T1}{ppl}{m}{n}
\rput(0.93640625,-1.5598438){\huge $v_{i_1}$}
\usefont{T1}{ppl}{m}{n}
\rput(4.1764064,-1.5398438){\huge $v_{i_3}$}
\psdots[dotsize=0.4](2.501875,-0.60984373)
\pscircle[linewidth=0.06,linecolor=color1703,dimen=outer,fillstyle=solid,fillcolor=color1703b](1.571875,-1.6998438){0.21}
\pscircle[linewidth=0.06,linecolor=color1703,dimen=outer,fillstyle=solid,fillcolor=color1703b](3.451875,-1.6998438){0.21}
\psline[linewidth=0.04cm,linestyle=dashed,dash=0.16cm 0.16cm](9.341875,-0.7498438)(8.281875,-1.8298438)
\psline[linewidth=0.04cm,linestyle=dashed,dash=0.16cm 0.16cm](9.481875,-0.6898438)(10.461875,-1.5898438)
\usefont{T1}{ppl}{m}{n}
\rput(9.546406,-1.2198437){\huge $b_i$}
\usefont{T1}{ppl}{m}{n}
\rput(11.176406,-1.6198437){\huge $v_{i_2}$}
\usefont{T1}{ppl}{m}{n}
\rput(7.856406,-1.6198437){\huge $v_{i_4}$}
\psdots[dotsize=0.4](9.481875,-0.60984373)
\pscircle[linewidth=0.06,linecolor=color1703,dimen=outer,fillstyle=solid,fillcolor=color1703b](8.491875,-1.6998438){0.21}
\pscircle[linewidth=0.06,linecolor=color1703,dimen=outer,fillstyle=solid,fillcolor=color1703b](10.471875,-1.6998438){0.21}
\psarc[linewidth=0.04](7.771875,2.4001563){3.41}{-60.05122}{80.26301}
\psdots[dotsize=0.4](13.521875,0.73015624)
\usefont{T1}{ppl}{m}{n}
\rput(15.596406,1.9801563){\huge $v^{''}_i$}
\psline[linewidth=0.04cm,linestyle=dashed,dash=0.16cm 0.16cm](13.481875,2.7301562)(14.881875,1.8301562)
\psline[linewidth=0.04cm,linestyle=dashed,dash=0.16cm 0.16cm](14.881875,1.8301562)(13.541875,0.73015624)
\psline[linewidth=0.04cm](13.521875,0.71015626)(13.521875,2.7701561)
\usefont{T1}{ppl}{m}{n}
\rput(13.696406,3.2001562){\huge $y_i$}
\psdots[dotsize=0.4](13.541875,2.6901562)
\usefont{T1}{ppl}{m}{n}
\rput(13.956407,0.40015626){\huge $y'_i$}
\pscircle[linewidth=0.06,linecolor=color1703,dimen=outer,fillstyle=solid,fillcolor=color1703b](14.811875,1.8001562){0.21}
\usefont{T1}{ppl}{m}{n}
\rput(0.9164063,5.0201564){\huge $v'_i$}
\psline[linewidth=0.04cm](1.3927741,4.96797)(1.7927623,3.773144)
\psline[linewidth=0.04cm](1.4244376,4.936388)(2.7899532,5.3935175)
\psline[linewidth=0.04cm](1.8054603,3.735213)(2.7457547,5.3365393)
\psdots[dotsize=0.4,dotangle=18.50884](2.7647202,5.3428884)
\psdots[dotsize=0.4,dotangle=18.50884](1.8117278,3.779493)
\rput{18.50884}(1.6423894,-0.20688945){\pscircle[linewidth=0.06,linecolor=color1703,dimen=outer,fillstyle=solid,fillcolor=color1703b](1.4560604,4.936429){0.21}}
\usefont{T1}{ppl}{m}{n}
\rput(2.8764062,5.800156){\huge $x_i$}
\usefont{T1}{ppl}{m}{n}
\rput(1.7364062,3.3201563){\huge $x'_i$}
\psdots[dotsize=0.4](9.281875,4.7901564)
\psdots[dotsize=0.4](7.301875,4.7901564)
\psdots[dotsize=0.4](8.321875,5.7501564)
\psdots[dotsize=0.4](4.101875,2.7101562)
\usefont{T1}{ppl}{m}{n}
\rput(8.596406,6.200156){\huge $w_{i_{\mbox{\small $1$}}}$}
\psline[linewidth=0.04cm](8.341875,5.7501564)(7.341875,4.8501563)
\psline[linewidth=0.04cm](7.301875,4.7901564)(9.321875,4.7901564)
\psline[linewidth=0.04cm](8.281366,2.770167)(7.2623634,3.711248)
\psdots[dotsize=0.4,dotangle=-180.06073](7.3023844,3.7312055)
\psdots[dotsize=0.4,dotangle=-180.06073](9.282383,3.7291064)
\psdots[dotsize=0.4,dotangle=-180.06073](8.261366,2.7701883)
\psline[linewidth=0.04cm](8.241366,2.7702096)(9.24232,3.669149)
\psline[linewidth=0.04cm](9.282383,3.7291064)(7.2623844,3.731248)
\psline[linewidth=0.04cm](7.301875,4.7901564)(9.301875,3.8501563)
\psline[linewidth=0.04cm](9.281875,4.8101563)(7.321875,3.7901564)
\psline[linewidth=0.04cm](4.041875,2.6501563)(2.521875,-0.52984375)
\psline[linewidth=0.04cm](8.241875,2.7501562)(2.521875,-0.56984377)
\psline[linewidth=0.04cm](4.141875,2.6101563)(9.481875,-0.6298438)
\psline[linewidth=0.04cm](4.241875,2.6901562)(13.401875,0.7501562)
\psline[linewidth=0.04cm](8.201875,2.8101563)(13.501875,2.7101562)
\usefont{T1}{ppl}{m}{n}
\rput(4.036406,3.1401563){\huge $w_i$}
\usefont{T1}{ppl}{m}{n}
\rput(6.916406,4.5001564){\huge $w_{i_{\mbox{\small $2$}}}$}
\usefont{T1}{ppl}{m}{n}
\rput(9.776406,4.4601564){\huge $w_{i_{\mbox{\small $3$}}}$}
\usefont{T1}{ppl}{m}{n}
\rput(6.956406,3.3001564){\huge $w_{i_{\mbox{\small $4$}}}$}
\usefont{T1}{ppl}{m}{n}
\rput(9.716406,3.3001564){\huge $w_{i_{\mbox{\small $5$}}}$}
\usefont{T1}{ppl}{m}{n}
\rput(8.5364065,2.3601563){\huge $w_{i_{\mbox{\small $6$}}}$}
\usefont{T1}{ppl}{m}{n}
\rput(6.238125,-6.219844){\huge (b)}
\end{pspicture} 
}
\caption{(a) Cycle $C_i\in\mathcal{C}_{H_G}, C'_j \in \mathcal{C}_{H'_G}$ and $C_i,C'_j$ containing $v \in V$ where $H_G,H'_G$ is a 2-factoring of $G$.
(b) Gadget $W$, where $v$ is split into $v',v''$ and $v'$ connected to $x,x'$ and $v''$ connected to $y,y'$.One neighbour of $v$ in $H_G,H'_G$ is connected to $a$
and other neighbour to $b$ and double lines showing path in cycle $C_i,C'_j$}.
\label{W}
\end{figure}
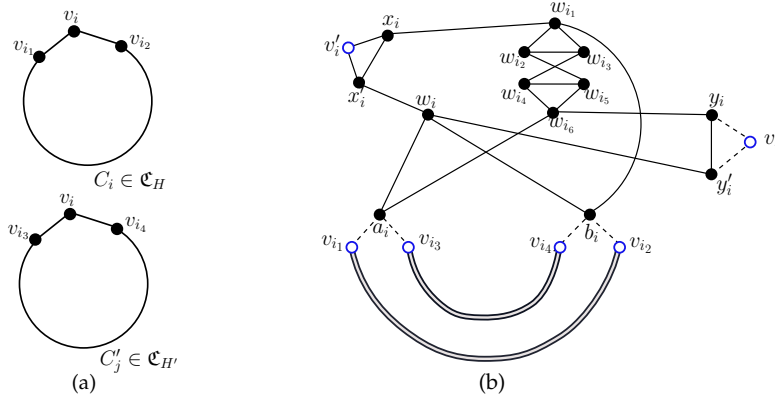

\paragraph{Gadgets used for reduction} The third gadget $W_i$ that we use for the reduction is shown in figure~\ref{W}[b] (indexed by $i$).
It is easy to see that it has exactly 2 paths from $v'_i$ to $v''_i$ [refer figure~\ref{path1}(a),~\ref{path1}(b)], where each path
covers all vertices of $W_i$ and all the path created by deletion of vertex from one of the cycles. Note other than the edges covered in this two paths we do not have any other edge in $W_i$.

We create a new graph $G_F$ as follows. Initially we have $G_F=\tilde{G}$.
As $\tilde{G}$ is a 4-regular graph therefore $v_i \in V_H \subseteq V(\tilde{G})$ has four neighbors say $v_{i_1},v_{i_2},v_{i_3}$ and $v_{i_4}$ and
let $v_{i_1},v_{i_2}$ be neighbor of $v_i$ in cycle $C_i \in \mathfrak{C}_H$
and $v_{i_3},v_{i_4}$ be neighbors of $v_i$ in cycle $C'_j \in \mathfrak{C}_{H'}$.
We delete $v_i$ from the graph $G_F$ and insert $W_i$ and add edges $(v_{i_1},a_i),(v_{i_3},a_i),(v_{i_2},b_i)$ and $(v_{i_4},b_i)$. Since in $\tilde{G}$
if at least one of $v_{i_1},v_{i_2},v_{i_3},v_{i_4}$ is not chosen in vertex cover then we have to include $v_i$ in the vertex cover. So by $W_i$ we ensure that if at least one of
$v_{i_1},v_{i_2},v_{i_3}v_{i_4}$ is not chosen in the vertex cover then the size of vertex cover increases exactly by one indicating in vertex cover for $\tilde{G}$ we need to include vertex $v_i$
[refer lemma~\ref{W1vc}].


\begin{lemma}
We can cover all edges of $W_i$ with $V'\subseteq V(W_i)$ s.t. $a_i,b_i \notin V'$ and $\lvert V' \rvert=9$ and
if at least one of $a_i,b_i$ is in $V'$ then we can cover all edges with a $V'$ s.t.
$\lvert V' \rvert = 10$ and this is minimum required.
\label{W1vc}
\end{lemma}
\begin{proof}
The vertex sets $\{v'_i,x_i,x'_i\}$, $\{v''_i,y_i,y'_i\}$, $\{w_{i_1},w_{i_2},w_{i_3}\}$, $\{w_{i_4},w_{i_5},w_{i_6}\}$ form triangles
and $(w_i,a_i)$ is an edge, therefore we require minimum of 9 vertices to cover edges of $W_i$. If we have $V'=\{x_i,x'_i,y_i,y'_i,w_{i_1},w_{i_2},w_{i_4},w_{i_6},w_i\}$ then we can cover all edges of $W_i$ with 9 vertices.

If at least one of $a_i$ or $b_i$ is forced due to one of $v_{i_1},v_{i_3}$ or $v_{i_2},v_{i_4}$ not chosen respectively, say $a_i$ is forced then apart from the $K_3$'s present we have an 
edge $(w_i,b_i)$, so including $a_i$ we need at least 10 vertices,
similar argument can be given when $b_i$ is forced (edge $(w_i,a_i)$ left).
If we have $V'=\{x_i,x'_i,y_i,y'_i,w_{i_1},w_{i_2},w_{i_4},w_{i_6},a_i,b_i\}$ then we can cover all edges of $W_i$ with 10 vertices.
\qed
\end{proof}
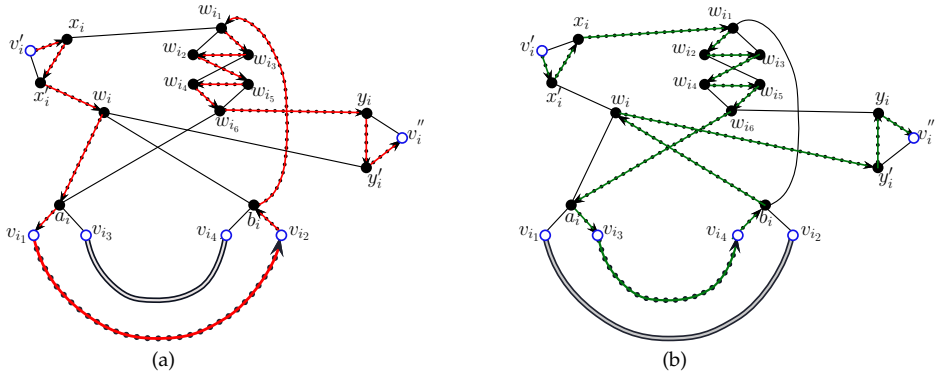
\begin{figure}
 \centering
\scalebox{0.37} 
{
\begin{pspicture}(0,-6.622031)(17.909063,6.622031)
\definecolor{color2486}{rgb}{0.12549019607843137,0.13725490196078433,0.18823529411764706}
\definecolor{color2498}{rgb}{0.054901960784313725,0.07450980392156863,0.12941176470588237}
\definecolor{color2498d}{rgb}{0.9098039215686274,0.8901960784313725,0.8901960784313725}
\definecolor{color2512b}{rgb}{0.996078431372549,0.996078431372549,0.996078431372549}
\definecolor{color2512}{rgb}{0.07450980392156863,0.07450980392156863,0.9176470588235294}
\psline[linewidth=0.04cm,linestyle=dotted,dotsep=0.16cm,arrowsize=0.013cm 2.0,arrowlength=1.4,arrowinset=0.4,doubleline=true,doublesep=0.06,doublecolor=red]{->}(10.022384,3.6691065)(8.002384,3.6712477)
\psline[linewidth=0.04cm,linestyle=dotted,dotsep=0.16cm,arrowsize=0.013cm 2.0,arrowlength=1.4,arrowinset=0.4,doubleline=true,doublesep=0.06,doublecolor=red]{->}(9.001875,5.6901565)(10.021875,4.7501564)
\psdots[dotsize=0.4](9.981875,4.7301564)
\psline[linewidth=0.04cm,linestyle=dotted,dotsep=0.16cm,arrowsize=0.05291667cm 2.0,arrowlength=1.4,arrowinset=0.4,doubleline=true,doublesep=0.06,doublecolor=red]{->}(8.001875,4.7301564)(10.001875,3.7901564)
\psdots[dotsize=0.4](8.001875,4.7301564)
\psline[linewidth=0.04cm](9.041875,5.6901565)(8.041875,4.7901564)
\psline[linewidth=0.04cm,linestyle=dotted,dotsep=0.16cm,arrowsize=0.013cm 2.0,arrowlength=1.4,arrowinset=0.4,doubleline=true,doublesep=0.06,doublecolor=red]{<-}(8.081875,4.7101564)(9.901875,4.7301564)
\psline[linewidth=0.04cm,linestyle=dotted,dotsep=0.16cm,arrowsize=0.05291667cm 2.0,arrowlength=1.4,arrowinset=0.4,doubleline=true,doublesep=0.06,doublecolor=red]{<-}(8.981366,2.7101672)(7.9623632,3.6512477)
\psdots[dotsize=0.4,dotangle=-180.06073](8.002384,3.6712055)
\psdots[dotsize=0.4,dotangle=-180.06073](9.982384,3.6691065)
\psline[linewidth=0.04cm,linestyle=dotted,dotsep=0.16cm,arrowsize=0.013cm 2.0,arrowlength=1.4,arrowinset=0.4,doubleline=true,doublesep=0.06,doublecolor=red]{->}(8.921875,2.7501562)(14.221875,2.6501563)
\psdots[dotsize=0.4,dotangle=-180.06073](8.961367,2.7101884)
\psline[linewidth=0.04cm](8.941366,2.7102096)(9.94232,3.609149)
\psline[linewidth=0.04cm](9.981875,4.7501564)(8.021875,3.7301562)
\psline[linewidth=0.04cm,linestyle=dotted,dotsep=0.16cm,arrowsize=0.013cm 2.0,arrowlength=1.4,arrowinset=0.4,doubleline=true,doublesep=0.06,doublecolor=red]{->}(2.6645117,3.656866)(4.741875,2.7101562)
\psline[linewidth=0.04cm](3.341875,5.3301563)(8.981875,5.7101564)
\psbezier[linewidth=0.04,linestyle=dotted,dotsep=0.16cm,doubleline=true,doublesep=0.06,doublecolor=red,arrowsize=0.013cm 2.0,arrowlength=1.4,arrowinset=0.4]{<-}(9.121875,5.7922497)(9.598293,6.2501564)(10.364156,5.825031)(10.755308,5.0863094)(11.14646,4.347588)(11.3295145,3.9531326)(11.35083,2.9875693)(11.372146,2.022006)(11.401875,2.4915035)(11.361875,1.5701562)(11.321875,0.6488091)(11.201875,0.37015626)(11.081875,0.09015625)(10.961875,-0.18984374)(10.722219,-0.5893087)(10.101875,-0.7498438)
\psbezier[linewidth=0.06,linecolor=color2486,linestyle=dotted,dotsep=0.16cm,doubleline=true,doublesep=0.1,doublecolor=red,arrowsize=0.05291667cm 2.0,arrowlength=1.4,arrowinset=0.4]{->}(2.281875,-1.8740149)(2.5243688,-2.9307108)(3.0975358,-3.7908127)(3.9793313,-4.4788938)(4.861127,-5.166975)(5.621875,-5.469844)(6.681875,-5.469844)(7.741875,-5.469844)(8.227849,-5.321291)(9.181875,-4.7298436)(10.135901,-4.1383967)(10.967561,-2.7238514)(11.121875,-1.8298438)
\usefont{T1}{ppl}{m}{n}
\rput(11.816406,-1.6798438){\huge $v_{i_2}$}
\usefont{T1}{ppl}{m}{n}
\rput(10.186406,-1.2198437){\huge $b_i$}
\usefont{T1}{ppl}{m}{n}
\rput(8.496407,-1.6998438){\huge $v_{i_4}$}
\psdots[dotsize=0.4](10.181875,-0.66984373)
\psbezier[linewidth=0.06,linecolor=color2498,doubleline=true,doublesep=0.08,doublecolor=color2498d](4.201875,-1.9350139)(4.337935,-2.6078036)(4.6595306,-3.1554232)(5.154294,-3.5935187)(5.649057,-4.0316143)(5.9460387,-4.078555)(6.5643687,-4.094199)(7.1826987,-4.1098437)(7.717625,-4.015261)(8.246563,-3.6561038)(8.775501,-3.2969463)(9.075292,-2.3887558)(9.161875,-1.8098438)
\psline[linewidth=0.04cm,linestyle=dotted,dotsep=0.16cm,arrowsize=0.013cm 2.0,arrowlength=1.4,arrowinset=0.4,doubleline=true,doublesep=0.06,doublecolor=red]{->}(3.121875,-0.7498438)(2.301875,-1.6298437)
\psline[linewidth=0.04cm](3.321875,-0.78984374)(4.201875,-1.7498437)
\usefont{T1}{ppl}{m}{n}
\rput(3.3164062,-1.1398437){\huge $a_i$}
\usefont{T1}{ppl}{m}{n}
\rput(1.6564063,-1.7598437){\huge $v_{i_1}$}
\usefont{T1}{ppl}{m}{n}
\rput(4.7364063,-1.6598438){\huge $v_{i_3}$}
\psdots[dotsize=0.4](3.201875,-0.66984373)
\pscircle[linewidth=0.06,linecolor=color2512,dimen=outer,fillstyle=solid,fillcolor=color2512b](2.271875,-1.7598437){0.21}
\pscircle[linewidth=0.06,linecolor=color2512,dimen=outer,fillstyle=solid,fillcolor=color2512b](4.151875,-1.7598437){0.21}
\psline[linewidth=0.04cm](10.041875,-0.8098438)(8.981875,-1.8898437)
\psline[linewidth=0.04cm,linestyle=dotted,dotsep=0.16cm,arrowsize=0.05291667cm 2.0,arrowlength=1.4,arrowinset=0.4,doubleline=true,doublesep=0.06,doublecolor=red]{<-}(10.181875,-0.7498438)(11.161875,-1.6498437)
\pscircle[linewidth=0.06,linecolor=color2512,dimen=outer,fillstyle=solid,fillcolor=color2512b](9.191875,-1.7598437){0.21}
\pscircle[linewidth=0.06,linecolor=color2512,dimen=outer,fillstyle=solid,fillcolor=color2512b](11.171875,-1.7598437){0.21}
\psline[linewidth=0.04cm,linestyle=dotted,dotsep=0.16cm,arrowsize=0.05291667cm 2.0,arrowlength=1.4,arrowinset=0.4,doubleline=true,doublesep=0.06,doublecolor=red]{<-}(15.501875,1.6701562)(14.281875,0.71015626)
\psdots[dotsize=0.4](14.221875,0.67015624)
\usefont{T1}{ppl}{m}{n}
\rput(16.006407,1.9601562){\huge $v^{''}_i$}
\psline[linewidth=0.04cm](14.181875,2.6701562)(15.581875,1.7701563)
\psline[linewidth=0.04cm,linestyle=dotted,dotsep=0.16cm,arrowsize=0.013cm 2.0,arrowlength=1.4,arrowinset=0.4,doubleline=true,doublesep=0.06,doublecolor=red]{<-}(14.221875,0.65015626)(14.221875,2.7101562)
\usefont{T1}{ppl}{m}{n}
\rput(14.156406,3.1201563){\huge $y_i$}
\psdots[dotsize=0.4](14.241875,2.6301563)
\usefont{T1}{ppl}{m}{n}
\rput(14.5364065,0.30015624){\huge $y'_i$}
\pscircle[linewidth=0.06,linecolor=color2512,dimen=outer,fillstyle=solid,fillcolor=color2512b](15.511875,1.7401563){0.21}
\usefont{T1}{ppl}{m}{n}
\rput(1.6564063,5.1201563){\huge $v'_i$}
\psline[linewidth=0.04cm](2.0927742,4.90797)(2.4927623,3.713144)
\psline[linewidth=0.04cm,linestyle=dotted,dotsep=0.16cm,arrowsize=0.013cm 1.97,arrowlength=1.38,arrowinset=0.4,doubleline=true,doublesep=0.06,doublecolor=red]{->}(2.1844375,4.836388)(3.401875,5.2701564)
\psline[linewidth=0.04cm,linestyle=dotted,dotsep=0.16cm,arrowsize=0.013cm 2.0,arrowlength=1.4,arrowinset=0.4,doubleline=true,doublesep=0.06,doublecolor=red]{<-}(2.581875,3.8101563)(3.421875,5.1701565)
\psdots[dotsize=0.4,dotangle=18.50884](3.4647202,5.2828884)
\psdots[dotsize=0.4,dotangle=18.50884](2.5117278,3.719493)
\rput{18.50884}(1.6595501,-0.43220866){\pscircle[linewidth=0.06,linecolor=color2512,dimen=outer,fillstyle=solid,fillcolor=color2512b](2.1560605,4.876429){0.21}}
\usefont{T1}{ppl}{m}{n}
\rput(3.8164062,5.780156){\huge $x_i$}
\usefont{T1}{ppl}{m}{n}
\rput(2.5564063,3.2401562){\huge $x'_i$}
\psdots[dotsize=0.4](9.021875,5.6901565)
\psdots[dotsize=0.4](4.801875,2.6501563)
\usefont{T1}{ppl}{m}{n}
\rput(8.616406,6.260156){\huge $w_{i_{\mbox{\small $1$}}}$}
\psline[linewidth=0.04cm,linestyle=dotted,dotsep=0.16cm,arrowsize=0.013cm 2.0,arrowlength=1.4,arrowinset=0.4,doubleline=true,doublesep=0.06,doublecolor=red]{->}(4.741875,2.5901563)(3.221875,-0.58984375)
\psline[linewidth=0.04cm](8.941875,2.6901562)(3.221875,-0.6298438)
\psline[linewidth=0.04cm](4.841875,2.5501564)(10.181875,-0.6898438)
\psline[linewidth=0.04cm](4.941875,2.6301563)(14.101875,0.6901562)
\usefont{T1}{ppl}{m}{n}
\rput(4.7364063,3.2201562){\huge $w_i$}
\usefont{T1}{ppl}{m}{n}
\rput(7.336406,4.880156){\huge $w_{i_{\mbox{\small $2$}}}$}
\usefont{T1}{ppl}{m}{n}
\rput(7.376406,3.6801562){\huge $w_{i_{\mbox{\small $4$}}}$}
\usefont{T1}{ppl}{m}{n}
\rput(10.516406,3.3001564){\huge $w_{i_{\mbox{\small $5$}}}$}
\usefont{T1}{ppl}{m}{n}
\rput(10.5764065,4.4601564){\huge $w_{i_{\mbox{\small $3$}}}$}
\usefont{T1}{ppl}{m}{n}
\rput(9.236406,2.2201562){\huge $w_{i_{\mbox{\small $6$}}}$}
\usefont{T1}{ppl}{m}{n}
\rput(6.938125,-6.279844){\huge (a)}
\end{pspicture} 
}
\scalebox{0.37} 
{
\begin{pspicture}(0,-6.602031)(17.829063,6.602031)
\definecolor{color2998d}{rgb}{0.0,0.47058823529411764,0.058823529411764705}
\definecolor{color3017}{rgb}{0.12549019607843137,0.13725490196078433,0.18823529411764706}
\definecolor{color3017d}{rgb}{0.788235294117647,0.788235294117647,0.788235294117647}
\definecolor{color3029}{rgb}{0.054901960784313725,0.07450980392156863,0.12941176470588237}
\definecolor{color3043b}{rgb}{0.996078431372549,0.996078431372549,0.996078431372549}
\definecolor{color3043}{rgb}{0.07450980392156863,0.07450980392156863,0.9176470588235294}
\psline[linewidth=0.04cm,linestyle=dotted,dotsep=0.16cm,arrowsize=0.013cm 2.0,arrowlength=1.4,arrowinset=0.4,doubleline=true,doublesep=0.06,doublecolor=color2998d]{<-}(9.901875,3.6301563)(8.062385,3.6512477)
\psline[linewidth=0.04cm](8.961875,5.7101564)(9.981875,4.7701564)
\psdots[dotsize=0.4](9.941875,4.7501564)
\psline[linewidth=0.04cm](7.961875,4.7501564)(9.961875,3.8101563)
\psdots[dotsize=0.4](7.961875,4.7501564)
\psline[linewidth=0.04cm,linestyle=dotted,dotsep=0.16cm,arrowsize=0.05291667cm 2.0,arrowlength=1.4,arrowinset=0.4,doubleline=true,doublesep=0.06,doublecolor=color2998d]{->}(9.001875,5.7101564)(8.001875,4.8101563)
\psline[linewidth=0.04cm,linestyle=dotted,dotsep=0.16cm,arrowsize=0.013cm 2.0,arrowlength=1.4,arrowinset=0.4,doubleline=true,doublesep=0.06,doublecolor=color2998d]{->}(8.041875,4.7301564)(9.861875,4.7501564)
\psline[linewidth=0.04cm](8.941366,2.7301672)(7.9223633,3.6712477)
\psdots[dotsize=0.4,dotangle=-180.06073](7.962384,3.6912055)
\psdots[dotsize=0.4,dotangle=-180.06073](9.942383,3.6891065)
\psline[linewidth=0.04cm](8.881875,2.7701561)(14.181875,2.6701562)
\psdots[dotsize=0.4,dotangle=-180.06073](8.921366,2.7301884)
\psline[linewidth=0.04cm,linestyle=dotted,dotsep=0.16cm,arrowsize=0.013cm 2.0,arrowlength=1.4,arrowinset=0.4,doubleline=true,doublesep=0.06,doublecolor=color2998d]{<-}(8.901366,2.7302096)(9.90232,3.629149)
\psline[linewidth=0.04cm,linestyle=dotted,dotsep=0.16cm,arrowsize=0.013cm 2.0,arrowlength=1.4,arrowinset=0.4,doubleline=true,doublesep=0.06,doublecolor=color2998d]{->}(9.941875,4.7701564)(7.981875,3.7501562)
\psline[linewidth=0.04cm](2.6245115,3.676866)(4.701875,2.7301562)
\psline[linewidth=0.04cm,linestyle=dotted,dotsep=0.16cm,arrowsize=0.05291667cm 2.0,arrowlength=1.4,arrowinset=0.4,doubleline=true,doublesep=0.06,doublecolor=color2998d]{->}(3.301875,5.3501563)(8.941875,5.7301564)
\psbezier[linewidth=0.04](9.081875,5.8122497)(9.558293,6.2701564)(10.324156,5.845031)(10.715308,5.1063094)(11.10646,4.367588)(11.289515,3.9731326)(11.31083,3.0075693)(11.332146,2.042006)(11.361875,2.5115035)(11.321875,1.5901562)(11.281875,0.66880906)(11.161875,0.39015624)(11.041875,0.11015625)(10.921875,-0.16984375)(10.6822195,-0.5693087)(10.061875,-0.72984374)
\psbezier[linewidth=0.06,linecolor=color3017,doubleline=true,doublesep=0.08,doublecolor=color3017d](2.241875,-1.8540149)(2.4843688,-2.9107108)(3.057536,-3.7708127)(3.9393313,-4.458894)(4.821127,-5.146975)(5.581875,-5.449844)(6.641875,-5.449844)(7.701875,-5.449844)(8.187849,-5.3012905)(9.141875,-4.7098436)(10.0959015,-4.1183968)(10.927561,-2.7038515)(11.081875,-1.8098438)
\usefont{T1}{ppl}{m}{n}
\rput(11.796406,-1.6798438){\huge $v_{i_2}$}
\usefont{T1}{ppl}{m}{n}
\rput(10.166407,-1.0998437){\huge $b_i$}
\usefont{T1}{ppl}{m}{n}
\rput(8.396406,-1.6798438){\huge $v_{i_4}$}
\psdots[dotsize=0.4](10.141875,-0.64984375)
\psbezier[linewidth=0.06,linecolor=color3029,linestyle=dotted,dotsep=0.16cm,doubleline=true,doublesep=0.08,doublecolor=color2998d,arrowsize=0.013cm 2.0,arrowlength=1.4,arrowinset=0.4]{->}(4.161875,-1.9150134)(4.297935,-2.5878031)(4.6195307,-3.135423)(5.114294,-3.5735188)(5.609057,-4.0116143)(5.9060388,-4.058555)(6.524369,-4.074199)(7.142699,-4.0898438)(7.677625,-3.9952612)(8.206563,-3.6361036)(8.735501,-3.2769463)(9.035292,-2.3687558)(9.121875,-1.7898438)
\psline[linewidth=0.04cm](3.081875,-0.72984374)(2.261875,-1.6098437)
\psline[linewidth=0.04cm,linestyle=dotted,dotsep=0.16cm,arrowsize=0.05291667cm 2.0,arrowlength=1.4,arrowinset=0.4,doubleline=true,doublesep=0.06,doublecolor=color2998d]{->}(3.281875,-0.76984376)(4.161875,-1.7298437)
\usefont{T1}{ppl}{m}{n}
\rput(3.2164063,-1.1598438){\huge $a_i$}
\usefont{T1}{ppl}{m}{n}
\rput(1.6564063,-1.6798438){\huge $v_{i_1}$}
\usefont{T1}{ppl}{m}{n}
\rput(4.6764064,-1.6198437){\huge $v_{i_3}$}
\psdots[dotsize=0.4](3.161875,-0.64984375)
\pscircle[linewidth=0.06,linecolor=color3043,dimen=outer,fillstyle=solid,fillcolor=color3043b](2.231875,-1.7398437){0.21}
\pscircle[linewidth=0.06,linecolor=color3043,dimen=outer,fillstyle=solid,fillcolor=color3043b](4.111875,-1.7398437){0.21}
\psline[linewidth=0.04cm,linestyle=dotted,dotsep=0.16cm,arrowsize=0.013cm 2.0,arrowlength=1.4,arrowinset=0.4,doubleline=true,doublesep=0.06,doublecolor=color2998d]{<-}(10.001875,-0.78984374)(8.941875,-1.8698437)
\psline[linewidth=0.04cm](10.141875,-0.72984374)(11.121875,-1.6298437)
\pscircle[linewidth=0.06,linecolor=color3043,dimen=outer,fillstyle=solid,fillcolor=color3043b](9.151875,-1.7398437){0.21}
\pscircle[linewidth=0.06,linecolor=color3043,dimen=outer,fillstyle=solid,fillcolor=color3043b](11.131875,-1.7398437){0.21}
\psline[linewidth=0.04cm](15.461875,1.6901562)(14.241875,0.73015624)
\psdots[dotsize=0.4](14.181875,0.6901562)
\usefont{T1}{ppl}{m}{n}
\rput(15.926406,1.9801563){\huge $v^{''}_i$}
\psline[linewidth=0.04cm,linestyle=dotted,dotsep=0.16cm,arrowsize=0.013cm 2.0,arrowlength=1.4,arrowinset=0.4,doubleline=true,doublesep=0.06,doublecolor=color2998d]{->}(14.141875,2.6901562)(15.541875,1.7901562)
\psline[linewidth=0.04cm,linestyle=dotted,dotsep=0.16cm,arrowsize=0.013cm 2.0,arrowlength=1.4,arrowinset=0.4,doubleline=true,doublesep=0.06,doublecolor=color2998d]{->}(14.181875,0.67015624)(14.181875,2.7301562)
\usefont{T1}{ppl}{m}{n}
\rput(14.356406,3.1601562){\huge $y_i$}
\psdots[dotsize=0.4](14.201875,2.6501563)
\usefont{T1}{ppl}{m}{n}
\rput(14.476406,0.34015626){\huge $y'_i$}
\pscircle[linewidth=0.06,linecolor=color3043,dimen=outer,fillstyle=solid,fillcolor=color3043b](15.471875,1.7601563){0.21}
\usefont{T1}{ppl}{m}{n}
\rput(1.5964062,5.0801563){\huge $v'_i$}
\psline[linewidth=0.04cm,linestyle=dotted,dotsep=0.16cm,arrowsize=0.05291667cm 2.0,arrowlength=1.4,arrowinset=0.4,doubleline=true,doublesep=0.06,doublecolor=color2998d]{->}(2.0527742,4.92797)(2.4527624,3.733144)
\psline[linewidth=0.04cm,arrowsize=0.013cm 1.97,arrowlength=1.38,arrowinset=0.4]{->}(2.1444376,4.856388)(3.361875,5.2901564)
\psline[linewidth=0.04cm,linestyle=dotted,dotsep=0.16cm,arrowsize=0.013cm 2.0,arrowlength=1.4,arrowinset=0.4,doubleline=true,doublesep=0.06,doublecolor=color2998d]{->}(2.541875,3.8301563)(3.381875,5.1901565)
\psdots[dotsize=0.4,dotangle=18.50884](3.4247203,5.3028884)
\psdots[dotsize=0.4,dotangle=18.50884](2.4717278,3.739493)
\rput{18.50884}(1.66383,-0.4184761){\pscircle[linewidth=0.06,linecolor=color3043,dimen=outer,fillstyle=solid,fillcolor=color3043b](2.1160603,4.896429){0.21}}
\usefont{T1}{ppl}{m}{n}
\rput(3.5364063,5.880156){\huge $x_i$}
\usefont{T1}{ppl}{m}{n}
\rput(2.5564063,3.2401562){\huge $x'_i$}
\psdots[dotsize=0.4](8.981875,5.7101564)
\psdots[dotsize=0.4](4.761875,2.6701562)
\usefont{T1}{ppl}{m}{n}
\rput(8.516406,6.240156){\huge $w_i{_{\mbox{\small $1$}}}$}
\psline[linewidth=0.04cm](4.701875,2.6101563)(3.181875,-0.56984377)
\psline[linewidth=0.04cm,linestyle=dotted,dotsep=0.16cm,arrowsize=0.05291667cm 2.0,arrowlength=1.4,arrowinset=0.4,doubleline=true,doublesep=0.06,doublecolor=color2998d]{->}(8.901875,2.7101562)(3.181875,-0.60984373)
\psline[linewidth=0.04cm,linestyle=dotted,dotsep=0.16cm,arrowsize=0.05291667cm 2.0,arrowlength=1.4,arrowinset=0.4,doubleline=true,doublesep=0.06,doublecolor=color2998d]{<-}(4.801875,2.5701563)(10.141875,-0.66984373)
\psline[linewidth=0.04cm,linestyle=dotted,dotsep=0.16cm,arrowsize=0.013cm 2.0,arrowlength=1.4,arrowinset=0.4,doubleline=true,doublesep=0.06,doublecolor=color2998d]{->}(4.901875,2.6501563)(14.061875,0.71015626)
\usefont{T1}{ppl}{m}{n}
\rput(4.936406,3.1001563){\huge $w_i$}
\usefont{T1}{ppl}{m}{n}
\rput(7.2164063,4.9201565){\huge $w_i{_{\mbox{\small $2$}}}$}
\usefont{T1}{ppl}{m}{n}
\rput(7.2564063,3.7201562){\huge $w_i{_{\mbox{\small $4$}}}$}
\usefont{T1}{ppl}{m}{n}
\rput(10.416407,4.5401564){\huge $w_i{_{\mbox{\small $3$}}}$}
\usefont{T1}{ppl}{m}{n}
\rput(10.356406,3.3801563){\huge $w_i{_{\mbox{\small $5$}}}$}
\usefont{T1}{ppl}{m}{n}
\rput(9.316406,2.2201562){\huge $w_i{_{\mbox{\small $6$}}}$}
\usefont{T1}{ppl}{m}{n}
\rput(6.888125,-6.259844){\huge (b)}
\end{pspicture} 
}

\caption{Edge Set of gadget $W$ is union of two paths and (a),(b) showing two paths from $v'$ to $v''$ covering all vertices and edges of $W_i$.}
\label{path1}
\end{figure}

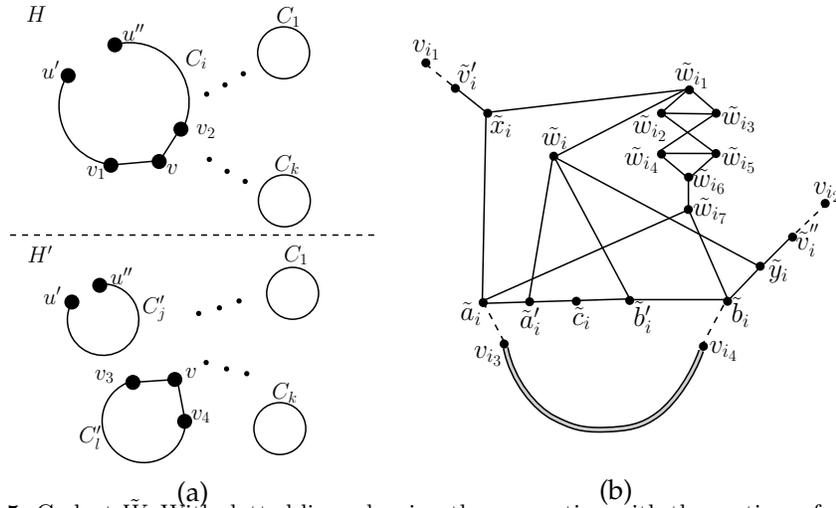
\begin{figure}
 \centering
\scalebox{0.5} 
{
\begin{pspicture}(0,-6.1427555)(9.72,6.1827555)
\psline[linewidth=0.04cm,linestyle=dashed,dash=0.16cm 0.16cm](0.02,-0.013181874)(9.7,-0.013181874)
\usefont{T1}{ptm}{m}{n}
\rput(0.82109374,-0.51318187){\LARGE $H'$}
\pscircle[linewidth=0.04,dimen=outer](7.53,-1.5631819){0.71}
\usefont{T1}{ptm}{m}{n}
\rput(7.6210938,-0.5531819){\LARGE $C_1$}
\pscircle[linewidth=0.04,dimen=outer](7.19,-5.163182){0.71}
\usefont{T1}{ptm}{m}{n}
\rput(7.3210936,-4.213182){\LARGE $C_k$}
\rput{-149.27724}(5.7867513,-2.9366221){\psarc[linewidth=0.04](2.49,-2.263182){0.95}{299.6994}{245.24171}}
\psdots[dotsize=0.4](1.66,-1.7931819)
\psdots[dotsize=0.4](2.4,-1.3531818)
\usefont{T1}{ptm}{m}{n}
\rput(1.1410937,-1.6131818){\LARGE $u'$}
\usefont{T1}{ptm}{m}{n}
\rput(2.9810936,-1.0931818){\LARGE $u''$}
\rput{-220.47987}(3.0692306,-11.053431){\psarc[linewidth=0.04](3.5724254,-4.960872){1.1076871}{325.9411}{221.7529}}
\psdots[dotsize=0.4](3.28,-3.9331818)
\psdots[dotsize=0.4](4.66,-4.9731817)
\psdots[dotsize=0.4](4.4,-3.8531818)
\psline[linewidth=0.04](3.28,-3.9331818)(4.46,-3.8531818)(4.68,-4.9931817)
\usefont{T1}{ptm}{m}{n}
\rput(4.831094,-3.6731818){\LARGE $v$}
\usefont{T1}{ptm}{m}{n}
\rput(5.0910935,-4.833182){\LARGE $v_4$}
\usefont{T1}{ptm}{m}{n}
\rput(2.5310938,-3.773182){\LARGE $v_3$}
\psdots[dotsize=0.14,dotangle=65.573395](5.0345883,-2.0963824)
\psdots[dotsize=0.14,dotangle=65.573395](5.5840745,-1.9505574)
\psdots[dotsize=0.14,dotangle=65.573395](6.1054115,-1.7699814)
\psdots[dotsize=0.14,dotangle=33.30056](5.23016,-3.4052866)
\psdots[dotsize=0.14,dotangle=33.30056](5.7726226,-3.5753882)
\psdots[dotsize=0.14,dotangle=33.30056](6.3098397,-3.7010772)
\usefont{T1}{ptm}{m}{n}
\rput(3.8510938,-1.9731818){\LARGE $C'_j$}
\usefont{T1}{ptm}{m}{n}
\rput(2.1510937,-5.393182){\LARGE $C'_l$}
\rput{-44.822815}(-1.5558413,3.2810142){\psarc[linewidth=0.04](3.2,3.526818){1.58}{16.750326}{149.94682}}
\psdots[dotsize=0.4](4.56,2.806818)
\psdots[dotsize=0.4](3.98,1.9468181)
\psdots[dotsize=0.4](2.7,1.8468181)
\psline[linewidth=0.04](2.66,1.8468181)(4.0,1.9668181)(4.56,2.8468182)
\psdots[dotsize=0.4](2.8,5.0468183)
\rput{-44.822815}(-1.5725626,3.0404682){\psarc[linewidth=0.04](2.9,3.4268181){1.58}{193.12138}{306.17026}}
\psdots[dotsize=0.4](1.56,4.226818)
\usefont{T1}{ptm}{m}{n}
\rput(1.1010938,4.386818){\LARGE $u'$}
\usefont{T1}{ptm}{m}{n}
\rput(3.3010938,5.386818){\LARGE $u''$}
\usefont{T1}{ptm}{m}{n}
\rput(4.331094,1.7468182){\LARGE $v$}
\usefont{T1}{ptm}{m}{n}
\rput(2.3110938,1.6468182){\LARGE $v_1$}
\usefont{T1}{ptm}{m}{n}
\rput(5.231094,2.746818){\LARGE $v_2$}
\psdots[dotsize=0.14,dotangle=-19.402328](6.218511,4.280052)
\psdots[dotsize=0.14,dotangle=-19.402328](5.732788,3.984632)
\psdots[dotsize=0.14,dotangle=-19.402328](5.241489,3.7335842)
\pscircle[linewidth=0.04,dimen=outer](7.31,0.8968181){0.71}
\usefont{T1}{ptm}{m}{n}
\rput(4.9710937,4.666818){\LARGE $C_i$}
\usefont{T1}{ptm}{m}{n}
\rput(7.381094,1.8668181){\LARGE $C_k$}
\pscircle[linewidth=0.04,dimen=outer](7.29,4.836818){0.71}
\usefont{T1}{ptm}{m}{n}
\rput(7.421094,5.786818){\LARGE $C_1$}
\psdots[dotsize=0.14,dotangle=17.075563](5.3129845,1.9996598)
\psdots[dotsize=0.14,dotangle=17.075563](5.7863135,1.6847634)
\psdots[dotsize=0.14,dotangle=17.075563](6.2670155,1.4139766)
\usefont{T1}{ptm}{m}{n}
\rput(0.7210938,5.866818){\LARGE $H$}
\usefont{T1}{ppl}{m}{n}
\rput(4.878125,-6.909844){\huge (a)}
\end{pspicture} 
}
\scalebox{0.5} 
{
\begin{pspicture}(0,-5.7120314)(12.689062,5.7120314)
\definecolor{color2862d}{rgb}{0.8117647058823529,0.8117647058823529,0.8117647058823529}
\psbezier[linewidth=0.04,doubleline=true,doublesep=0.08,doublecolor=color2862d](3.041875,-2.5798438)(3.2171307,-3.4798439)(3.8410108,-4.199844)(4.308359,-4.4398437)(4.7757077,-4.679844)(5.0483274,-4.779844)(5.6714587,-4.759844)(6.29459,-4.739844)(6.664574,-4.679844)(7.1124496,-4.3398438)(7.560325,-3.9998438)(8.027674,-3.2798438)(8.241875,-2.6598437)
\psdots[dotsize=0.24](2.441875,-1.3798437)
\psdots[dotsize=0.24](3.001875,-2.4798439)
\psline[linewidth=0.04cm,linestyle=dashed,dash=0.16cm 0.16cm](2.441875,-1.3598437)(3.001875,-2.4198437)
\psdots[dotsize=0.24](8.941875,-1.3398438)
\psdots[dotsize=0.24](8.301875,-2.5398438)
\psline[linewidth=0.04cm](2.421875,-1.3398438)(7.921875,1.1001563)
\usefont{T1}{ppl}{m}{n}
\rput(2.5764062,-2.7898438){\huge $v_{i_3}$}
\usefont{T1}{ppl}{m}{n}
\rput(8.836407,-2.6698437){\huge $v_{i_4}$}
\psline[linewidth=0.04cm,linestyle=dashed,dash=0.16cm 0.16cm](8.841875,-1.4398438)(8.241875,-2.5598438)
\usefont{T1}{ppl}{m}{n}
\rput(2.1164062,-1.6698438){\huge $\tilde{a}_i$}
\usefont{T1}{ppl}{m}{n}
\rput(9.286406,-1.6298438){\huge $\tilde{b}_i$}
\psdots[dotsize=0.24](1.701875,4.320156)
\psdots[dotsize=0.24](2.561875,3.6601562)
\psline[linewidth=0.04cm](1.681875,4.360156)(2.481875,3.7201562)
\psline[linewidth=0.04cm,linestyle=dashed,dash=0.16cm 0.16cm](1.621875,4.400156)(0.901875,4.9801564)
\psdots[dotsize=0.24](0.921875,5.0001564)
\usefont{T1}{ppl}{m}{n}
\rput(0.93640625,5.3501563){\huge $v_{i_1}$}
\usefont{T1}{ppl}{m}{n}
\rput(2.0164063,4.6701565){\huge $\tilde{v}_i'$}
\usefont{T1}{ppl}{m}{n}
\rput(2.9164062,3.3101562){\huge $\tilde{x}_i$}
\psline[linewidth=0.04cm](2.541875,3.6801562)(7.921875,4.300156)
\psdots[dotsize=0.24,dotangle=-98.19838](10.670655,0.36546034)
\psline[linewidth=0.04cm](10.681875,0.42015624)(9.828979,-0.4342195)
\usefont{T1}{ppl}{m}{n}
\rput(11.596406,1.5701562){\huge $v_{i_2}$}
\usefont{T1}{ppl}{m}{n}
\rput(11.076406,0.29015625){\huge $\tilde{v}^{''}_i$}
\usefont{T1}{ppl}{m}{n}
\rput(10.296406,-0.52984375){\huge $\tilde{y}_i$}
\psline[linewidth=0.04cm,linestyle=dashed,dash=0.16cm 0.16cm](11.581875,1.2601563)(10.761875,0.48015624)
\psdots[dotsize=0.24](11.541875,1.2601563)
\psdots[dotsize=0.24,dotangle=-98.19838](9.809182,-0.41136748)
\psdots[dotsize=0.24](7.926216,4.280156)
\psdots[dotsize=0.24](7.188215,3.6601562)
\psdots[dotsize=0.24](8.639617,3.6601562)
\psline[linewidth=0.04cm](7.901616,4.300156)(7.212815,3.7001562)
\psline[linewidth=0.04cm](7.901616,4.340156)(8.639617,3.7001562)
\psline[linewidth=0.04cm](7.188215,3.6601562)(8.590417,3.6401563)
\psdots[dotsize=0.24,dotangle=-179.16212](7.907193,1.9600432)
\psdots[dotsize=0.24,dotangle=-179.16212](8.633964,2.5887508)
\psdots[dotsize=0.24,dotangle=-179.16212](7.182717,2.5714955)
\psline[linewidth=0.04cm](7.9321504,1.9403378)(8.610086,2.5484626)
\psline[linewidth=0.04cm](7.93287,1.900342)(7.1834364,2.5314996)
\psline[linewidth=0.04cm](8.633964,2.5887508)(7.2315516,2.5920782)
\psline[linewidth=0.04cm](7.1636147,3.6801562)(8.590417,2.6201563)
\psline[linewidth=0.04cm](7.0898147,2.6201563)(8.639617,3.6601562)
\usefont{T1}{ppl}{m}{n}
\rput(7.996406,4.6701565){\huge $\tilde{w}_{i_{1}}$}
\usefont{T1}{ppl}{m}{n}
\rput(6.8764064,3.2701562){\huge $\tilde{w}_{i_2}$}
\usefont{T1}{ppl}{m}{n}
\rput(9.236406,3.4901562){\huge $\tilde{w}_{i_3}$}
\usefont{T1}{ppl}{m}{n}
\rput(6.6564064,2.4101562){\huge $\tilde{w}_{i_4}$}
\usefont{T1}{ppl}{m}{n}
\rput(9.236406,2.4101562){\huge $\tilde{w}_{i_5}$}
\usefont{T1}{ppl}{m}{n}
\rput(8.436406,1.8501563){\huge $\tilde{w}_{i_6}$}
\usefont{T1}{ppl}{m}{n}
\rput(8.476406,1.0701562){\huge $\tilde{w}_{i_7}$}
\psdots[dotsize=0.24](7.901875,1.1001563)
\psline[linewidth=0.04cm](7.901875,2.0201561)(7.901875,1.1201563)
\psdots[dotsize=0.24](3.681875,-1.3598437)
\psdots[dotsize=0.24](4.921875,-1.3398438)
\psdots[dotsize=0.24](6.341875,-1.3198438)
\psline[linewidth=0.04cm](2.421875,-1.3998437)(6.361875,-1.2798438)
 \usefont{T1}{ppl}{m}{n}
 \rput(3.6964064,-1.8698438){\huge $\tilde{a}'_i$}
 \usefont{T1}{ppl}{m}{n}
 \rput(5.056406,-1.7898437){\huge $\tilde{c}_i$}
 \usefont{T1}{ppl}{m}{n}
 \rput(6.686406,-1.7698438){\huge $\tilde{b}'_i$}
\psline[linewidth=0.04cm](4.341875,2.4801562)(6.321875,-1.2398437)
\psline[linewidth=0.04cm](2.521875,3.7201562)(2.421875,-1.3598437)
\psline[linewidth=0.04cm](7.801875,4.280156)(4.341875,2.5001562)
\psline[linewidth=0.04cm](4.401875,2.4601562)(9.801875,-0.41984376)
\psdots[dotsize=0.24](4.321875,2.5201561)
\usefont{T1}{ppl}{m}{n}
\rput(4.346406,2.9901563){\huge $\tilde{w}_i$}
\psline[linewidth=0.04cm](4.301875,2.5401564)(3.661875,-1.2998438)
\psline[linewidth=0.04cm](7.921875,1.0601562)(8.961875,-1.3198438)
\psline[linewidth=0.04cm](8.981875,-1.2998438)(9.861875,-0.37984374)
\usefont{T1}{ppl}{m}{n}
\rput(5.978125,-6.409844){\huge (b)}
\psline[linewidth=0.04cm](6.341875,-1.2998438)(8.921875,-1.2998438)
\end{pspicture} 
}
\caption{Gadget $\tilde{W}_i$.With dotted lines showing the connection with the vertices of graph $G$.}
\label{W2}
\end{figure}

In Figure~\ref{W2} we show fourth gadget $\tilde{W}_i$ used for reduction and in Figure~\ref{W2-PATH} we show that
there are no extra edges other that two paths covering all vertices of $\tilde{W}_i$ from $\tilde{v}'$ and
$\tilde{v}''$. The need of this gadget is when we have a cycle in which more than one vertex is deleted,
and we do not have a path for at least one neighbor of deleted vertex to attach to gadget $W_i$, refer Figure~\ref{W2}[a].

\begin{figure}
\centering
\scalebox{0.45} 
{
\begin{pspicture}(0,-5.7120314)(12.689062,5.7120314)
\definecolor{color2862d}{rgb}{0.8117647058823529,0.8117647058823529,0.8117647058823529}
\psbezier[linewidth=0.04,doubleline=true,doublesep=0.08,doublecolor=color2862d](3.041875,-2.5798438)(3.2171307,-3.4798439)(3.8410108,-4.199844)(4.308359,-4.4398437)(4.7757077,-4.679844)(5.0483274,-4.779844)(5.6714587,-4.759844)(6.29459,-4.739844)(6.664574,-4.679844)(7.1124496,-4.3398438)(7.560325,-3.9998438)(8.027674,-3.2798438)(8.241875,-2.6598437)
\psdots[dotsize=0.24](2.441875,-1.3798437)
\psdots[dotsize=0.24](3.001875,-2.4798439)
\psline[linewidth=0.04cm,linestyle=dashed,dash=0.16cm 0.16cm](2.441875,-1.3598437)(3.001875,-2.4198437)
\psdots[dotsize=0.24](8.941875,-1.3398438)
\psdots[dotsize=0.24](8.301875,-2.5398438)
\psline[linewidth=0.04cm](2.421875,-1.3398438)(7.921875,1.1001563)
\usefont{T1}{ppl}{m}{n}
\rput(2.5764062,-2.7898438){\huge $v_{i_3}$}
\usefont{T1}{ppl}{m}{n}
\rput(8.836407,-2.6698437){\huge $v_{i_4}$}
\psline[linewidth=0.04cm,linestyle=dashed,dash=0.16cm 0.16cm](8.841875,-1.4398438)(8.241875,-2.5598438)
\usefont{T1}{ppl}{m}{n}
\rput(2.1164062,-1.6698438){\huge $\tilde{a}_i$}
\usefont{T1}{ppl}{m}{n}
\rput(9.286406,-1.6298438){\huge $\tilde{b}_i$}
\psdots[dotsize=0.24](1.701875,4.320156)
\psdots[dotsize=0.24](2.561875,3.6601562)
\psline[linewidth=0.04cm,linestyle=dotted,dotsep=0.16cm,arrowsize=0.013cm 2.0,arrowlength=1.4,arrowinset=0.4,doubleline=true,doublesep=0.06,doublecolor=red]{->}(1.681875,4.360156)(2.481875,3.7201562)
\psline[linewidth=0.04cm,linestyle=dotted,dotsep=0.16cm,arrowsize=0.013cm 2.0,arrowlength=1.4,arrowinset=0.4,doubleline=true,doublesep=0.06,doublecolor=red]{<-}(1.621875,4.400156)(0.901875,4.9801564)
\psdots[dotsize=0.24](0.921875,5.0001564)
\usefont{T1}{ppl}{m}{n}
\rput(0.93640625,5.3501563){\huge $v_{i_1}$}
\usefont{T1}{ppl}{m}{n}
\rput(2.0164063,4.6701565){\huge $\tilde{v}_i'$}
\usefont{T1}{ppl}{m}{n}
\rput(2.9164062,3.3101562){\huge $\tilde{x}_i$}
\psline[linewidth=0.04cm](2.541875,3.6801562)(7.921875,4.300156)
\psdots[dotsize=0.24,dotangle=-98.19838](10.670655,0.36546034)
\psline[linewidth=0.04cm,linestyle=dotted,dotsep=0.16cm,arrowsize=0.05291667cm 2.0,arrowlength=1.4,arrowinset=0.4,doubleline=true,doublesep=0.06,doublecolor=red]{<-}(10.681875,0.42015624)(9.828979,-0.4342195)
\usefont{T1}{ppl}{m}{n}
\rput(11.596406,1.5701562){\huge $v_{i_2}$}
\usefont{T1}{ppl}{m}{n}
\rput(11.076406,0.29015625){\huge $\tilde{v}^{''}_i$}
\usefont{T1}{ppl}{m}{n}
\rput(10.296406,-0.52984375){\huge $\tilde{y}_i$}
\psline[linewidth=0.04cm,linestyle=dotted,dotsep=0.16cm,arrowsize=0.013cm 2.0,arrowlength=1.4,arrowinset=0.4,doubleline=true,doublesep=0.06,doublecolor=red]{<-}(11.581875,1.2601563)(10.761875,0.48015624)
\psdots[dotsize=0.24](11.541875,1.2601563)
\psdots[dotsize=0.24,dotangle=-98.19838](9.809182,-0.41136748)
\psdots[dotsize=0.24](7.926216,4.280156)
\psdots[dotsize=0.24](7.188215,3.6601562)
\psdots[dotsize=0.24](8.639617,3.6601562)
\psline[linewidth=0.04cm](7.901616,4.300156)(7.212815,3.7001562)
\psline[linewidth=0.04cm,linestyle=dotted,dotsep=0.16cm,arrowsize=0.013cm 2.0,arrowlength=1.4,arrowinset=0.4,doubleline=true,doublesep=0.06,doublecolor=red]{->}(7.901616,4.340156)(8.639617,3.7001562)
\psline[linewidth=0.04cm,linestyle=dotted,dotsep=0.16cm,arrowsize=0.013cm 2.0,arrowlength=1.4,arrowinset=0.4,doubleline=true,doublesep=0.06,doublecolor=red]{<-}(7.188215,3.6601562)(8.590417,3.6401563)
\psdots[dotsize=0.24,dotangle=-179.16212](7.907193,1.9600432)
\psdots[dotsize=0.24,dotangle=-179.16212](8.633964,2.5887508)
\psdots[dotsize=0.24,dotangle=-179.16212](7.182717,2.5714955)
\psline[linewidth=0.04cm](7.9321504,1.9403378)(8.610086,2.5484626)
\psline[linewidth=0.04cm,linestyle=dotted,dotsep=0.16cm,arrowsize=0.013cm 2.0,arrowlength=1.4,arrowinset=0.4,doubleline=true,doublesep=0.06,doublecolor=red]{<-}(7.93287,1.900342)(7.1834364,2.5314996)
\psline[linewidth=0.04cm,linestyle=dotted,dotsep=0.16cm,arrowsize=0.013cm 2.0,arrowlength=1.4,arrowinset=0.4,doubleline=true,doublesep=0.06,doublecolor=red]{->}(8.633964,2.5887508)(7.2315516,2.5920782)
\psline[linewidth=0.04cm,linestyle=dotted,dotsep=0.16cm,arrowsize=0.013cm 2.0,arrowlength=1.4,arrowinset=0.4,doubleline=true,doublesep=0.06,doublecolor=red]{->}(7.1636147,3.6801562)(8.590417,2.6201563)
\psline[linewidth=0.04cm](7.0898147,2.6201563)(8.639617,3.6601562)
\usefont{T1}{ppl}{m}{n}
\rput(7.996406,4.6701565){\huge $\tilde{w}_{i_{1}}$}
\usefont{T1}{ppl}{m}{n}
\rput(6.8764064,3.2701562){\huge $\tilde{w}_{i_2}$}
\usefont{T1}{ppl}{m}{n}
\rput(9.236406,3.4901562){\huge $\tilde{w}_{i_3}$}
\usefont{T1}{ppl}{m}{n}
\rput(6.6564064,2.4101562){\huge $\tilde{w}_{i_4}$}
\usefont{T1}{ppl}{m}{n}
\rput(9.236406,2.4101562){\huge $\tilde{w}_{i_5}$}
\usefont{T1}{ppl}{m}{n}
\rput(8.436406,1.8501563){\huge $\tilde{w}_{i_6}$}
\usefont{T1}{ppl}{m}{n}
\rput(8.476406,1.0701562){\huge $\tilde{w}_{i_7}$}
\psdots[dotsize=0.24](7.901875,1.1001563)
\psline[linewidth=0.04cm,linestyle=dotted,dotsep=0.16cm,arrowsize=0.013cm 2.0,arrowlength=1.4,arrowinset=0.4,doubleline=true,doublesep=0.06,doublecolor=red]{->}(7.901875,2.0201561)(7.901875,1.1201563)
\psdots[dotsize=0.24](3.681875,-1.3598437)
\psdots[dotsize=0.24](4.921875,-1.3398438)
\psdots[dotsize=0.24](6.341875,-1.3198438)
\psline[linewidth=0.04cm,linestyle=dotted,dotsep=0.16cm,arrowsize=0.013cm 2.0,arrowlength=1.4,arrowinset=0.4,doubleline=true,doublesep=0.06,doublecolor=red]{->}(2.421875,-1.3998437)(6.361875,-1.2798438)
 \usefont{T1}{ppl}{m}{n}
 \rput(3.6964064,-1.8698438){\huge $\tilde{a}'_i$}
 \usefont{T1}{ppl}{m}{n}
 \rput(5.056406,-1.7898437){\huge $\tilde{c}_i$}
 \usefont{T1}{ppl}{m}{n}
 \rput(6.686406,-1.7698438){\huge $\tilde{b}'_i$}
\psline[linewidth=0.04cm,linestyle=dotted,dotsep=0.16cm,arrowsize=0.013cm 2.0,arrowlength=1.4,arrowinset=0.4,doubleline=true,doublesep=0.06,doublecolor=red]{<-}(4.341875,2.4801562)(6.321875,-1.2398437)
\psline[linewidth=0.04cm,linestyle=dotted,dotsep=0.16cm,arrowsize=0.013cm 2.0,arrowlength=1.4,arrowinset=0.4,doubleline=true,doublesep=0.06,doublecolor=red]{->}(2.521875,3.7201562)(2.421875,-1.3598437)
\psline[linewidth=0.04cm,linestyle=dotted,dotsep=0.16cm,arrowsize=0.013cm 2.0,arrowlength=1.4,arrowinset=0.4,doubleline=true,doublesep=0.06,doublecolor=red]{<-}(7.801875,4.280156)(4.341875,2.5001562)
\psline[linewidth=0.04cm](4.401875,2.4601562)(9.801875,-0.41984376)
\psdots[dotsize=0.24](4.321875,2.5201561)
\usefont{T1}{ppl}{m}{n}
\rput(4.346406,2.9901563){\huge $\tilde{w}_i$}
\psline[linewidth=0.04cm](4.301875,2.5401564)(3.661875,-1.2998438)
\psline[linewidth=0.04cm,linestyle=dotted,dotsep=0.16cm,arrowsize=0.013cm 2.0,arrowlength=1.4,arrowinset=0.4,doubleline=true,doublesep=0.06,doublecolor=red]{->}(7.921875,1.0601562)(8.961875,-1.3198438)
\psline[linewidth=0.04cm,linestyle=dotted,dotsep=0.16cm,arrowsize=0.013cm 2.0,arrowlength=1.4,arrowinset=0.4,doubleline=true,doublesep=0.06,doublecolor=red]{->}(8.981875,-1.2998438)(9.861875,-0.37984374)
\usefont{T1}{ppl}{m}{n}
\rput(5.978125,-5.369844){\huge (b)}
\psline[linewidth=0.04cm](6.341875,-1.2998438)(8.921875,-1.2998438)
\end{pspicture} 
}
\scalebox{0.5} 
{
\begin{pspicture}(0,-5.7120314)(12.689062,5.7120314)
\definecolor{color3377d}{rgb}{0.0,0.47058823529411764,0.058823529411764705}
\psbezier[linewidth=0.04,linestyle=dotted,dotsep=0.16cm,doubleline=true,doublesep=0.08,doublecolor=color3377d,arrowsize=0.013cm 2.0,arrowlength=1.4,arrowinset=0.4]{->}(3.041875,-2.5798438)(3.2171307,-3.4798439)(3.8410108,-4.199844)(4.308359,-4.4398437)(4.7757077,-4.679844)(5.0483274,-4.779844)(5.6714587,-4.759844)(6.29459,-4.739844)(6.664574,-4.679844)(7.1124496,-4.3398438)(7.560325,-3.9998438)(8.027674,-3.2798438)(8.241875,-2.6598437)
\psdots[dotsize=0.24](2.441875,-1.3798437)
\psdots[dotsize=0.24](3.001875,-2.4798439)
\psline[linewidth=0.04cm,linestyle=dotted,dotsep=0.16cm,arrowsize=0.013cm 2.0,arrowlength=1.4,arrowinset=0.4,doubleline=true,doublesep=0.06,doublecolor=color3377d]{->}(2.441875,-1.3598437)(3.001875,-2.4198437)
\psdots[dotsize=0.24](8.941875,-1.3398438)
\psdots[dotsize=0.24](8.301875,-2.5398438)
\psline[linewidth=0.04cm,linestyle=dotted,dotsep=0.16cm,arrowsize=0.013cm 2.0,arrowlength=1.4,arrowinset=0.4,doubleline=true,doublesep=0.06,doublecolor=color3377d]{<-}(2.421875,-1.3398438)(7.921875,1.1001563)
\usefont{T1}{ppl}{m}{n}
\rput(2.5764062,-2.7898438){\huge $v_{i_3}$}
\usefont{T1}{ppl}{m}{n}
\rput(8.836407,-2.6698437){\huge $v_{i_4}$}
\psline[linewidth=0.04cm,linestyle=dotted,dotsep=0.16cm,arrowsize=0.013cm 2.0,arrowlength=1.4,arrowinset=0.4,doubleline=true,doublesep=0.06,doublecolor=color3377d]{<-}(8.841875,-1.4398438)(8.241875,-2.5598438)
\usefont{T1}{ppl}{m}{n}
\rput(2.1164062,-1.6698438){\huge $\tilde{a}_i$}
\usefont{T1}{ppl}{m}{n}
\rput(9.286406,-1.7298438){\huge $\tilde{b}_i$}
\psdots[dotsize=0.24](1.701875,4.320156)
\psdots[dotsize=0.24](2.561875,3.6601562)
\psline[linewidth=0.04cm,linestyle=dotted,dotsep=0.16cm,arrowsize=0.013cm 2.0,arrowlength=1.4,arrowinset=0.4,doubleline=true,doublesep=0.06,doublecolor=color3377d]{->}(1.681875,4.360156)(2.481875,3.7201562)
\psline[linewidth=0.04cm](1.621875,4.400156)(0.901875,4.9801564)
\psdots[dotsize=0.24](0.921875,5.0001564)
\usefont{T1}{ppl}{m}{n}
\rput(0.93640625,5.3501563){\huge $v_{i_1}$}
\usefont{T1}{ppl}{m}{n}
\rput(2.0164063,4.6701565){\huge $\tilde{v}_i'$}
\usefont{T1}{ppl}{m}{n}
\rput(2.9164062,3.3101562){\huge $\tilde{x}_i$}
\psline[linewidth=0.04cm,linestyle=dotted,dotsep=0.16cm,arrowsize=0.013cm 2.0,arrowlength=1.4,arrowinset=0.4,doubleline=true,doublesep=0.06,doublecolor=color3377d]{->}(2.541875,3.6801562)(7.921875,4.300156)
\psdots[dotsize=0.24,dotangle=-98.19838](10.670655,0.36546034)
\psline[linewidth=0.04cm,linestyle=dotted,dotsep=0.16cm,arrowsize=0.05291667cm 2.0,arrowlength=1.4,arrowinset=0.4,doubleline=true,doublesep=0.06,doublecolor=color3377d]{<-}(10.681875,0.42015624)(9.828979,-0.4342195)
\usefont{T1}{ppl}{m}{n}
\rput(11.596406,1.5701562){\huge $v_{i_2}$}
\usefont{T1}{ppl}{m}{n}
\rput(11.076406,0.29015625){\huge $\tilde{v}^{''}_i$}
\usefont{T1}{ppl}{m}{n}
\rput(10.296406,-0.52984375){\huge $\tilde{y}_i$}
\psline[linewidth=0.04cm](11.581875,1.2601563)(10.761875,0.48015624)
\psdots[dotsize=0.24](11.541875,1.2601563)
\psdots[dotsize=0.24,dotangle=-98.19838](9.809182,-0.41136748)
\psdots[dotsize=0.24](7.926216,4.280156)
\psdots[dotsize=0.24](7.188215,3.6601562)
\psdots[dotsize=0.24](8.639617,3.6601562)
\psline[linewidth=0.04cm,linestyle=dotted,dotsep=0.16cm,arrowsize=0.013cm 2.0,arrowlength=1.4,arrowinset=0.4,doubleline=true,doublesep=0.06,doublecolor=color3377d]{->}(7.901616,4.300156)(7.212815,3.7001562)
\psline[linewidth=0.04cm](7.901616,4.340156)(8.639617,3.7001562)
\psline[linewidth=0.04cm,linestyle=dotted,dotsep=0.16cm,arrowsize=0.013cm 2.0,arrowlength=1.4,arrowinset=0.4,doubleline=true,doublesep=0.06,doublecolor=color3377d]{->}(7.188215,3.6601562)(8.590417,3.6401563)
\psdots[dotsize=0.24,dotangle=-179.16212](7.907193,1.9600432)
\psdots[dotsize=0.24,dotangle=-179.16212](8.633964,2.5887508)
\psdots[dotsize=0.24,dotangle=-179.16212](7.182717,2.5714955)
\psline[linewidth=0.04cm,linestyle=dotted,dotsep=0.16cm,arrowsize=0.05291667cm 2.0,arrowlength=1.4,arrowinset=0.4,doubleline=true,doublesep=0.06,doublecolor=color3377d]{<-}(7.9321504,1.9403378)(8.610086,2.5484626)
\psline[linewidth=0.04cm](7.93287,1.900342)(7.1834364,2.5314996)
\psline[linewidth=0.04cm,linestyle=dotted,dotsep=0.16cm,arrowsize=0.013cm 2.0,arrowlength=1.4,arrowinset=0.4,doubleline=true,doublesep=0.06,doublecolor=color3377d]{<-}(8.633964,2.5887508)(7.2315516,2.5920782)
\psline[linewidth=0.04cm](7.1636147,3.6801562)(8.590417,2.6201563)
\psline[linewidth=0.04cm,linestyle=dotted,dotsep=0.16cm,arrowsize=0.013cm 2.0,arrowlength=1.4,arrowinset=0.4,doubleline=true,doublesep=0.06,doublecolor=color3377d]{<-}(7.0898147,2.6201563)(8.639617,3.6601562)
\usefont{T1}{ppl}{m}{n}
\rput(7.996406,4.650156){\huge $\tilde{w}_{i_1}$}
\usefont{T1}{ppl}{m}{n}
\rput(6.7764064,3.2501562){\huge $\tilde{w}_{i_2}$}
\usefont{T1}{ppl}{m}{n}
\rput(9.236406,3.4701562){\huge $\tilde{w}_{i_3}$}
\usefont{T1}{ppl}{m}{n}
\rput(6.6564064,2.3901563){\huge $\tilde{w}_{i_4}$}
\usefont{T1}{ppl}{m}{n}
\rput(9.236406,2.3901563){\huge $\tilde{w}_{i_5}$}
\usefont{T1}{ppl}{m}{n}
\rput(8.436406,1.8301562){\huge $\tilde{w}_{i_6}$}
\usefont{T1}{ppl}{m}{n}
\rput(8.476406,1.0501562){\huge $\tilde{w}_{i_7}$}
\psdots[dotsize=0.24](7.901875,1.1001563)
\psline[linewidth=0.04cm,linestyle=dotted,dotsep=0.16cm,arrowsize=0.013cm 2.0,arrowlength=1.4,arrowinset=0.4,doubleline=true,doublesep=0.06,doublecolor=color3377d]{->}(7.901875,2.0201561)(7.901875,1.1201563)
\psdots[dotsize=0.24](3.761875,-1.3998437)
\psdots[dotsize=0.24](5.161875,-1.3398438)
\psdots[dotsize=0.24](6.341875,-1.3198438)
\psline[linewidth=0.04cm,linestyle=dotted,dotsep=0.16cm,arrowsize=0.013cm 2.0,arrowlength=1.4,arrowinset=0.4,doubleline=true,doublesep=0.06,doublecolor=color3377d]{<-}(3.741875,-1.3798437)(6.301875,-1.3398438)
\usefont{T1}{ppl}{m}{n}
\rput(3.6964064,-1.8698438){\huge $\tilde{a}'_i$}
\usefont{T1}{ppl}{m}{n}
\rput(5.0964065,-1.7498437){\huge $\tilde{c}_i$}
\usefont{T1}{ppl}{m}{n}
\rput(6.586406,-1.8698438){\huge $\tilde{b}'_i$}
\psline[linewidth=0.04cm](4.341875,2.4801562)(6.321875,-1.2398437)
\psline[linewidth=0.04cm](2.521875,3.7201562)(2.421875,-1.3598437)
\psline[linewidth=0.04cm](7.801875,4.280156)(4.341875,2.5001562)
\psline[linewidth=0.04cm,linestyle=dotted,dotsep=0.16cm,arrowsize=0.013cm 2.0,arrowlength=1.4,arrowinset=0.4,doubleline=true,doublesep=0.06,doublecolor=color3377d]{->}(4.401875,2.4601562)(9.801875,-0.41984376)
\psdots[dotsize=0.24](4.321875,2.5201561)
\usefont{T1}{ppl}{m}{n}
\rput(4.276406,2.8901563){\huge $\tilde{w}_i$}
\psline[linewidth=0.04cm,linestyle=dotted,dotsep=0.16cm,arrowsize=0.013cm 2.0,arrowlength=1.4,arrowinset=0.4,doubleline=true,doublesep=0.06,doublecolor=color3377d]{<-}(4.361875,2.5601563)(3.741875,-1.4198438)
\psline[linewidth=0.04cm](7.921875,1.0601562)(8.961875,-1.3198438)
\psline[linewidth=0.04cm](8.981875,-1.2998438)(9.861875,-0.37984374)
\usefont{T1}{ppl}{m}{n}
\rput(5.968125,-5.369844){\huge (a)}
\psline[linewidth=0.04cm,linestyle=dotted,dotsep=0.16cm,arrowsize=0.013cm 2.0,arrowlength=1.4,arrowinset=0.4,doubleline=true,doublesep=0.06,doublecolor=color3377d]{<-}(6.341875,-1.2998438)(8.921875,-1.2998438)
\psline[linewidth=0.04cm](3.721875,-1.3798437)(2.441875,-1.3998437)
\end{pspicture} 
}

\caption{Edge Set of gadget $\tilde{W}$ is exactly union of two hamiltonian paths. (a),(b) showing two path from $v'$ to $v''$ covering all edges in $\tilde{W}$.}
\label{W2-PATH}
\end{figure}
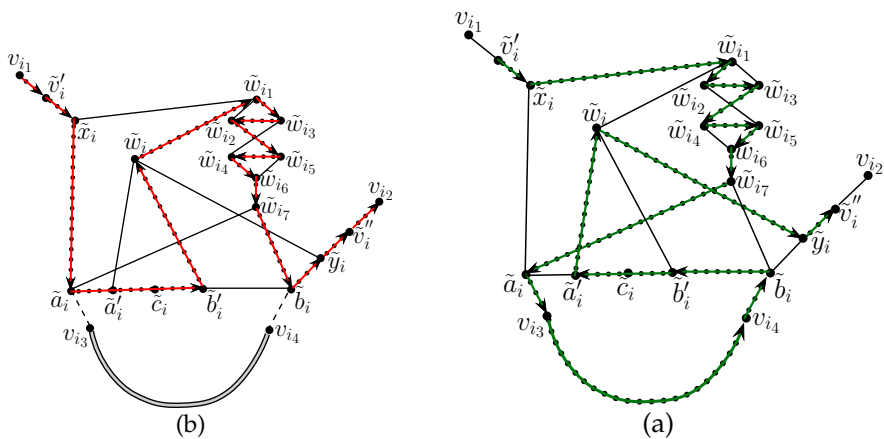

In Lemma~\ref{vctildew} we prove the minimum number of vertices required to cover edges of $\tilde{W}_i$ and in Lemma~\ref{tildwwvc17}
we prove that if at least one neighbor of deleted vertex say $v_i$ for which $\tilde{W}_i$ is inserted is not chosen in vertex cover then
we require 10 vertices to cover the edges and if all neighbors are included then we require 9 vertices only.

\begin{lemma}
 To cover edges of gadget $\tilde{W}_i$, 9 vertices are necessary and sufficient.
 \label{vctildew}
\end{lemma}
\begin{proof}
Consider gadget $\tilde{W}_i$ [figure~\ref{W2}].
$\{\tilde{w}_{i_1},\tilde{w}_{i_2},\tilde{w}_{i_3}\}$, $\{\tilde{w}_{i_4},\tilde{w}_{i_5},\tilde{w}_{i_6}\}$,
forms $K_3$ and $\{(\tilde{v}'_i,\tilde{x}_i),(\tilde{v}^{''}_i$ $,\tilde{y}_i),$ $(\tilde{c}_i,$ $\tilde{b}'_i),(\tilde{b}_i,\tilde{w_{i_7}})
,(\tilde{a}_i,\tilde{a}'_i)\}$ are edges (all vertex disjoint), so we require at least 9 vertices to cover all edges of $\tilde{W}_i$.
If $V'=\{\tilde{w}_{i_1},\tilde{w}_{i_2},\tilde{w}_{i_4},\tilde{w}_{i_6},$ $\tilde{w}_{i_7},\tilde{x}_i,\tilde{y}_i,\tilde{a}'_i,\tilde{b}'_i\}$ then
$V'$can cover all edges of $\tilde{W}_i$ with $\lvert V' \rvert=9$
\qed
\end{proof}
\begin{figure}
\centering
\scalebox{0.5} 
{
\begin{pspicture}(0,-5.0889063)(7.385625,5.0889063)
\rput{14.307357}(0.77244216,-0.56487215){\psarc[linewidth=0.04](2.6365626,2.7948241){1.5392593}{-28.325666}{83.03499}}
\psline[linewidth=0.04](1.2254708,3.3964462)(1.6853186,4.1947823)(2.557436,4.2831645)
\psdots[dotsize=0.32,dotangle=32.93641](1.6853186,4.1947823)
\psdots[dotsize=0.32,dotangle=32.93641](2.557436,4.2831645)
\psdots[dotsize=0.32,dotangle=32.93641](1.2313821,3.424106)
\rput{14.307357}(0.7994253,-0.5806471){\psarc[linewidth=0.04](2.7128985,2.8944323){1.5451068}{150.5012}{257.11948}}
\psline[linewidth=0.04](4.105533,2.3515332)(3.6948066,1.5268493)(2.8296304,1.3858873)
\psdots[dotsize=0.32,dotangle=-143.5965](2.8296304,1.3858873)
\usefont{T1}{ppl}{m}{n}
\rput(0.9828125,3.3692188){\Large $u_1$}
\usefont{T1}{ppl}{m}{n}
\rput(2.8028126,4.6492186){\Large $u_2$}
\usefont{T1}{ppl}{m}{n}
\rput(1.5528125,4.5292187){\Large $u$}
\usefont{T1}{ppl}{m}{n}
\rput(3.0028124,1.0092187){\Large $v_1$}
\usefont{T1}{ppl}{m}{n}
\rput(4.5828123,2.3492188){\Large $v_2$}
\psdots[dotsize=0.32,dotangle=-143.5965](3.6948066,1.5268493)
\psdots[dotsize=0.32,dotangle=-143.5965](4.1013055,2.3235667)
\usefont{T1}{ppl}{m}{n}
\rput(3.9728124,1.2492187){\Large $v$}
\rput{2.972276}(-0.071043625,-0.22336951){\psarc[linewidth=0.04](4.269347,-1.4808667){1.0131662}{112.62248}{2.232202}}
\psline[linewidth=0.04](3.8778276,-0.50997764)(4.7967534,-0.5761285)(5.256222,-1.3226463)
\psdots[dotsize=0.32,dotangle=-31.238762](4.7967534,-0.5761285)
\psdots[dotsize=0.32,dotangle=-31.238762](5.256222,-1.3226463)
\psdots[dotsize=0.32,dotangle=-31.238762](3.9053001,-0.50324947)
\rput{56.91868}(-0.8870747,-2.898868){\psarc[linewidth=0.04](2.2305117,-2.2677126){1.0131662}{112.62248}{2.232202}}
\psline[linewidth=0.04](1.2151545,-2.0128343)(1.8094614,-1.3088461)(2.6834126,-1.3767381)
\psdots[dotsize=0.32,dotangle=22.70764](1.8094614,-1.3088461)
\psdots[dotsize=0.32,dotangle=22.70764](2.6834126,-1.3767381)
\psdots[dotsize=0.32,dotangle=22.70764](1.2258837,-1.986664)
\usefont{T1}{ppl}{m}{n}
\rput(1.6728125,-0.99078125){\Large $v$}
\usefont{T1}{ppl}{m}{n}
\rput(0.9228125,-1.8507812){\Large $v_3$}
\usefont{T1}{ppl}{m}{n}
\rput(2.8428125,-1.1107812){\Large $v_4$}
\usefont{T1}{ppl}{m}{n}
\rput(4.9528127,-0.27078125){\Large $u$}
\usefont{T1}{ppl}{m}{n}
\rput(4.0828123,-0.13078125){\Large $u_3$}
\usefont{T1}{ppl}{m}{n}
\rput(5.7228127,-1.2107812){\Large $u_4$}
\psline[linewidth=0.04cm,linestyle=dashed,dash=0.16cm 0.16cm](0.18,0.46421874)(6.62,0.46421874)
\psdots[dotsize=0.06](3.66,-2.6157813)
\psdots[dotsize=0.06](3.8,-2.9157813)
\psdots[dotsize=0.06](3.94,-3.1557813)
\pscircle[linewidth=0.04,dimen=outer](4.45,-3.9257812){0.53}
\pscircle[linewidth=0.04,dimen=outer](5.41,4.054219){0.53}
\psdots[dotsize=0.06](4.18,3.8042188)
\psdots[dotsize=0.06](4.46,3.9642189)
\psdots[dotsize=0.06](4.7,4.124219)
\pscircle[linewidth=0.04,dimen=outer](6.17,1.3542187){0.53}
\psdots[dotsize=0.06](4.76,1.7442187)
\psdots[dotsize=0.06](5.12,1.5642188)
\psdots[dotsize=0.06](5.44,1.4042188)
\pscircle[linewidth=0.04,dimen=outer](0.91,-4.3057814){0.53}
\psdots[dotsize=0.06](1.38,-3.6157813)
\psdots[dotsize=0.06](1.5,-3.4957812)
\psdots[dotsize=0.06](1.6,-3.3557813)
\usefont{T1}{ptm}{m}{n}
\rput(1.1828125,1.9292188){\Large $C_i$}
\usefont{T1}{ptm}{m}{n}
\rput(6.0828123,4.6492186){\Large $C_1$}
\usefont{T1}{ptm}{m}{n}
\rput(6.6128125,2.0292187){\Large $C_k$}
\usefont{T1}{ptm}{m}{n}
\rput(1.5328125,-4.790781){\Large $C'_1$}
\usefont{T1}{ptm}{m}{n}
\rput(5.2628126,-3.9307814){\Large $C'_k$}
\usefont{T1}{ptm}{m}{n}
\rput(1.1828125,-2.7907813){\Large $C'_{i'}$}
\usefont{T1}{ptm}{m}{n}
\rput(5.5028124,-2.2907813){\Large $C'_{j'}$}
\usefont{T1}{ptm}{m}{n}
\rput(0.4028125,4.829219){\Large $H$}
\usefont{T1}{ptm}{m}{n}
\rput(0.4728125,0.22921875){\Large $H'$}
\end{pspicture} 
}
\scalebox{0.6} 
{
\begin{pspicture}(0,-5.283863)(6.845625,5.283863)
\psframe[linewidth=0.09,linestyle=dashed,dash=0.16cm 0.16cm,dimen=outer](4.6,4.839175)(1.98,2.9591753)
\psdots[dotsize=0.36](0.82,4.7791753)
\psdots[dotsize=0.36](5.78,4.7791753)
\psline[linewidth=0.04](2.0,4.1791754)(0.78,4.7991753)(2.0,3.3591754)
\psline[linewidth=0.04](4.6,4.1591754)(5.8,4.7791753)(4.6,3.3391755)
\psdots[dotsize=0.36](1.98,4.1991754)
\psdots[dotsize=0.36](1.98,3.3791754)
\psdots[dotsize=0.36](4.58,4.1591754)
\psdots[dotsize=0.36](4.6,3.3591754)
\psdots[dotsize=0.36](2.58,2.9791753)
\psdots[dotsize=0.36](4.0,2.9791753)
\psline[linewidth=0.04cm](2.58,2.9591753)(2.78,2.1391754)
\psline[linewidth=0.04cm](4.0,2.9591753)(3.8,2.1591754)
\psbezier[linewidth=0.04,doubleline=true,doublesep=0.06](2.8156753,2.1791754)(2.64,1.8540422)(2.7200198,0.6334069)(3.32,0.6391754)(3.9199803,0.64494395)(3.9224298,1.8392633)(3.764322,2.149618)
\psdots[dotsize=0.36](2.78,2.1791754)
\psdots[dotsize=0.36](3.78,2.1791754)
\psline[linewidth=0.04cm](2.54,2.9991753)(1.38,2.1791754)
\psline[linewidth=0.04cm](4.06,2.9791753)(5.22,2.1591754)
\psbezier[linewidth=0.04,doubleline=true,doublesep=0.06](1.56,2.1791754)(1.2151455,2.0921822)(0.9258064,1.5752019)(0.8949384,1.0665926)(0.8640704,0.55798334)(0.85328364,-0.31632128)(1.24,-0.8608246)
\psdots[dotsize=0.36](1.46,2.2191753)
\psdots[dotsize=0.36,dotangle=-11.187331](1.2785989,-0.7740887)
\psline[linewidth=0.04cm](1.2432394,-0.74670804)(1.8260415,-1.8201798)
\psdots[dotsize=0.36,dotangle=-11.187331](1.8614011,-1.8475605)
\psframe[linewidth=0.09,linestyle=dashed,dash=0.16cm 0.16cm,dimen=outer](3.98,-0.6608246)(2.58,-3.4408245)
\psbezier[linewidth=0.04,doubleline=true,doublesep=0.06](5.0781455,2.0991755)(5.4230003,2.0121822)(5.7123394,1.495202)(5.743207,0.98659265)(5.774075,0.47798336)(5.784862,-0.39632127)(5.3981457,-0.94082457)
\psdots[dotsize=0.36](5.1781454,2.1391754)
\psdots[dotsize=0.36,dotangle=-348.81268](5.3595467,-0.85408866)
\psline[linewidth=0.04cm](5.394906,-0.826708)(4.812104,-1.9001799)
\psdots[dotsize=0.36,dotangle=-348.81268](4.7767444,-1.9275604)
\psdots[dotsize=0.36](2.62,-1.2208246)
\psdots[dotsize=0.36](2.6,-2.6408246)
\psdots[dotsize=0.36](3.96,-1.2608246)
\psdots[dotsize=0.36](3.96,-2.6408246)
\psline[linewidth=0.04](2.6,-1.2208246)(1.86,-1.8808246)(2.56,-2.6408246)
\psline[linewidth=0.04](3.94,-1.2408246)(4.74,-1.8808246)(4.02,-2.6608245)
\psdots[dotsize=0.36](3.0,-3.4208245)
\psdots[dotsize=0.36](3.6,-3.4008245)
\psline[linewidth=0.04cm](2.96,-3.3808246)(2.36,-3.8608246)
\psline[linewidth=0.04cm](3.66,-3.4208245)(4.26,-3.9008245)
\psbezier[linewidth=0.04,doubleline=true,doublesep=0.06](2.36596,-3.9008245)(2.2,-4.245958)(2.6771555,-5.266593)(3.3409748,-5.2008247)(4.004794,-5.135056)(4.26,-4.2808247)(4.22,-3.9208245)
\psdots[dotsize=0.36](2.34,-3.9208245)
\psdots[dotsize=0.36](4.24,-3.9408245)
\usefont{T1}{ptm}{m}{n}
\rput(3.3310938,3.7991755){\LARGE $W$}
\usefont{T1}{ptm}{m}{n}
\rput(3.2728126,-1.7958245){\Large $\tilde{W}$}
\usefont{T1}{ptm}{m}{n}
\rput(0.5028125,4.984175){\Large $u'$}
\usefont{T1}{ptm}{m}{n}
\rput(6.2128124,5.0241756){\Large $u''$}
\usefont{T1}{ptm}{m}{n}
\rput(1.0228125,2.4041755){\Large $u_1$}
\usefont{T1}{ptm}{m}{n}
\rput(5.6828127,2.3441753){\Large $u_2$}
\usefont{T1}{ptm}{m}{n}
\rput(4.3628125,2.0041754){\Large $u_4$}
\usefont{T1}{ptm}{m}{n}
\rput(2.3228126,2.0041754){\Large $u_3$}
\usefont{T1}{ptm}{m}{n}
\rput(0.8428125,-0.8758246){\Large $v_1$}
\usefont{T1}{ptm}{m}{n}
\rput(5.9628124,-0.8758246){\Large $v_2$}
\usefont{T1}{ptm}{m}{n}
\rput(1.8628125,-3.9558246){\Large $v_3$}
\usefont{T1}{ptm}{m}{n}
\rput(4.7228127,-3.9558246){\Large $v_4$}
\usefont{T1}{ptm}{m}{n}
\rput(2.3528125,4.1641755){\Large $x$}
\usefont{T1}{ptm}{m}{n}
\rput(2.3228126,3.6041753){\Large $x'$}
\usefont{T1}{ptm}{m}{n}
\rput(4.1728125,4.1441755){\Large $y$}
\usefont{T1}{ptm}{m}{n}
\rput(4.2628126,3.5841753){\Large $y'$}
\usefont{T1}{ptm}{m}{n}
\rput(2.8428125,3.3041754){\Large $a$}
\usefont{T1}{ptm}{m}{n}
\rput(3.7528124,3.2841754){\Large $b$}
\usefont{T1}{ptm}{m}{n}
\rput(1.3728125,-1.8358246){\Large $\tilde{v}'$}
\usefont{T1}{ptm}{m}{n}
\rput(5.2028127,-1.8758246){\Large $\tilde{v}''$}
\usefont{T1}{ptm}{m}{n}
\rput(2.9328125,-2.9758246){\Large $\tilde{a}$}
\usefont{T1}{ptm}{m}{n}
\rput(3.5228126,-2.9758246){\Large $\tilde{b}$}
\usefont{T1}{ptm}{m}{n}
\rput(2.9628124,-1.2358246){\Large $\tilde{x}$}
\usefont{T1}{ptm}{m}{n}
\rput(3.0128126,-2.4158247){\Large $\tilde{x}'$}
\usefont{T1}{ptm}{m}{n}
\rput(3.5428126,-1.2558246){\Large $\tilde{y}$}
\usefont{T1}{ptm}{m}{n}
\rput(3.5928125,-2.4358246){\Large $\tilde{y}'$}
\end{pspicture} 
}
\label{Connection1}
\caption{Connection of gadget $W$ and $\tilde{W}$ in one cycle when a cycle is broken at two points, where $H,H'$ is 2-factoring of given 4-regular graph}
\end{figure}
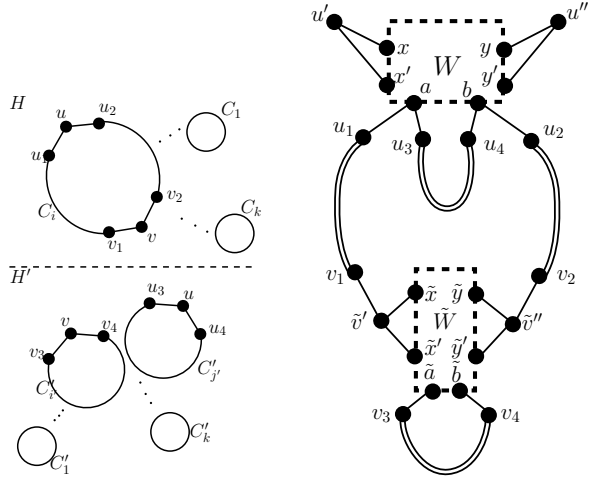
 
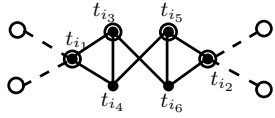
\begin{figure}
\centering
\scalebox{0.9} 
{
\begin{pspicture}(0,-0.88921875)(3.92,0.88921875)
\psdots[dotsize=0.16](0.96,-0.00421875)
\psdots[dotsize=0.16](1.56,0.39578125)
\psdots[dotsize=0.16](1.56,-0.40421876)
\psdots[dotsize=0.16](2.38,0.39578125)
\psdots[dotsize=0.16](2.38,-0.40421876)
\psdots[dotsize=0.16](2.96,-0.00421875)
\psline[linewidth=0.04](0.94,-0.00421875)(1.56,-0.38421875)(1.56,0.39578125)(2.36,-0.40421876)(2.36,0.39578125)(2.96,0.01578125)
\psline[linewidth=0.04cm](0.92,-0.00421875)(1.56,0.41578126)
\psline[linewidth=0.04cm](1.56,-0.38421875)(2.36,0.41578126)
\psline[linewidth=0.04cm](2.36,-0.42421874)(2.94,-0.00421875)
\pscircle[linewidth=0.04,dimen=outer](0.96,-0.00421875){0.14}
\pscircle[linewidth=0.04,dimen=outer](1.56,0.39578125){0.14}
\pscircle[linewidth=0.04,dimen=outer](2.38,0.39578125){0.14}
\pscircle[linewidth=0.04,dimen=outer](2.96,-0.00421875){0.14}
\psline[linewidth=0.04cm,linestyle=dashed,dash=0.16cm 0.16cm](0.94,-0.00421875)(0.24,0.41578126)
\psline[linewidth=0.04cm,linestyle=dashed,dash=0.16cm 0.16cm](3.72,-0.40421876)(2.98,-0.00421875)
\psline[linewidth=0.04cm,linestyle=dashed,dash=0.16cm 0.16cm](0.88,-0.04421875)(0.22,-0.40421876)
\psline[linewidth=0.04cm,linestyle=dashed,dash=0.16cm 0.16cm](3.7,0.39578125)(3.08,0.03578125)
\usefont{T1}{ptm}{m}{n}
\rput(1.0445312,0.26578125){$t_{i_1}$}
\usefont{T1}{ptm}{m}{n}
\rput(3.1845312,-0.33421874){$t_{i_2}$}
\usefont{T1}{ptm}{m}{n}
\rput(1.5645312,-0.63421875){$t_{i_4}$}
\usefont{T1}{ptm}{m}{n}
\rput(1.4445312,0.68578124){$t_{i_3}$}
\usefont{T1}{ptm}{m}{n}
\rput(2.4445312,0.68578124){$t_{i_5}$}
\usefont{T1}{ptm}{m}{n}
\rput(2.4445312,-0.65421873){$t_{i_6}$}
\pscircle[linewidth=0.04,dimen=outer](3.79,0.44578126){0.13}
\pscircle[linewidth=0.04,dimen=outer](3.79,-0.45421875){0.13}
\pscircle[linewidth=0.04,dimen=outer](0.15,0.42578125){0.13}
\pscircle[linewidth=0.04,dimen=outer](0.13,-0.39421874){0.13}
\end{pspicture} 
}

\caption{Gadget $W_{C_i}$, double circled vertex forming a minimum vertex cover and empty circle showing outside vertex to which $W_{C_i}$ is connected and dotted lines connection to them}
\label{gadgetW_c}
\end{figure}

\begin{lemma}
 Consider a vertex $v_i \in V_H$ and $v_i \in V(\tilde{C})$, $\tilde{C} \in \mathfrak{C}_H \cup \mathfrak{C}_{H'}$.
 Let neighbors of $v_i$ be $v_{i_1},v_{i_2},v_{i_3},v_{i_4}$. Here $v_i$ is deleted and $\tilde{W}_i$ is inserted as shown in figure ~\ref{W2}.
 If at least one of $v_{i_1},v_{i_2},v_{i_3},v_{i_4}$ is not included in vertex cover then minimum number of vertex required
 to cover $E'= E(\tilde{W}_i) \cup \{(v_{i_1},\tilde{v}'_i),(\tilde{v}''_i,v_{i_2}),(v_{i_3},\tilde{a}_i),(v_{i_4},$ $\tilde{b}_i)\}$
 apart from $v_1,v_2,v_3,v_4$ is 10.
 \label{tildwwvc17}
\end{lemma}
\begin{proof}
We know that $\{\tilde{w}_{i_1},\tilde{w}_{i_2},\tilde{w}_{i_3}\}$, $\{\tilde{w}_{i_4},\tilde{w}_{i_5},\tilde{w}_{i_6}\}$ form triangles, so we require at least 4 vertices among them.

Suppose $v_{i_1}$ is not chosen in the vertex cover so $\tilde{v}'_i$ is forced, but then we have
$\{(\tilde{x}_i,\tilde{a}_i),$ $(\tilde{a}'_i,\tilde{c}_i),(\tilde{w}_i,\tilde{b}'_i),$ $(\tilde{b}_i,\tilde{w}_{i_7}),(\tilde{y}_i,\tilde{v}''_i)\}$
as vertex disjoint edges left to be covered.
Suppose $v_{i_2}$ is not chosen in the vertex cover so $\tilde{v}''_i$ is forced, but then we have
$\{(\tilde{x}_i,\tilde{v}'_i),(\tilde{a}'_i,\tilde{c}_i),(\tilde{w}_i,\tilde{b}'_i),(\tilde{a}_i,\tilde{w}_{i_7}),(\tilde{y}_i,\tilde{b}_i)\}$
as vertex disjoint edges left to be covered.
Similarly if $\tilde{a}_i$ forced due to $v_{i_3}$ not chosen in vertex cover we have edges left as
$\{(\tilde{x}_i,\tilde{v}'_i),(\tilde{a}'_i,\tilde{w}_i),(\tilde{c}_i,\tilde{b}'_i),(\tilde{b}_i,\tilde{w}_{i_7}),(\tilde{y}_i,\tilde{v}''_i)\}$
and for $\tilde{b}_i$ forced we have
$\{(\tilde{x}_i,\tilde{v}'_i),(\tilde{a}'_i,\tilde{w}_i),(\tilde{c}_i,\tilde{b}'_i),(\tilde{a}_i,\tilde{w}_{i_7}),(\tilde{y}_i,\tilde{v}''_i)\}$
If we have
$V'=\{\tilde{v}'_i,\tilde{v}''_i,\tilde{a}_i,\tilde{b}_i,\tilde{w}_i,\tilde{c}_i,\tilde{w}_{i_1},\tilde{w}_{i_2},\tilde{w}_{i_4},\tilde{w}_{i_6}\}$, then we can cover all edges in $E'$ with 10 vertices when at least one of $v_{i_1},v_{i_2},v_{i_3},v_{i_3}$ is not chosen in the vertex cover.
\qed
\end{proof}

\paragraph{The Overall Connection.} The fifth (and final) gadget used is denoted by $W_{C_i}$, and is shown in figure \ref{gadgetW_c}. We use this gadget for connecting broken cycles. It is easy to see that the gadget itself is an union of two hamiltonian paths, and that we need four vertices to cover all the edges. We have fixed the ordering of cycles $S_H,S_{H'}$, and we use this for the overall connection.

For cycle $C_i \in S_H$ let the vertex from which the specified path (path hamiltonian with respect to $C_i$ and inserted gadget to break cycle) starts be $s_i$ and where it ends be $e_i$,
similarly for cycles $C'_i \in S_{H'}$ let the start vertex be $s'_i$ and end vertex be $e'_i$.
For a cycle $C_i,C_{i+1} \in S_H$ and $C'_i,C'_{i+1} \in S_{H'}$ we insert the gadget $W_{C_i}$ and add edges $(e_i,t_{i_1}),(s_{i+1},t_{i_2}),(e'_i,t_{i_1}),(s'_{i+},t_{i_2})$,
for $ 1 \leq i < k $, where $k$ is the number of cycles in $H$. It is easy to see that after making these additions, the graph has edges which is exactly union of two hamiltonian paths. One of the hamiltonian path starting from $s_1$ and ending at $e_k$ and second hamiltonian path starting at $s'_1$ and ending at $e'_k$.

Putting together all the constructions described above and using the translations of the vertex cover at every stage, we have the following result. 

\begin{theorem}
The problem of finding a vertex cover of size at most $k$ in a braid graph is \NPC{}.
\end{theorem}
\begin{figure}
\centering
\scalebox{0.65} 
{
\begin{pspicture}(0,-2.2967188)(11.99,2.2967188)
\definecolor{color3}{rgb}{0.5490196078431373,0.39215686274509803,0.6274509803921569}
\psarc[linewidth=0.08,linecolor=color3](4.51,0.87828124){0.4}{315.0}{225.0}
\psarc[linewidth=0.08,linecolor=color3](5.55,0.87828124){0.4}{315.0}{225.0}
\psarc[linewidth=0.08,linecolor=color3](6.77,0.87828124){0.4}{315.0}{225.0}
\psarc[linewidth=0.08,linecolor=color3](8.99,0.85828125){0.4}{315.0}{225.0}
\psarc[linewidth=0.08,linecolor=color3](7.97,0.85828125){0.4}{315.0}{225.0}
\psarc[linewidth=0.08,linecolor=color3](11.51,0.87828124){0.4}{315.0}{225.0}
\psdots[dotsize=0.06](6.03,0.91828126)
\psdots[dotsize=0.06](6.17,0.91828126)
\psdots[dotsize=0.06](6.29,0.91828126)
\psdots[dotsize=0.06](9.93,0.8382813)
\psdots[dotsize=0.06](10.07,0.8382813)
\psdots[dotsize=0.06](10.19,0.8382813)
\psdots[dotsize=0.16](4.23,0.61828125)
\psdots[dotsize=0.16](4.77,0.61828125)
\psdots[dotsize=0.16](5.27,0.59828126)
\psdots[dotsize=0.16](5.81,0.59828126)
\psdots[dotsize=0.16](6.49,0.61828125)
\psdots[dotsize=0.16](7.03,0.59828126)
\psdots[dotsize=0.16](7.69,0.5782812)
\psdots[dotsize=0.16](8.25,0.5782812)
\psdots[dotsize=0.16](8.71,0.5782812)
\psdots[dotsize=0.16](9.27,0.5782812)
\psdots[dotsize=0.16](11.23,0.59828126)
\psdots[dotsize=0.16](11.81,0.59828126)
\psbezier[linewidth=0.04](4.062307,1.3782812)(4.261055,1.5582813)(3.9903634,1.5068387)(4.74245,1.4516646)(5.494537,1.3964906)(5.4948654,1.5485802)(5.5400105,1.5799851)(5.585155,1.6113902)(5.5400105,1.402541)(6.322522,1.4516646)(7.105034,1.5007882)(7.0435243,1.6782813)(7.17,1.3877541)
\psbezier[linewidth=0.04](7.591269,1.3582813)(7.8700247,1.5382812)(7.4903636,1.4868387)(8.545212,1.4316646)(9.6000595,1.3764906)(9.60052,1.5285802)(9.663838,1.5599852)(9.727158,1.5913903)(9.663838,1.3825411)(10.76136,1.4316646)(11.858881,1.4807881)(11.77261,1.6582812)(11.95,1.3677541)
\usefont{T1}{ptm}{m}{n}
\rput(5.362344,2.0982811){\small cycles in $\mathcal{C}_J$}
\usefont{T1}{ptm}{m}{n}
\rput(9.834531,1.9782813){\small Cycles in $\mathcal{C}_v$}
\psarc[linewidth=0.08,linecolor=color3](1.74,1.0882813){0.75}{333.92465}{282.38077}
\psarc[linewidth=0.08,linecolor=color3](4.53,-0.86328125){0.4}{142.77368}{40.436424}
\psarc[linewidth=0.08,linecolor=color3](5.57,-0.86328125){0.4}{139.56358}{45.22294}
\psarc[linewidth=0.08,linecolor=color3](6.79,-0.86328125){0.4}{146.10384}{47.123375}
\psarc[linewidth=0.08,linecolor=color3](9.01,-0.84328127){0.4}{123.46189}{47.461666}
\psarc[linewidth=0.08,linecolor=color3](7.99,-0.84328127){0.4}{131.7656}{50.6564}
\psarc[linewidth=0.08,linecolor=color3](11.51,-0.8232812){0.4}{144.04112}{39.261974}
\psdots[dotsize=0.16](11.21,-0.5232813)
\psdots[dotsize=0.16](11.79,-0.5232813)
\psdots[dotsize=0.06](6.05,-0.9032813)
\psdots[dotsize=0.06](6.19,-0.9032813)
\psdots[dotsize=0.06](6.31,-0.9032813)
\psdots[dotsize=0.06](9.97,-0.8232812)
\psdots[dotsize=0.06](10.11,-0.8232812)
\psdots[dotsize=0.06](10.23,-0.8232812)
\psdots[dotsize=0.16](4.25,-0.60328126)
\psdots[dotsize=0.16](4.79,-0.60328126)
\psdots[dotsize=0.16](5.29,-0.5832813)
\psdots[dotsize=0.16](5.83,-0.5832813)
\psdots[dotsize=0.16](6.51,-0.60328126)
\psdots[dotsize=0.16](7.05,-0.5832813)
\psdots[dotsize=0.16](7.71,-0.56328124)
\psdots[dotsize=0.16](8.27,-0.56328124)
\psdots[dotsize=0.16](8.73,-0.56328124)
\psdots[dotsize=0.16](9.29,-0.56328124)
\psbezier[linewidth=0.04](4.082307,-1.3632812)(4.281055,-1.5432812)(4.0103636,-1.4918387)(4.76245,-1.4366646)(5.514537,-1.3814906)(5.5148654,-1.5335802)(5.5600104,-1.5649852)(5.605155,-1.5963902)(5.5600104,-1.387541)(6.342522,-1.4366646)(7.125034,-1.4857881)(7.0635242,-1.6632812)(7.19,-1.3727541)
\psbezier[linewidth=0.04](7.611269,-1.3432813)(7.8900247,-1.5232812)(7.5103636,-1.4718386)(8.565211,-1.4166646)(9.62006,-1.3614906)(9.620521,-1.5135801)(9.683839,-1.5449852)(9.747157,-1.5763903)(9.683839,-1.3675411)(10.78136,-1.4166646)(11.8788805,-1.4657881)(11.79261,-1.6432812)(11.97,-1.3527541)
\usefont{T1}{ptm}{m}{n}
\rput(10.014531,-1.8214062){\small Cycles in $\mathcal{C}_u$}
\psarc[linewidth=0.08,linecolor=color3](1.76,-1.0732813){0.75}{77.7621}{25.423094}
\psdots[dotsize=0.16](1.89,0.35828125)
\psdots[dotsize=0.16](2.41,0.7782813)
\psdots[dotsize=0.16](1.91,-0.34171876)
\psdots[dotsize=0.16](2.43,-0.74171877)
\psframe[linewidth=0.04,dimen=outer](3.63,0.85828125)(2.79,0.37828124)
\psdots[dotsize=0.16](2.81,0.61828125)
\psdots[dotsize=0.16](3.59,0.61828125)
\psframe[linewidth=0.04,dimen=outer](5.33,0.23828125)(4.73,-0.20171875)
\psdots[dotsize=0.16](4.75,-0.00171875)
\psdots[dotsize=0.16](5.33,0.01828125)
\psarc[linewidth=0.04](4.9109135,0.2965194){0.36176184}{111.2505}{237.09476}
\psarc[linewidth=0.04](5.0892005,0.27908063){0.38079935}{315.0}{68.7495}
\psarc[linewidth=0.04](4.9012623,-0.25045654){0.3312622}{111.2505}{246.8014}
\psarc[linewidth=0.04](5.161262,-0.25045654){0.3312622}{293.1986}{68.7495}
\psframe[linewidth=0.04,dimen=outer](7.65,0.19828124)(7.05,-0.24171875)
\psdots[dotsize=0.16](7.07,-0.04171875)
\psdots[dotsize=0.16](7.65,-0.02171875)
\psarc[linewidth=0.04](7.2309136,0.2565194){0.36176184}{111.2505}{237.09476}
\psarc[linewidth=0.04](7.4392004,0.2690806){0.41079935}{315.0}{40.426216}
\psarc[linewidth=0.04](7.221262,-0.29045653){0.3312622}{111.2505}{246.8014}
\psarc[linewidth=0.04](7.481262,-0.29045653){0.3312622}{293.1986}{68.7495}
\psframe[linewidth=0.04,dimen=outer](6.39,0.17828125)(5.79,-0.26171875)
\psdots[dotsize=0.16](5.81,-0.06171875)
\psdots[dotsize=0.16](6.39,-0.04171875)
\psarc[linewidth=0.04](5.9709134,0.2365194){0.36176184}{111.2505}{237.09476}
\psarc[linewidth=0.04](6.1492004,0.21908061){0.38079935}{315.0}{359.82437}
\psarc[linewidth=0.04](5.961262,-0.31045654){0.3312622}{111.2505}{246.8014}
\psarc[linewidth=0.04](6.221262,-0.31045654){0.3312622}{23.612942}{68.7495}
\psdots[dotsize=0.06](6.59,-0.02171875)
\psdots[dotsize=0.06](6.73,-0.02171875)
\psdots[dotsize=0.06](6.85,-0.02171875)
\psframe[linewidth=0.04,dimen=outer](8.81,0.19828124)(8.21,-0.24171875)
\psdots[dotsize=0.16](8.23,-0.04171875)
\psdots[dotsize=0.16](8.81,-0.02171875)
\psarc[linewidth=0.04](8.390914,0.2565194){0.36176184}{111.2505}{237.09476}
\psarc[linewidth=0.04](8.5692005,0.23908062){0.38079935}{315.0}{68.7495}
\psarc[linewidth=0.04](8.381262,-0.29045653){0.3312622}{111.2505}{246.8014}
\psarc[linewidth=0.04](8.641262,-0.29045653){0.3312622}{293.1986}{68.7495}
\psframe[linewidth=0.04,dimen=outer](9.91,0.21828125)(9.31,-0.22171874)
\psdots[dotsize=0.16](9.33,-0.02171875)
\psdots[dotsize=0.16](9.91,-0.00171875)
\psarc[linewidth=0.04](9.490913,0.27651942){0.36176184}{111.2505}{237.09476}
\psarc[linewidth=0.04](9.669201,0.25908062){0.38079935}{315.0}{359.82437}
\psarc[linewidth=0.04](9.481262,-0.27045652){0.3312622}{111.2505}{246.8014}
\psarc[linewidth=0.04](9.741262,-0.27045652){0.3312622}{23.612942}{68.7495}
\psframe[linewidth=0.04,dimen=outer](11.25,0.23828125)(10.65,-0.20171875)
\psdots[dotsize=0.16](10.67,-0.00171875)
\psdots[dotsize=0.16](11.25,0.01828125)
\psarc[linewidth=0.04](10.830914,0.2965194){0.36176184}{204.3786}{237.09476}
\psarc[linewidth=0.04](11.009201,0.27908063){0.38079935}{315.0}{61.141262}
\psarc[linewidth=0.04](10.821262,-0.25045654){0.3312622}{111.2505}{159.62999}
\psarc[linewidth=0.04](11.06,-0.22919431){0.31}{294.6524}{68.7495}
\psdots[dotsize=0.06](10.15,-0.02171875)
\psdots[dotsize=0.06](10.29,-0.02171875)
\psdots[dotsize=0.06](10.41,-0.02171875)
\psline[linewidth=0.04cm](2.39,0.79828125)(2.79,0.6382812)
\psline[linewidth=0.04cm](3.61,0.6382812)(4.21,0.61828125)
\psline[linewidth=0.04cm](2.41,-0.7217187)(2.81,0.6382812)
\psline[linewidth=0.04cm](3.63,0.6382812)(4.27,-0.6017187)
\psbezier[linewidth=0.04](0.7527805,0.15870218)(0.51250976,0.2774425)(0.5735902,0.11828125)(0.6760773,0.5578743)(0.77856445,0.99746734)(0.5696045,0.99967045)(0.52805525,1.0265166)(0.48650596,1.0533626)(0.77186406,1.0241704)(0.73212516,1.4829552)(0.69238627,1.94174)(0.44632828,1.908075)(0.85,1.9782813)
\psbezier[linewidth=0.04](0.7727805,-1.8212978)(0.53250974,-1.7025574)(0.5935902,-1.8617188)(0.6960773,-1.4221257)(0.79856443,-0.9825327)(0.5896045,-0.9803296)(0.54805523,-0.95348346)(0.50650597,-0.92663735)(0.79186404,-0.95582956)(0.75212514,-0.49704474)(0.71238625,-0.038259905)(0.4663283,-0.07192503)(0.87,-0.00171875)
\usefont{T1}{ptm}{m}{n}
\rput(5.5346875,1.7882812){$0\leq j \leq \lfloor \frac{n}{2} \rfloor$}
\usefont{T1}{ptm}{m}{n}
\rput(5.4023438,-1.7017188){\small cycles in $\mathcal{C}_J'$}
\usefont{T1}{ptm}{m}{n}
\rput(5.5746875,-2.0717187){$0\leq j \leq \lfloor \frac{n}{2} \rfloor$}
\usefont{T1}{ptm}{m}{n}
\rput(0.30453125,1.0082812){$H$}
\usefont{T1}{ptm}{m}{n}
\rput(0.34453124,-0.99171877){$H'$}
\end{pspicture} 
}
\caption{Overall Connection, With thick-light lines showing the path in the cycle, rectangles representing copies gadget of $W_c$ and thin-dark line representing connection of various cycles to 
copy of gadget $W_c$.}
\label{connection}
\end{figure}
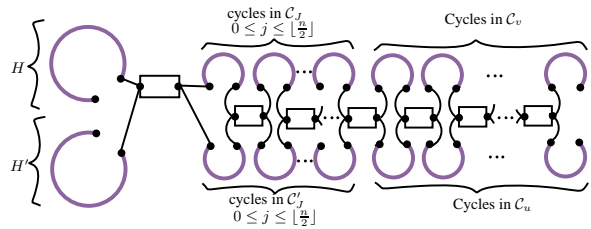

}


\section{An Improved Branching Algorithm}

In this section we describe an improved FPT algorithm for the vertex cover problem on graphs with maximum degree at most four. 
The algorithm is essentially a search tree, and the analysis is based on the branch-and-bound technique. We use standard notation with regards to branching vectors as described in~\cite{RN}.
The input to the algorithm is denoted by a pair $(G,k)$, where $G$ is a graph, and the question is whether $G$ admits a vertex cover of size at most $k$. 

We work with $k$, the size of the vertex cover sought, as the measure --- sometimes referred to as the \emph{budget}. When we say that we \emph{branch on a vertex~$v$}, we mean that we recursively generate two instances, one where $v$ belongs to the vertex cover, the other where $v$ does not belong to the vertex cover. 
This is a standard method of exhaustive branching, where the measure drops, respectively, by one and $d(v)$ in the two branches (since the neighbors of $v$ are forced to be in the vertex cover when $v$ does not belong to the vertex cover).  

\paragraph{Preprocessing.} 

We begin by eliminating simplicial vertices, that is, vertices whose neighborhoods form a clique. If the graph induced by $N[v]$ is a clique, then it is easy to see that there is a minimum vertex cover containing $N(v)$ and not containing $v$ (by a standard shifting argument). We therefore preprocess the graph in such a situation by deleting $N[v]$ and reducing the budget to $k - |N(v)|$. 


Our algorithm makes extensive use of the \emph{folding} technique, as described in past work~\cite{CKJ99,CKX06}. This allows us to preprocess vertices of degree two in polynomial time, while also reducing the size of the vertex cover sought by one. We briefly describe how we might handle degree-2 vertices in polynomial time.
Suppose $v$ is a degree-$2$ vertex in the graph $G$ with two neighbors $u$ and $w$ such that $u$ and $w$ are not adjacent to each other. We construct a new graph $G'$ as follows: 
remove the vertices $v$, $u$, and $w$ and introduce a new 
vertex $v^\star$ that is adjacent to all neighbors of the vertices $u$ and $w$ in $G$ (other than $v$). We say that the graph $G'$ is obtained from the graph $G$ by ``folding'' the vertex $v$, and we say 
that $v^\star$ is the vertex generated by folding $v$, or simply that $v^\star$ is the \emph{folded vertex} (when the context is clear). It turns out that the folding operation preserves equivalence, as shown below. 

\begin{proposition}\cite[Lemma 2.3]{CKJ99}
Let $G$ be a graph obtained by folding a degree-$2$ vertex $v$ in a graph $G$, where the two neighbors of $v$ are not adjacent to each other. Then the graph $G$ has a vertex cover of size bounded by $k$ if and only if the graph $G'$ has a vertex cover of size bounded by $(k - 1)$.
\end{proposition}

Note that the new vertex generated by the folding operation can have more than four neighbors, especially if the vertices adjacent to 
the degree two vertex have, for example, degree four to begin with. The branching algorithm that we will propose assumes that we will always find a 
vertex whose degree is bounded by $3$ to branch on, therefore it is important to avoid the situation where the graph obtained after folding all available degree two vertices is completely devoid of vertices of degree bounded by three 
(which is conceivable if all degree three vertices are adjacent to degree two vertices that in turn get affected by the folding operation). 
Therefore, we apply the folding operation somewhat tactfully$-$ we apply it only when we are sure that the folded vertex has degree at most four. We call such a vertex a \emph{foldable} vertex.
Further, a vertex is said to be \emph{easily foldable} if, after folding, it has degree at most $3$.
We avert the danger of leading ourselves to a four-regular graph recursively by explicitly ensuring that vertices of degree at most three are created whenever a folded vertex has degree four.
Note that in the preprocessing step we will be folding only easily foldable vertices.


Typically, we ensure a reasonable drop on all branches by creating the following win-win situation: if a vertex is foldable, then we fold it, if it is not, 
then this is the case since there are sufficiently many neighbors in the second neighborhood of the vertex, and in many situations, this would lead to a good branching vector. Also, during the course of the branching, we appeal to a couple of simple facts about the structure of a vertex cover, which we state below.


\begin{lemma}\cite[First part of Lemma 3.2]{CKJ99}
\label{lem:deg3neighborhood}
Let $v$ be a vertex of degree $3$ in a graph $G$. Then there is a minimum vertex cover
of $G$ that contains either all three neighbors of $v$ or at most one neighbor of $v$.
\end{lemma}

This follows from the fact that a vertex cover that contains $v$ (where $d(v)=3$) and two of its neighbors can be easily transformed into one, of the same size, that omits $v$ and contains all of its neighbors. 

\begin{proposition}
\label{prop:c4}
If $x,a,y,b$ form a cycle of length four in $G$ (in that order), and the degree of $a$ and $b$ in $G$ is two, then there exists an optimal vertex cover that does not pick $a$ or $b$ and contains both $x$ and $y$.
\end{proposition}

\longversion{
\begin{proof}
Let $S$ be an optimal vertex cover. Any vertex cover must pick at least two vertices among $x,y,a,b$. If $x$ and $y$ belong to $S$, 
then clearly $S$ does not contain $a$ and $b$ (otherwise $S \setminus \{a,b\}$ would continue to be a vertex cover, contradicting optimality). 
If $S$ does not contain $x$ (or $y$, or both), then $S$ must contain both $a$ and $b$. Note that $(S \setminus \{a,b\}) \cup \{x,y\}$ is a vertex cover 
whose size is at most $|S|$, and is a vertex cover of the desired size. 
\qed
\end{proof}
}

\paragraph{Overall Algorithm.} To begin with, the branching algorithm tries to branch mainly on a vertex of degree three or two. If the input graph is four-regular, 
then we simply branch on an arbitrary vertex to create two instances both of which have at least one vertex of degree at most three. We note that this is 
an off-branching step, in the future, the algorithm maintains the invariant that at each step, the smaller graph produced has at least one vertex whose degree is at most three. 

\longversion{

%
%

}

After this, we remove all the simplicial vertices and then fold all easily-foldable vertices. If a degree two vertex $v$ with neighbors $u$ and $w$ is not easily-foldable,
then note that there exists an optimal vertex cover that either contains $v$ or does not contain $v$ and includes both its neighbors. 
Indeed, if an optimal vertex cover $S$ contains, say $v$ and $u$, then note that $(S \setminus \{v\}) \cup \{w\}$ is a vertex cover of the same size. So we branch on the vertex $v$:

\begin{itemize}
\item when $v$ does not belong to the vertex cover, we pick $u,w$ in the vertex cover, leading to a drop of two in the measure,
\item when $v$ does belong to the vertex cover, we have that $N(u) \cup N(w)$ must belong to the vertex cover,
and we know that $|N(u) \cup N(w) \setminus \{v\}| \geq 4$ (otherwise, $v$ would be easily-foldable), and this leads to a drop of five in the measure. 
\end{itemize}

So we either preprocess degree two vertices in polynomial time, or branch on them with a branching vector of \branchvector{2,5}. At the leaves of this branching tree, if we have a sub-cubic graph,
then we employ the algorithm of~\cite{X10}. Otherwise, we have at least one degree three vertex which is adjacent to at least one degree four vertex. 
We branch on these vertices next. The case analysis is based on the neighborhood of the vertex --- broadly, we distinguish between when the neighborhood  
has at least one edge, and when it has no edges. The latter case is the most demanding in terms of a case analysis. For the rest of this section, we describe all the scenarios that arise in this context.

%
%
%

\paragraph{Degree three vertices with edges in their neighborhood.}

For this part of the algorithm, we can always assume that we are given a degree three vertex with a degree four neighbor. 
Let $v$ be a degree three vertex, and let~$N(v) := \{u,w,x\}$, where we let $u$ denote a degree four vertex.
Note that $u,w,x$ does not form a triangle, otherwise $v$ would be a simplicial vertex and we would have handled it earlier.
So, we deal with the case when $N(v)$ is not a triangle, but has at least one edge. If $(w,x)$ is an edge, then we branch on $u$:

\begin{itemize}
\item when $u$ does not belong to the vertex cover, we pick four of its neighbors in the vertex cover, leading to a drop of four in the measure,
\item when $u$ does belong to the vertex cover, we delete $u$ from the graph, and we are left with $v,w,x$ being a triangle where $v$ is a degree two vertex, and 
therefore we may pick $w,x$ in the vertex cover --- together, this leads to a drop of three in the measure. 
\end{itemize}


On the other hand, if $w,x$ is not an edge, then there is an edge incident to $u$. Suppose the edge is $u,w$ (the case when the edge is $u,x$ is symmetric). In this case, we branch on $x$ 
exactly as above. The measure may drop by three when $x$ does not belong to the vertex cover, if $x$ happens to be a degree three vertex. Therefore, our worst-case 
branching vector in the situation when $N(v)$ is not a triangle, but has at least one edge is \branchvector{3,3}.

\paragraph{Degree three vertices whose neighborhoods are independent.} 

Here we consider several cases. Broadly, we have two situations based on whether $u,w,x$ have any common neighbors or not. 


\longversion{
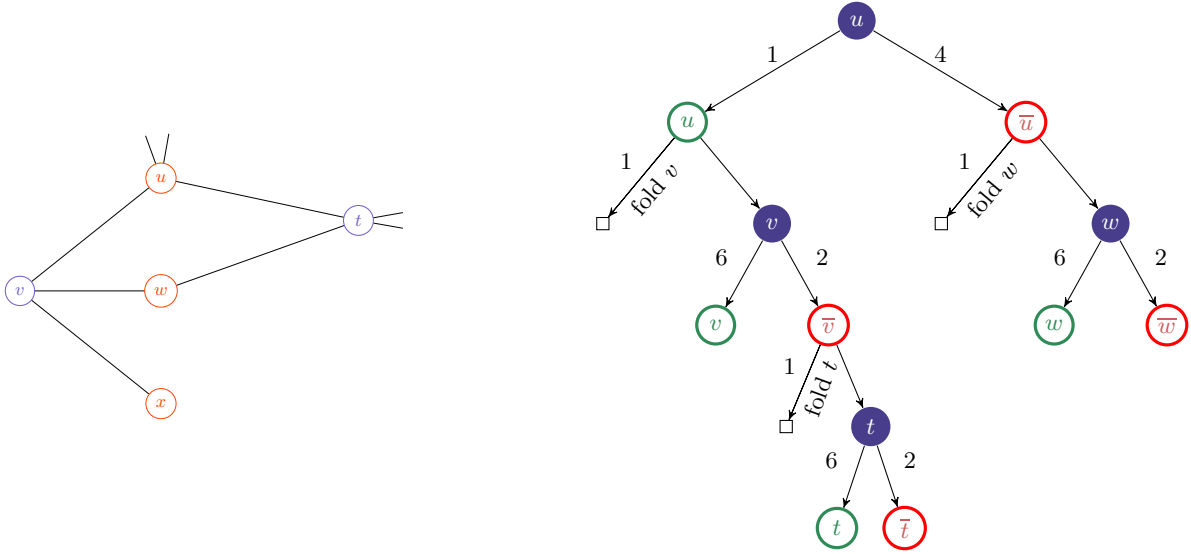
\begin{figure}[ht]
\centering
\begin{minipage}[c]{0.45\linewidth}
\label{fig:case1}
\scalebox{.75}{%
\begin{tikzpicture}
[outline/.style={color=SlateBlue,thin},
happy/.style={color=OrangeRed,thin},
general/.style={color=black,thin}]

\node[outline,shape=circle,draw] (central)  at (0,0) {$v$};
\node[happy,shape=circle,draw,xshift=1.5cm] (w) [right of=central] {$w$};
\node[happy,shape=circle,draw] (x) [below of=w,yshift=-1cm] {$x$};
\node[happy,shape=circle,draw] (u) [above of=w,yshift=1cm] {$u$};

\node[outline,shape=circle,draw] [right of=w,above of=w,xshift=2.5cm,yshift=.25cm] (t)  {$t$};

\draw (central) -- (u);
\draw (central) -- (w);
\draw (central) -- (x);

\draw (t) -- (u);
\draw (t) -- (w);

\draw[general] (t) -- +(10:.8cm);
\draw[general] (t) -- +(-10:.8cm);

\draw[general] (u) -- +(80:.8cm);
\draw[general] (u) -- +(110:.8cm);

\end{tikzpicture}
}
                                                                        
\end{minipage}
\quad
\begin{minipage}[c]{0.45\linewidth}
\label{fig:case1}
\centering
\begin{tikzpicture}[->,>=stealth',level/.style={sibling distance = 5cm/#1,
  level distance = 1.5cm},scale=.9]
  
\tikzset{
  treenode/.style = {align=center, inner sep=0pt, text centered,
    font=\sffamily},
  arn_n/.style = {treenode, circle,fill=DarkSlateBlue,text=white,text width=1.5em},
  arn_r/.style = {treenode, circle, IndianRed, draw=red, 
    text width=1.5em, very thick},
  arn_g/.style = {treenode, circle, SeaGreen, draw=SeaGreen, 
    text width=1.5em, very thick},
  arn_x/.style = {treenode, rectangle, draw=black,
    minimum width=0.5em, minimum height=0.5em}
} 
\node [arn_n] {$u$}
    child{ node [arn_g] {$u$} 
    child{ node [arn_x] {}
        edge from parent  [above=3pt] node[xshift=-.2cm] {$1$}
        edge from parent  [below=3pt] node[rotate=50] {fold $v$}
    }            
            child{ node [arn_n] {$v$}
							child{ node [arn_g] {$v$}
							edge from parent  [above=3pt] node[xshift=-.3cm] {$6$}}
							child{ node [arn_r] {$\overline{v}$}
							child{node [arn_x] {}
        edge from parent  [above=3pt] node[xshift=-.2cm] {$1$}
        edge from parent  [below=3pt] node[rotate=70] {fold $t$}
}
							child{
							node [arn_n] {$t$}
							child{
							node [arn_g] {$t$}
							edge from parent  [above=3pt] node[xshift=-.3cm] {$6$}
							}
							child{
							node [arn_r] {$\overline{t}$}
							edge from parent  [above=3pt] node[xshift=.3cm] {$2$}
							} 	
							}
							edge from parent  [above=3pt] node[xshift=.3cm] {$2$}}
            }                           
    edge from parent  [above=3pt] node {$1$}
    }
    child{ 
        node [arn_r] {$\overline{u}$} 
        child{ node [arn_x] {}        
        edge from parent  [above=3pt] node[xshift=-.2cm] {$1$}
        edge from parent  [below=3pt] node[rotate=50] {fold $w$}
        }
        child{ node [arn_n] {$w$} 
        child{
							node [arn_g] {$w$}
							edge from parent  [above=3pt] node[xshift=-.3cm] {$6$}
							}
							child{
							node [arn_r] {$\overline{w}$}
							edge from parent  [above=3pt] node[xshift=.3cm] {$2$}
							} 	
							}
        edge from parent  [above=3pt] node {$4$}
    }; 
\end{tikzpicture}

\end{minipage}
\caption{Scenario A, Case 1: The situation (left) and the suggested branching (right).}
\end{figure}

First, suppose there exists a vertex $t$ that is adjacent to at least two vertices in $N(v)$. Here, let us begin by considering the situation when $t$ is adjacent to $u$ and one other vertex. We will call this {\bf Scenario A}.

In this scenario, we distinguish two cases based on the degree of $t$, and whether $t$ is adjacent to a degree four vertex or not.
Here after for ease in specification we will refer to a degree $1$ vertex also as a foldable vertex. 

\begin{enumerate}[series=main,label=\bfseries Case~\arabic*:] \item {\bf The vertex $t$ has degree four.} Here, we branch on $u$ as follows. We let $(t,w)$ to be an edge in the graph.

\begin{enumerate}
\item If $u$ belongs to the vertex cover, then we delete $u$ from $G$. Then, if $v$ is foldable in the resulting graph, then we fold $v$. Otherwise, we branch further on $v$:
\begin{enumerate}
\item when $v$ does not belong to the vertex cover, we pick $u,w$ in the vertex cover, leading to a drop of two in the measure. Here, we delete $v,u$ and $w$, after which $t$ becomes a degree two vertex. Let $t^\prime,t^{\prime \prime}$ denote the two neighbors of $t$. Then, if $t$ is foldable in the resulting graph, then we fold $v$. Otherwise, we branch on $t$:
\begin{enumerate}
\item when $t$ does not belong to the vertex cover, we pick $t^\prime,t^{\prime\prime}$ in the vertex cover, leading to a drop of two in the measure.
\item when $t$ does belong to the vertex cover, we have that $N(t^\prime) \cup N(t^{\prime\prime})$ must belong to the vertex cover, and this leads to a drop of six in the measure.
\end{enumerate}
\item when $v$ does belong to the vertex cover, we have that $N(u) \cup N(w)$ must belong to the vertex cover, and we know that $|N(u) \cup N(w) \setminus \{v\}| \geq 5$ (otherwise, $v$ would be foldable), and this leads to a drop of six in the measure.
\end{enumerate}
\item If $u$ does not belong to the vertex cover, then we pick all of its neighbors in the vertex cover. Since the degree of $u$ is four, this leads the measure to drop by four. 
Also, after removing $N[u]$ from $G$, the vertex $w$ lose two neighbors (namely $v$ and $t$). If it is foldable, then we proceed by folding the said vertex. 
Otherwise, we branch further on $w$, letting $w^\prime,w^{\prime\prime}$ denote the neighbors of $w$.
\begin{enumerate}  
\item when $w$ does not belong to the vertex cover, we pick $w^\prime,w^{\prime\prime}$ in the vertex cover, leading to a drop of two in the measure,
\item when $w$ does belong to the vertex cover, we have that $N(w^\prime) \cup N(w^{\prime\prime})$ and this leads to a drop of six in the measure. 
\end{enumerate}
\end{enumerate}

Depending on the situations that arise, the branching vectors can be one of the following (we use $S$ to denote the vertex cover that will be output by the algorithm):

\begin{itemize}
\item $w$ is foldable in $G \setminus N[u]$, and $v$ is foldable in $G \setminus \{u\}$. \branchvector{2,5}
\item $w$ is foldable in $G \setminus N[u]$, $v$ is not foldable in $G \setminus \{u\}$, and $t$ is foldable in $G \setminus \{u\} \setminus N[v]$. \branchvector{7,4,5}
\item $w$ is foldable in $G \setminus N[u]$, $v$ is not foldable in $G \setminus \{u\}$, and $t$ not foldable in $G \setminus \{u\} \setminus N[v]$. \branchvector{7,9,5,5}
\item $w$ is not foldable in $G \setminus N[u]$, and $v$ is foldable in $G \setminus \{u\}$. \branchvector{2,10,6}
\item $w$ is not foldable in $G \setminus N[u]$, $v$ is not foldable in $G \setminus \{u\}$, and $t$ is foldable in $G \setminus \{u\} \setminus N[v]$. \branchvector{7,4,10,6}
\item $w$ is not foldable in $G \setminus N[u]$, $v$ is not foldable in $G \setminus \{u\}$, and $t$ is not foldable in $G \setminus \{u\} \setminus N[v]$. \branchvector{7,9,5,10,6}
\end{itemize}

\end{enumerate}

The reason we needed to have $d(t) = 4$ in the case above was to ensure that we have a vertex that we can either fold or branch on in the graph $G \setminus \{u\} \setminus N(v)$, which is the situation that arises
when $v$ is not foldable, and $N(v)$ is included in the vertex cover. If $w$ and $x$ both have degree three, then $v$ is indeed foldable and the branching above gives the desired guarantee. 
Otherwise, if $t$ has degree three and in particular $(t,x)$ is an edge, then $t$ becomes isolated in this situation, and we have no clear way of further progress. 
}

Before embarking on the case analysis, we describe a branching strategy for some specific situations --- these mostly involve two non-adjacent vertices that
have more than two neighbhors in common, with at least one of them of degree $4$. This will be useful in scenarios that arise later. 


\longversion{

\begin{figure}[ht]
\centering
\begin{minipage}[c]{0.3\linewidth}
\scalebox{.5}{%
\begin{tikzpicture}
[outline/.style={color=SlateBlue,thin},
happy/.style={color=OrangeRed,thin},
general/.style={color=black,thin}]

\node[outline,shape=circle,draw] (p)  at (0,0) {\Large $p$};
\node[outline,shape=circle,draw] (q) at (5,0) {\Large $q$};
\node[happy,shape=circle,draw] (a) at (2.5,2) {\Large $a$};
\node[happy,shape=circle,draw] (b) at (2.5,0.5) {\Large $b$};
\node[happy,shape=circle,draw] (c) at (2.5,-0.5) {\Large $c$};

\node[happy,shape=circle,draw] (x) at (-1,-2) {\Large $x$};
\node[happy,shape=circle,draw] (y) at (6,-2) {\Large $y$};

\draw (p) -- (a);
\draw (p) -- (b);
\draw (p) -- (c);
\draw (p) -- (x);

\draw (q) -- (a);
\draw (q) -- (b);
\draw (q) -- (c);
\draw (q) -- (y);

\end{tikzpicture}
}
\label{fig:case2a}
\end{minipage}
\caption{The cases involving at least two or three common neighbors.}
\end{figure}
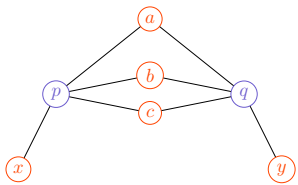

}


We consider the case when a degree four vertex $p$ non-adjacent to a vertex $q$ has at least three neighbhors in commmon, say $a, b, c$ and let $x$ be the other neighbhor of $p$ that may or may not be adjacent to $q$. 
Notice that there always exists an optimal vertex cover that either contains both $p$ and $q$ or omits both $p$ and $q$. To see this, consider an optimal vertex cover $S$ that contains 
$p$ and omits $q$. Then, $S$ clearly contains $a,b,c$. Notice now that $T := (S \setminus \{p\}) \cup \{x\}$ is also a vertex cover, and
$T$ contains neither $p$ or $q$, and has the same size as $S$. This suggests the following branching strategy:

\begin{enumerate}
\item If $p$ and $q$ both belong to the vertex cover, then the measure clearly drops by two. We proceed by deleting $p$ and $q$ from $G$. Now note that the degree of the vertices $\{a,b,c\}$ reduces by two, 
and they become vertices of degree one or two (note that they cannot be isolated because we always begin by eliminating vertices of degree two by preprocessing or branching). If any one of these vertices is simplicial or
foldable then we process it or fold it respectively. Otherwise, we branch on $a$:
\begin{enumerate}
\item when $a$ does not belong to the vertex cover, we pick its neighbhors in the vertex cover, leading to a drop of two in the measure.
\item when $a$ does belong to the vertex cover, we have that its second neighborhood must belong to the vertex cover, and this leads to a drop of six in the measure.
\end{enumerate}
\item If $p$ and $q$ are both omitted from the vertex cover, then we pick $a,b,c,x$ in the vertex cover and the measure drops by four.
%

%
\end{enumerate}

Note that if $a$ is foldable in $G \setminus \{p,q\}$, then we have the branch vector \branchvector{3,4}, otherwise, we have the branch vector \branchvector{4,8,4}.
We refer to the branching strategies outlined above as the {\bf CommonNeighborBranch} strategy.

\longversion{

We now continue our case analysis. Recall that we would like to address the situation that $t$ is degree three and all of its neighbors are common with $v$, and further that at least one of $w$ or $x$ have degree four. Let us say, without loss of generality, that $w$ has degree four.

\begin{enumerate}[resume=main,label=\bfseries Case~\arabic*:]
\item {\bf The vertex $t$ has degree three, the vertices $u,w$ have degree four, and $(t,x) \in E$.}

Here, we let $u^\prime$ and $u^{\prime\prime}$ denote the neighbors of $u$. Our case analysis is now based on the degrees of these vertices.


\begin{enumdescript}
\item[At least one of $u^\prime$ or $u^{\prime\prime}$ has degree three.]

Suppose, without loss of generality, that $u^\prime$ has degree three. Note that $x$ and $u$ are two non-adjacent degree four vertices. The vertices $v$ and $t$ are already in their 
common neighborhood. If they have more common neighbors, then we branch according to the {\bf CommonNeighborBranch} strategy. Otherwise, we branch on the vertex $w$ as follows. 

\begin{enumerate}
\item If $w$ belongs to the vertex cover, then we delete $w$ from $G$. Here, the measure drops by one. In the remaining graph, branch on $u$:
\begin{enumerate}
\item When $u$ does not belong to the vertex cover, we pick the neighbors of $u$ in the vertex cover. Since $u$ is not in the vertex cover, and $w$ is in the vertex cover, we know by 
Lemma~\ref{lem:deg3neighborhood} that the neighbors of $w$ and $x$ must be in the vertex cover. Note that $u$ and $x$ have no common neighbhors other than $v$ and $t$, otherwise the {\bf CommonNeighborBranch} strategy
would apply. Therefore, we have that the measure drops by at least six more (the vertex $u$ has at least four neighbors and $x$ has at least two private neighbors). 
\item When $u$ does belong to the vertex cover, then we also pick $x$ in the vertex cover (note that to cover the edge $(v,x)$, we may pick $x$ without loss of generality if both $u$ and $w$ are in
the vertex cover). Further, we fold $u^\prime$ if it is foldable, otherwise we branch on $u^\prime$:
\begin{enumerate}
\item When $u^\prime$ does not belong to the vertex cover, we pick its neighbhors in the vertex cover, leading to a drop of two in the measure.
\item When $u^\prime$ does belong to the vertex cover, we have that its second neighborhood must belong to the vertex cover, and this leads to a drop of six in the measure.
\end{enumerate} 
\end{enumerate}    
\item If $w$ does not belong to the vertex cover, then we pick all of its neighbors in the vertex cover. This immediately leads the measure to drop by four. Also, after removing $N[w]$ from $G$, the vertex $u$ loses
two neighbors (namely $v$ and $t$). If $u$ is foldable we fold $u$, otherwise we branch on $u$:
\begin{enumerate}  
\item When $u$ does not belong to the vertex cover, we pick $u^\prime,u^{\prime\prime}$ in the vertex cover, leading to a drop of two in the measure,
\item When $u$ does belong to the vertex cover, we have that $N(u^\prime) \cup N(u^{\prime\prime})$ and this leads to a drop of six in the measure. 
\end{enumerate}
\end{enumerate}

Depending on the situations that arise, the branching vectors can be one of the following (we use $S$ to denote the vertex cover that will be output by the algorithm):

\begin{itemize}
\item $u$ is foldable in $G \setminus N[w]$, and $u^\prime$ is foldable in $G \setminus \{u,w\}$. \branchvector{4,7,5}
\item $u$ is foldable in $G \setminus N[w]$, and $u^\prime$ is not foldable in $G \setminus \{u,w\}$. \branchvector{9,5,7,5}
\item $u$ is not foldable in $G \setminus N[w]$, and $u^\prime$ is foldable in $G \setminus \{u,w\}$. \branchvector{4,7,10,6}
\item $u$ is not foldable in $G \setminus N[w]$, and $u^\prime$ is not foldable in $G \setminus \{u,w\}$. \branchvector{9,5,7,10,6}
\end{itemize}

\begin{figure}[ht]
\centering
\begin{minipage}[c]{0.45\linewidth}
\scalebox{.75}{%
\begin{tikzpicture}
[outline/.style={color=SlateBlue,thin},
happy/.style={color=OrangeRed,thin},
general/.style={color=black,thin}]

\node[outline,shape=circle,draw] (central)  at (0,0) {$v$};
\node[happy,shape=circle,draw,xshift=1.5cm] (w) [right of=central] {$w$};
\node[happy,shape=circle,draw] (x) [below of=w,yshift=-1cm] {$x$};
\node[happy,shape=circle,draw] (u) [above of=w,yshift=1cm] {$u$};

\node[outline,shape=circle,draw] [right of=central,xshift=4.5cm] (t)  {$t$};

\draw (central) -- (u);
\draw (central) -- (w);
\draw (central) -- (x);

\draw (t) -- (x);
\draw (t) -- (u);
\draw (t) -- (w);

\draw[general] (w) -- +(80:.8cm);
\draw[general] (w) -- +(110:.8cm);

\draw[general] (x) -- +(80:.8cm);
\draw[general] (x) -- +(110:.8cm);

\node[outline,shape=circle,draw] [above right = of u,xshift=-.5cm,yshift=.9cm] (u1) {$u^\prime$};

\node[outline,shape=circle,draw] [above left = of u,xshift=.5cm,yshift=.9cm] (u2) {$u^{\prime\prime}$};

\draw[general] (u) -- (u1);
\draw[general] (u) -- (u2);

\draw[general] (u1) -- +(80:.8cm);
\draw[general] (u1) -- +(110:.8cm);

\end{tikzpicture}
}
\label{fig:case2a}
\end{minipage}
\quad
\begin{minipage}[c]{0.45\linewidth}
\centering
\begin{tikzpicture}[->,>=stealth',level/.style={sibling distance = 5cm/#1,
  level distance = 1.5cm},scale=.9]
  
\tikzset{
  treenode/.style = {align=center, inner sep=0pt, text centered,
    font=\sffamily},
  arn_n/.style = {treenode, circle,fill=DarkSlateBlue,text=white,text width=1.5em},
  arn_r/.style = {treenode, circle, IndianRed, draw=red, 
    text width=1.5em, very thick},
  arn_g/.style = {treenode, circle, SeaGreen, draw=SeaGreen, 
    text width=1.5em, very thick},
  arn_x/.style = {treenode, rectangle, draw=black,
    minimum width=0.5em, minimum height=0.5em}
} 
\node [arn_n] {$w$}
	child{ node [arn_g] {$w$}
	    child{ node [arn_n] {$u$}
			child{ node [arn_g] {$u$}
				child{node [arn_x] {}
			        edge from parent  [above=3pt] node[xshift=-.2cm] {$1$}
			        edge from parent  [below=3pt] node[rotate=70] {fold $u^\prime$}
				}
				child{
					node [arn_n] {$u^\prime$}
					child{
						node [arn_g] {$u^\prime$}
						edge from parent  [above=3pt] node[xshift=-.3cm] {$6$}
					}
					child{
						node [arn_r] {$\overline{u^\prime}$}
						edge from parent  [above=3pt] node[xshift=.3cm] {$2$}
					} 	
				}
				edge from parent  [above=3pt] node[xshift=-.5cm] {$1~(+1)$}
			}
			child{ 
				node [arn_r] {$\overline{u}$}
				edge from parent  [above=3pt] node[xshift=.5cm] {$4~(+2)$}
			}			
	    }
	    edge from parent  [above=3pt] node[xshift=-.2cm] {$1$}
    }
	child{ node [arn_r] {$\overline{w}$} 
		child{node [arn_x] {}
	        edge from parent  [above=3pt] node[xshift=-.2cm] {$1$}
	        edge from parent  [below=3pt] node[rotate=50] {fold $u$}
		}
		child{
					node [arn_n] {$u$}
					child{
						node [arn_g] {$u$}
						edge from parent  [above=3pt] node[xshift=-.3cm] {$6$}
					}
					child{
						node [arn_r] {$\overline{u}$}
						edge from parent  [above=3pt] node[xshift=.3cm] {$2$}
					} 	
				}
		edge from parent [above=3pt] node[xshift=.2cm] {$4$}
	};
\end{tikzpicture}
\label{fig:case2a_branch}
\end{minipage}
\caption{Scenario A, Case 2.I: The situation (left) and the suggested branching (right).}
\end{figure}

\item[Both $u^\prime$ or $u^{\prime\prime}$ have degree four.]

Here, we branch on $u^\prime$, as described below.

\begin{enumerate}
\item If $u^\prime$ belongs to the vertex cover, then we delete $u^\prime$ from $G$. Here, the measure drops by one. In the remaining graph, branch on $u$:
\begin{enumerate}
\item When $u$ does not belong to the vertex cover, we pick the neighbors of $u$ in the vertex cover. Also, after removing $N[u]$ from $G$, the vertices $w$ and $x$ loose two neighbors each (namely $v$ and $t$). 
If either of them are foldable, then we proceed by folding. Notice that none of them become isolated because degree two vertices are eliminated. Otherwise, we branch on $w$:
\begin{enumerate}
\item When $w$ does not belong to the vertex cover, we pick its neighbhors in the vertex cover, leading to a drop of two in the measure.
\item When $w$ does belong to the vertex cover, we have that its second neighborhood must belong to the vertex cover, and this leads to a drop of six in the measure.
\end{enumerate}
\item When $u$ does belong to the vertex cover, remove $u$ from $G$. In the remaining graph, the vertices $v$ and $t$ loses one neighbor each (namely $u$), and are now vertices of degree two. 
Note that $v,w,t,x$ now form a $C_4$, and since $v$ and $t$ have degree two, we may pick $w$ and $x$ in the vertex cover without loss of generality (see Proposition~\ref{prop:c4}).  
\end{enumerate}    
\item If $u^\prime$ does not belong to the vertex cover, then we pick all of its neighbors in the vertex cover. This immediately leads the measure to drop by four. Also, after 
removing $N[u^\prime]$ from $G$, the vertices $v$ and $t$ loses one neighbor each (namely $u$), and are now vertices of degree two. Note that $v,w,t,x$ now form a $C_4$, and since $v$ and $t$ have degree two, 
we may pick $w$ and $x$ in the vertex cover without loss of generality (see Proposition~\ref{prop:c4}).  
\end{enumerate}

If one of $w$ or $x$ is foldable in $G \setminus \{u^\prime\} \setminus N[v]$, then we have a branching vector of \branchvector{4,5,6}. Otherwise, we have a branching vector of \branchvector{4,10,6,6}.

\begin{figure}[ht]
\centering
\begin{minipage}[c]{0.45\linewidth}
\scalebox{.75}{%
\begin{tikzpicture}
[outline/.style={color=SlateBlue,thin},
happy/.style={color=OrangeRed,thin},
general/.style={color=black,thin}]

\node[outline,shape=circle,draw] (central)  at (0,0) {$v$};
\node[happy,shape=circle,draw,xshift=1.5cm] (w) [right of=central] {$w$};
\node[happy,shape=circle,draw] (x) [below of=w,yshift=-1cm] {$x$};
\node[happy,shape=circle,draw] (u) [above of=w,yshift=1cm] {$u$};

\node[outline,shape=circle,draw] [right of=central,xshift=4.5cm] (t)  {$t$};

\draw (central) -- (u);
\draw (central) -- (w);
\draw (central) -- (x);

\draw (t) -- (x);
\draw (t) -- (u);
\draw (t) -- (w);

\draw[general] (w) -- +(80:.8cm);
\draw[general] (w) -- +(110:.8cm);

\draw[general] (x) -- +(80:.8cm);
\draw[general] (x) -- +(110:.8cm);

\node[outline,shape=circle,draw] [above right = of u,xshift=-.5cm,yshift=.9cm] (u1) {$u^\prime$};

\node[outline,shape=circle,draw] [above left = of u,xshift=.5cm,yshift=.9cm] (u2) {$u^{\prime\prime}$};

\draw[general] (u) -- (u1);
\draw[general] (u) -- (u2);

\draw[general] (u1) -- +(90:.8cm);
\draw[general] (u1) -- +(70:.8cm);
\draw[general] (u1) -- +(110:.8cm);

\draw[general] (u2) -- +(90:.8cm);
\draw[general] (u2) -- +(70:.8cm);
\draw[general] (u2) -- +(110:.8cm);

\end{tikzpicture}
}
\label{fig:case2a}
\end{minipage}
\quad
\begin{minipage}[c]{0.45\linewidth}
\centering
\begin{tikzpicture}[->,>=stealth',level/.style={sibling distance = 5cm/#1,
  level distance = 1.5cm},scale=.9]
  
\tikzset{
  treenode/.style = {align=center, inner sep=0pt, text centered,
    font=\sffamily},
  arn_n/.style = {treenode, circle,fill=DarkSlateBlue,text=white,text width=1.5em},
  arn_r/.style = {treenode, circle, IndianRed, draw=red, 
    text width=1.5em, very thick},
  arn_g/.style = {treenode, circle, SeaGreen, draw=SeaGreen, 
    text width=1.5em, very thick},
  arn_x/.style = {treenode, rectangle, draw=black,
    minimum width=0.5em, minimum height=0.5em}
} 
\node [arn_n] {$u^\prime$}
	child{ node [arn_g] {$u^\prime$}
	    child{ node [arn_n] {$u$}
			child{ node [arn_g] {$u$}
				edge from parent  [above=3pt] node[xshift=-.5cm] {$1~(+2)$}
			}
			child{ 
				node [arn_r] {$\overline{u}$}
				child{node [arn_x,yshift=-0.5cm,xshift=-1cm] {}
			        edge from parent  [above=3pt] node[xshift=-.2cm] {$1$}
			        edge from parent  [below=3pt] node[rotate=50] {fold $w$ or $x$}
				}
				child{
					node [arn_n] {$w$}
					child{
						node [arn_g] {$w$}
						edge from parent  [above=3pt] node[xshift=-.3cm] {$6$}
					}
					child{
						node [arn_r] {$\overline{w}$}
						edge from parent  [above=3pt] node[xshift=.3cm] {$2$}
					} 	
				}
				edge from parent  [above=3pt] node[xshift=.5cm] {$4$}
			}			
	    }
	    edge from parent  [above=3pt] node[xshift=-.2cm] {$1$}
    }
	child{ node [arn_r] {$\overline{u^\prime}$} 
		edge from parent [above=3pt] node[xshift=.2cm] {$6$}
	};
\end{tikzpicture}
\label{fig:case2a_branch}
\end{minipage}
\caption{Scenario A, Case 2.II: The situation (left) and the suggested branching (right).}
\end{figure}

\end{enumdescript} 

\end{enumerate}


This completes the description of {\bf Scenario A}, where we assumed that $t$ was adjacent to $u$. Now, let us turn to the situation when $t$  is not adjacent to $u$. Since we are in the setting where $t$ is adjacent to two neighbors of $v$, this implies that $w$ and $x$ are both neighbors of $t$. In fact, we can even assume that both $w$ and $x$ are vertices of degree three, otherwise we would be in {\bf Scenario A} by a simple renaming of vertices. We call this setup {\bf Scenario B}.

Here, we simply branch on the vertex $u$, as follows:

\begin{enumerate}
\item If $u$ belongs to the vertex cover, then we delete $u$ from $G$. In the resulting graph, $v$ is evidently a foldable degree two vertex, so we fold $v$. Notice that the measure altogether drops by two in this branch.
\item If $u$ does not belong to the vertex cover, then we pick all its neighbors in the vertex cover. After removing $N[u]$, note that $w$ and $x$ have degree two. If either of them are foldable, then we proceed by folding. Otherwise, we branch on $w$:
\begin{enumerate}
\item When $w$ does not belong to the vertex cover, we pick its neighbhors in the vertex cover, leading to a drop of two in the measure.
\item When $w$ does belong to the vertex cover, we have that its second neighborhood must belong to the vertex cover, and this leads to a drop of six in the measure.
\end{enumerate}
\end{enumerate}

Note that if $w$ is foldable in $G \setminus N[u]$, then we have a branch vector of \branchvector{2,5}, otherwise, we have a branch vector of \branchvector{2,6,10}. Now we have covered all the cases that arise when the neighbors of $v$ have a shared neighbor other than $v$, which we called $t$.

\begin{figure}[ht]
\centering
\begin{minipage}[c]{0.45\linewidth}
\scalebox{.75}{%
\begin{tikzpicture}
[outline/.style={color=SlateBlue,thin},
happy/.style={color=OrangeRed,thin},
general/.style={color=black,thin}]

\node[outline,shape=circle,draw] (central)  at (0,0) {$v$};
\node[happy,shape=circle,draw,xshift=1.5cm] (w) [right of=central] {$w$};
\node[happy,shape=circle,draw] (x) [below of=w,yshift=-1cm] {$x$};
\node[happy,shape=circle,draw] (u) [above of=w,yshift=1cm] {$u$};

\node[outline,shape=circle,draw] [right of=central,xshift=4.5cm] (t)  {$t$};

\draw (central) -- (u);
\draw (central) -- (w);
\draw (central) -- (x);

\draw (t) -- (x);
\draw (t) -- (w);

\draw[general] (w) -- +(60:.8cm);
\draw[general] (x) -- +(60:.8cm);
\draw[general,dashed] (t) -- (u);

\draw[general] (u) -- +(90:.8cm);
\draw[general] (u) -- +(70:.8cm);
\draw[general] (u) -- +(110:.8cm);

\end{tikzpicture}
}
\label{fig:case2a}
\end{minipage}
\quad
\begin{minipage}[c]{0.45\linewidth}
\centering
\begin{tikzpicture}[->,>=stealth',level/.style={sibling distance = 5cm/#1,
  level distance = 1.5cm},scale=.9]
  
\tikzset{
  treenode/.style = {align=center, inner sep=0pt, text centered,
    font=\sffamily},
  arn_n/.style = {treenode, circle,fill=DarkSlateBlue,text=white,text width=1.5em},
  arn_r/.style = {treenode, circle, IndianRed, draw=red, 
    text width=1.5em, very thick},
  arn_g/.style = {treenode, circle, SeaGreen, draw=SeaGreen, 
    text width=1.5em, very thick},
  arn_x/.style = {treenode, rectangle, draw=black,
    minimum width=0.5em, minimum height=0.5em}
} 
\node [arn_n] {$u$}
	child{ node [arn_g] {$u$}
	    edge from parent  [above=3pt] node[xshift=-.2cm] {$1~(+1)$}
    }
	child{ node [arn_r] {$\overline{u}$}
	    child{node [arn_x,yshift=-0.5cm,xshift=-1cm] {}
	        edge from parent  [above=3pt] node[xshift=-.2cm] {$1$}
	        edge from parent  [below=3pt] node[rotate=42] {fold $w$ or $x$}
		}
		child{
			node [arn_n] {$w$}
			child{
				node [arn_g] {$w$}
				edge from parent  [above=3pt] node[xshift=-.3cm] {$6$}
			}
			child{
				node [arn_r] {$\overline{w}$}
				edge from parent  [above=3pt] node[xshift=.3cm] {$2$}
			} 	
		} 
		edge from parent [above=3pt] node[xshift=.2cm] {$4$}
	};
\end{tikzpicture}
\label{fig:case2a_branch}
\end{minipage}
\caption{Scenario B: The situation (left) and the suggested branching (right).}
\end{figure}
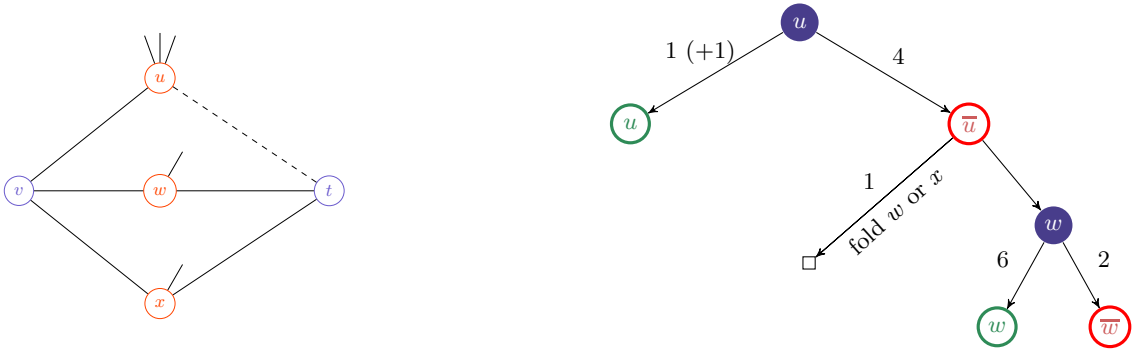

The remaining case is when the vertices $u,w,x$ have no common neighbors other than $v$. We call this {\bf Scenario C} .
Here, we have cases depending on the degree of $w$ and $x$ --- the first setting is when both $w,x$ are vertex of degree three. Second when both have degree four and third when one has degree three and other has degree four. 


\begin{enumerate}[series=scenarioC,label=\bfseries Case~\arabic*:]
\item {\bf When the degree of both $w$ and $x$ is three.}
This branching is identical to the branching for {\bf Sceneario B}. Note that the important aspect there was the fact that $v$ is foldable 
in $G \setminus \{u\}$, which continues to be the case here. It is easy to check that all other elements are identical. 

\item {\bf When the degree of both $w$ and $x$ is four.}
In this case, we branch on $w$. 
\begin{enumerate}
\item If $w$ belongs to the vertex cover, then we delete $w$ from $G$. Here, the measure drops by one. In the remaining graph, branch on $v$:
\begin{enumerate}
\item When $v$ does not belong to the vertex cover, we pick the neighbors of $v$ in the vertex cover and the measure drops by three. 
\item When $v$ does belong to the vertex cover and we are in case when $w$ is in vertex cover, we know by Lemma~\ref{lem:deg3neighborhood} that the neighbors of $u$ and $x$ must be in the vertex cover. 
 Note that $u, w$ and $x$ have no common neighbors other than $v$ (or we would be in {\bf Scenario A} or {\bf Scenario B}). 
But, both $u$ and $x$ are degree four vertex with no vertex common in their neighborhood other than $v$, so we include in vertex cover neighbhors of $u$ and $x$ and delete from graph $N[u]\cup N[x]\cup \{w\}$,
with a total drop in the measure by 8.
\end{enumerate}    
\item If $w$ does not belong to the vertex cover, then we pick all of its neighbors in the vertex cover and we branch on $x$.
\begin{enumerate}
 \item When $x$ does not belong to the vertex cover, we include all neighbhors of $x$ in the vertex cover to get a total drop of $7$.
 \item When $x$ does belong to the vertex cover, and we have that $w$ is not in the vertex cover, we know by Lemma~\ref{lem:deg3neighborhood} that neighbors of $u$ and $w$ must be in vertex cover. So we include
 neighbors of $u$ and $x$ in the vertex cover, to get a total drop of $8$.
\end{enumerate}

Here we have the branch vector as $(3,8,7,8)$.
\end{enumerate}

\end{enumerate}

\begin{enumerate}[resume=scenarioC,label=\bfseries Case~\arabic*:]
\item {\bf When the degree of $w$ is four and $x$ is three.}
We let the two other neighbors of $x$ to be $x_1,x_2$. If $x_1$ is a degree $4$ vertex and $x_2$ is a degree $3$ vertex (or vice-versa) then we have a degree $3$ vertex $x$ adjacent 
to two degree $3$ vertex $v,x_2$ and a degree $4$ vertex
$x_1$ and we can apply the rules in {\bf Scenario $C$: case $1$}. So we are left with two cases one when both $x_1,x_2$ are degree $3$ vertex and second when both $x_1,x_2$ are degree $4$ vertex.
\begin{enumerate}
 \item Both $x_1,x_2$ are degree three vertex. In this case we branch on $u$.
 \begin{enumerate}
  \item When $u$ does belong to the vertex cover then we branch on $v$.
  \begin{enumerate}
   \item When $v$ does not belong to the vertex cover then we include neighbhors of $v$ in the vertex cover, to get a drop in the measure by $3$.
   \item When $v$ does belong to the vertex cover and we know $u$ is in vertex cover, we know by Lemma~\ref{lem:deg3neighborhood} that neighbors of $w,x$ are in vertex cover. So we include neighbors of $w,x$ in vertex cover and get
   a drop in measure by $7$.
  \end{enumerate}
  \item When $u$ does not belong to the vertex cover. Then we include neighbors of $u$ in the vertex cover and delete $N[u]$ from the graph. Now $x$ is a degree $2$ vertex and $\lvert N(x_1)\cup N(x_2) \setminus \{v\} \rvert \leq 4$,
  so we fold $x$ to get a drop of one more in the measure.

  Here we have the branch vector as $(3,7,5)$.
 \end{enumerate}
\item Both $x_1,x_2$ are degree four vertex. We branch on $x_1$.

\begin{enumerate}
 \item When $x_1$ does not belong to the vertex cover, we get an immediate drop of four in the measure. We delete $N[x_1]$ from the graph, after deletion $v$ is a degree $2$ vertex. If $v$ is foldable we fold $v$ and get a drop of $1$,
 otherwise we branch on $v$.
 \begin{enumerate}
  \item When $v$ does belong to the vertex cover then we include neighborhood of $u$ and $w$ in the vertex cover and get a total drop in the measure of atleast $10$.
  \item When $v$ does not belong in the vertex cover then we include neighbhors of $v$ in the vertex cover and get a total drop of $6$.
 \end{enumerate}
\item When $x_1$ does belong to the vertex cover, we branch on $x$.
\begin{enumerate}
 \item When $x$ does belong to the vertex cover, we know by Lemma~\ref{lem:deg3neighborhood} that neighbhors of $v$ and $x_2$ are in vertex cover. So we inlcude neighbhors of $v$ and $x_2$ in the vertex cover and get 
 a total drop of $7$ in the measure.
 \item When $x$ does not belong to the vertex cover, we get an immediate drop in measure by $3$, now we branch on $u$.
 \begin{itemize}
  \item When $u$ does belong to the vertex cover and we are in the case when $x$ not belong to the vertex cover, we know by Lemma~\ref{lem:deg3neighborhood} that neighbhors of $w,x$ must be in vertex cover.
  So we include neighbors of $w,x$ in the vertex cover and get a total drop in the measure by $7$
  \item When $u$ does not belong to the vertex cover then we include neighbhors of $u$ in the vertex cover and get a total drop of $6$ in the measure.
 \end{itemize}
\end{enumerate}
\end{enumerate}

If $v$ is foldable in $G \setminus N[x_1]$ we have the branch vector $(5,7,7,6)$, otherwise the branch vector is $(10,6,7,7,6)$.
 \end{enumerate}

\end{enumerate}

Note that the correctness of the algorithm is implicit in the description, and follows from the fact that the cases are exhaustive and so is the branching. The branch vectors are summarized in Figure~\ref{tab:runningtime}. We have, consequently, the following theorem.

\begin{theorem}
The \name{Vertex Cover} problem on graphs that have maximum degree at most four can be solved in $\OO(\runtime^k \cdot nm)$ worst-case running time.
\end{theorem}

\begin{figure}[ht]
\centering
\begin{minipage}[c]{0.45\linewidth}
\begin{tabular}{ |c|l|l|c| }
\hline
Scenario & Cases & Branch Vector & $c$ \\
\hline
\multirow{12}{*}{Scenario A.} & \multirow{6}{*}{Case 1} & \branchvector{2,5} & $1.2365$\\
							 & 						   & \branchvector{7,4,5} & $1.2365$\\
							 & 						   &  \branchvector{7,9,5,5} & $1.2498$\\
							 & 						   &  \branchvector{2,10,6} & $1.2530$\\
							 & 						   &  \branchvector{7,4,10,6} & $1.2475$\\
							 & 						   &  \branchvector{7,9,5,10,6} & $1.2575$\\
							 \cline{2-4}
							 & \multirow{4}{*}{Case 2 (I)} & \branchvector{4,7,5} & $1.2365$\\
							 & 						   &  \branchvector{9,5,7,5} & $1.2498$\\
							 & 						   &  \branchvector{4,7,10,6} & $1.2475$\\
							 & 						   &  \branchvector{9,5,7,10,6} & $1.2575$\\
							 \cline{2-4}
							 & \multirow{2}{*}{Case 2 (II)} & \branchvector{4,5,6} & $1.2498$\\
							 & 						   &  \branchvector{4,10,6,6} & $1.2590$\\
\hline
\multirow{2}{*}{Scenario B.} & \multirow{2}{*}{} & \branchvector{2,5} & $1.2365$\\
							 & 						   & \branchvector{2,6,10} & $1.2530$\\
\hline
\end{tabular}
\end{minipage}
\quad
\begin{minipage}[c]{0.45\linewidth}
\begin{tabular}{ |c|l|l|c| }
\hline
Scenario & Cases & Branch Vector & $c$ \\
\hline
Degree Two &  & \branchvector{2,6} & $1.2365$ \\
Edge in $N(v)$ & & \branchvector{3,3} & $1.2599$ \\
\hline
\multirow{3}{*}{CNB} &  & \branchvector{2,5} & $1.2365$ \\
                                             &  & \branchvector{3,4} & $1.2207$ \\
                                             &  & \branchvector{4,8,4} & $1.2465$ \\
\hline
\multirow{5}{*}{Scenario C.} & \multirow{1}{*}{Case 1} & \branchvector{2,10,6} & $1.2530$\\
							 \cline{2-4}
							 & \multirow{1}{*}{Case 2} & \branchvector{8,3,8,7} & $1.2631$\\
							 \cline{2-4}
							 & \multirow{3}{*}{Case 3} & \branchvector{3,7,5} & $1.2637$\\
							 & 						   &  \branchvector{5,7,7,6} & $1.2519$\\
							 & 						   &  \branchvector{10,6,7,7,6} & $1.2592$\\
							 \hline
\end{tabular}
\end{minipage}
\caption{The branch vectors and the corresponding running times across various scenarios and cases.}
\label{tab:runningtime}
\end{figure}

}

\shortversion{
A broad overview of all the other cases is as follows.

\begin{enumerate}
\item {\bf Scenario A.} There exists a vertex $t$ that is adjacent to at least two vertices in $N(v)$. Further, $t$ is adjacent to $u$ and one other vertex.
\begin{itemize}
\item The vertex $t$ has degree four.
\item The vertex $t$ has degree three, $u,w,x$ have degree four, and $(t,x) \notin E$. We let $u^\prime$ and $u^{\prime \prime}$ denote the neighbors of $u$ other than $t$ and $v$. 
\begin{itemize}
\item The degree of both $u^\prime$ and $u^{\prime \prime}$ is four. 
\item At least one of $u^\prime$ and $u^{\prime \prime}$ has degree three. 
\end{itemize}
\end{itemize}

\item {\bf Scenario B.} There exists a vertex $t$ that is adjacent to at least two vertices in $N(v)$. The vertex $t$ is not adjacent to $u$ and is therefore adjacent to $w$ and $x$. 
\item {\bf Scenario C.} The vertices $u,v,w$ have no common neighbors other than $v$. We have the following cases based on degree of $w,x$.
\begin{itemize}
\item The degree of both $w$ and $x$ is three.
\item The degree of both $w$ and $x$ is four.
\item The degree of $w$ is four and $x$ is three.
\end{itemize}

\end{enumerate}

\begin{theorem}
There is an algorithm that determines if a graph with maximum degree at most four has a vertex cover of size at most $k$ in $\OO^*(\runtime^k)$ worst-case running time.
\end{theorem}

\begin{figure}[ht]
\centering
\begin{minipage}[c]{0.45\linewidth}
\scalebox{.8}{%
\begin{tabular}{ |c|l|l|c| }
\hline
Scenario & Cases & Branch Vector & $c$ \\
\hline
\multirow{12}{*}{Scenario A} & \multirow{6}{*}{Case 1} & \branchvector{2,5} & $1.2365$\\
							 & 						   & \branchvector{7,4,5} & $1.2365$\\
							 & 						   &  \branchvector{7,9,5,5} & $1.2498$\\
							 & 						   &  \branchvector{2,10,6} & $1.2530$\\
							 & 						   &  \branchvector{7,4,10,6} & $1.2475$\\
							 & 						   &  \branchvector{7,9,5,10,6} & $1.2575$\\
							 \cline{2-4}
							 & \multirow{4}{*}{Case 2 (I)} & \branchvector{4,7,5} & $1.2365$\\
							 & 						   &  \branchvector{9,5,7,5} & $1.2498$\\
							 & 						   &  \branchvector{4,7,10,6} & $1.2475$\\
							 & 						   &  \branchvector{9,5,7,10,6} & $1.2575$\\
							 \cline{2-4}
							 & \multirow{2}{*}{Case 2 (II)} & \branchvector{4,5,6} & $1.2498$\\
							 & 						   &  \branchvector{4,10,6,6} & $1.2590$\\
\hline
\end{tabular}
}
\end{minipage}
\quad
\begin{minipage}[c]{0.45\linewidth}
\scalebox{.8}{%
\begin{tabular}{ |c|l|l|c| }
\hline
Scenario & Cases & Branch Vector & $c$ \\
\hline
\multirow{3}{*}{CNB} &  & \branchvector{2,5} & $1.2365$ \\
                                             &  & \branchvector{3,4} & $1.2207$ \\
                                             &  & \branchvector{4,8,4} & $1.2465$ \\
\hline
Degree Two &  & \branchvector{2,6} & $1.2365$ \\
Edge in $N(v)$ & & \branchvector{3,3} & $1.2599$ \\
\hline
\multirow{2}{*}{Scenario B} & \multirow{2}{*}{} & \branchvector{2,5} & $1.2365$\\
							 & 						   & \branchvector{2,6,10} & $1.2530$\\
\hline
\multirow{5}{*}{Scenario C} & \multirow{1}{*}{Case 1} & \branchvector{2,10,6} & $1.2530$\\
							 \cline{2-4}
							 & \multirow{1}{*}{Case 2} & \branchvector{8,3,8,7} & $1.2631$\\
							 \cline{2-4}
							 & \multirow{3}{*}{Case 3} & \branchvector{7,3,5} & $1.2637$\\
							 & 						   &  \branchvector{5,7,7,6} & $1.2519$\\
							 & 						   &  \branchvector{10,6,7,7,6} & $1.2592$\\
							 \hline
\end{tabular}
}
\end{minipage}
\caption{The branch vectors and the corresponding running times across various scenarios and cases. (This table is a truncated version due to lack of space.)}
\label{tab:runningtime}
\end{figure}

}


\section{Conclusions}

In this work we showed that the problem of hitting all axis-parallel slabs induced by a point set $P$ is equivalent to the problem of finding a vertex cover on a graph whose edge set is the union of two Hamiltonian Paths. We established that this problem is \NPC{}. Finally, we also gave an algorithm for Vertex Cover on graphs of maximum degree four whose running time is $O^\star(\runtime{}^k)$.  It would be interesting to know if there are better algorithms for braid graphs in particular.


\printbibliography

\end{document}